\documentclass[11pt]{article}
\usepackage[dvips, paper=letterpaper, top=1in, bottom=1in, left=1in, right=1in]{geometry}
\usepackage{hyperref}

\usepackage{booktabs} 
\usepackage[ruled]{algorithm2e} 
\usepackage{enumitem}
\usepackage{multirow}
\usepackage{makecell}

\SetAlFnt{\small}
\SetAlCapFnt{\small}
\SetAlCapNameFnt{\small}
\SetAlCapHSkip{0pt}
\IncMargin{-\parindent}

\usepackage{color}
\usepackage{amsmath}
\usepackage{amssymb}
\usepackage{amsthm}
\usepackage{amsfonts}
\usepackage{latexsym}
\usepackage{bbm}
\usepackage{xspace}
\usepackage{graphicx}
\usepackage{float}
\usepackage{mathtools}

\usepackage{comment}
\usepackage{tikz}
\usepackage{verbatim}

\usepackage{dsfont}
\usepackage{bm}
\usepackage{enumitem}
\newtheorem{remark}{Remark}

\newtheorem{theorem}{Theorem}

\newtheorem{lemma}{Lemma}

\newtheorem{definition}{Definition}

\newcommand{\single}{\textsc{Single}}
\newcommand{\ind}{\mathds{1}}
\newcommand{\nf}{\textsc{Non-Favorite}}
\newcommand{\core}{\textsc{Core}}
\newcommand{\tail}{\textsc{Tail}}
\newcommand{\rev}{\textsc{Rev}}

\newcommand{\copies}{{\textsc{Copies}}}
\newcommand{\rcopies}{\textsc{Ronen}^{\textsc{Copies}}}

\newcommand{\srev}{\textsc{SRev}}
\newcommand{\brev}{\textsc{BRev}}

\newcommand{\opt}{\textsc{OPT}}

\newcommand{\vT}{\boldsymbol{t}}
\newcommand{\supp}{\textsc{Supp}}
\newcommand{\rew}{\textsc{Reward}}
\newcommand{\median}{\textsc{Median}}
\newcommand{\Var}{\textsc{Var}}
\newcommand{\Cov}{\textsc{Cov}}

\def \EE  {{\mathcal{E}}}
\newcommand{\notshow}[1]{{}}
\newcommand{\todo}[1]{{\bf \color{red} TODO: #1}}

\newcommand{\argyrisnote}[1]{{\color{magenta}{#1}}}
\newenvironment{prevproof}[2]{\noindent {\em {Proof of {#1}~\ref{#2}:}}}{$\Box$\vskip \belowdisplayskip}

\DeclareMathOperator{\argmax}{argmax}
\DeclareMathOperator{\argmin}{argmin}
\DeclareMathOperator{\E}{\mathbb{E}}

\title{On Simple Mechanisms for Dependent Items}

\begin{document}

\author{Yang Cai\footnote{Supported by a Sloan Foundation Research Fellowship and the NSF Award CCF-1942583 (CAREER).} \\Yale University, USA\\yang.cai@yale.edu
 \and Argyris Oikonomou\footnote{Supported by a Sloan Foundation Research Fellowship.}\\Yale University, USA\\ argyris.oikonomou@yale.edu 
}

\date{}

\notshow{
\author{Yang Cai}
\authornote{Supported by the NSF Award CCF-1942583 (CAREER) and a Sloan Foundation Research Fellowship.}
\email{yang.cai@yale.edu}
\affiliation{%
  \institution{Yale University}
  \city{New Haven}
  \state{CT}
}
\author{Argyris Oikonomou}
\authornote{Supported by a Sloan Foundation Research Fellowship.}
\email{argyris.oikonomou@yale.edu}
\affiliation{%
  \institution{Yale University}
  \city{New Haven}
  \state{CT}
}

}


\maketitle
\begin{abstract}
We study the problem of selling $n$ heterogeneous items to a single buyer,
whose values for different items are \emph{dependent}. Under arbitrary dependence, Hart and Nisan~\cite{HARTN_2019} show that no simple mechanism can achieve a non-negligible fraction of the optimal revenue even with only two items. We consider the setting where the buyer's type is drawn from a correlated distribution that can be captured by a Markov Random Field, one of the  most prominent frameworks for modeling high-dimensional distributions with structure. 

If the buyer's valuation is additive or unit-demand, we extend the result to all MRFs and show that $\max\{\srev,\brev\}$ can achieve an $\Omega\left(\frac{1}{e^{O(\Delta)}}\right)$-fraction of the optimal revenue, where $\Delta$ is a parameter of the MRF that is determined by how much the value of an item can be influenced by the values of the other items. We further show that the exponential dependence on $\Delta$ is unavoidable for our approach and a polynomial dependence on $\Delta$ is unavoidable for any approach.
When the buyer has a XOS valuation, we show that $\max\{\srev,\brev\}$ achieves at least an $\Omega\left(\frac{1}{e^{O(\Delta)}+\frac{1}{\sqrt{n\gamma}}}\right)$-fraction of the optimal revenue, where $\gamma$ is the spectral gap of the Glauber dynamics of the MRF.  Note that 
in the special case of independently distributed items, $\Delta=0$ and $\frac{1}{n\gamma}\leq 1$, and our results recover the known constant factor approximations for a XOS buyer~\cite{RubinsteinW15}. We further extend our parametric approximation to several other well-studied dependency measures such as the \emph{Dobrushin coefficient}~\cite{Dobruschin68} and the \emph{inverse temperature}. In particular, we show that if the MRF is in the \emph{high temperature regime}, $\max\{\srev,\brev\}$ is still a constant factor approximation to the optimal revenue even for a XOS buyer.  Our results are based on the Duality-Framework by Cai et al. \cite{CaiDW16} and a new concentration inequality for XOS functions over  dependent random variables.
\end{abstract}

\thispagestyle{empty}
\addtocounter{page}{-1}
\newpage

\section{Introduction}\label{sec:intro}
The design of revenue-optimal auctions for selling multiple items is a central problem in Economics and Computer Science.  In the past decade, significant progress has been made, first in efficient computation of revenue-optimal auctions~\cite{ChawlaHK07,ChawlaHMS10,Alaei11,CaiD11b,AlaeiFHHM12,CaiDW12a,CaiDW12b,CaiH13,CaiDW13b,AlaeiFHH13,BhalgatGM13,DaskalakisDW15}, and then in the identification of {\em simple auctions} that achieve constant factor approximations to the optimal revenue~\cite{BabaioffILW14,Yao15,RubinsteinW15,CaiDW16,ChawlaM16, CaiZ17} under the {\em item-independence} assumptions.~\footnote{Intuitively, item-independence means that each bidder's value for each item is independently distributed, and this definition has been suitably generalized to set value functions such as submodular or subadditive functions~\cite{RubinsteinW15}.} Despite being theoretically appealing, item-independence is an unrealistic assumption in practice. In this paper, we go beyond the item-independence assumption and study simple and approximately optimal auctions for selling \emph{dependent items}. 

Unfortunately, strong negative results exist if we allow the items to be arbitrarily dependent~\cite{briest2015pricing, HARTN_2019}. For example, Hart and Nisan~\cite{HARTN_2019} show that the revenue of the best deterministic mechanism is unboundedly smaller than the revenue of the optimal randomized mechanism even when we are only selling two correlated items to a single buyer. Since all simple mechanisms in the literature are deterministic, the result also implies that no simple mechanism that has been considered so far can provide any guarantee to the revenue for even two correlated items.  Arguably, however, high-dimensional distributions that arise in practice are rarely arbitrary, as arbitrary high-dimensional distributions cannot be represented efficiently, and are known to require exponentially many samples to learn or even perform the most basic statistical tests on them; see e.g.~\cite{Daskalakis18} for a discussion. To overcome the curse of dimensionality, a major focus of Statistics and Machine Learning has been on identifying and exploiting the structural properties of high-dimensional distributions for succinct representation, efficient learning, and efficient statistical inference. There are several widely-studied frameworks to model the structure of dependence in high-dimensional distributions. In this work, we propose capturing the dependence between item values using one of the most prominent graphical models -- \emph{Markov Random Fields (MRFs)}. 
Note that MRFs are fully general and can be used to express arbitrary high-dimensional distributions. The main advantage of MRFs is that there are several natural complexity parameters that allow the user to tune the dependence structure in the distributions represented by MRFs from product measures all the way up to arbitrary distributions. \textbf{Our goal is to provide parametric approximation ratios of simple mechanisms that degrade gracefully with respect to these natural parameters.}

MRFs are formally defined in Definition~\ref{def:MRF}. Intuitively, a MRF can be thought of as a graph (or a hypergraph) where each node represents a random variable (or item value in our case). There is a potential function associated with each edge that captures the correlation between the two incident random variables. How does it represent a joint distribution? The probability for a particular realization or the random variables, or known as a configuration of the random field, is proportional to the exponential of the total potential of the configuration. MRF is a flexible model. For example, we can capture the degree of (positive or negative) correlations between two random variables by controlling the corresponding potential function. Here we provide a stylistic example to illustrate the suitability of MRFs for modeling buyers' joint value distributions. Imagine that we manage a car dealership. A potential buyer is hoping to purchase one car, i.e., has a unit-demand valuation. The dealership carries various brands and types of vehicles, and will like to find the optimal way to price each car. However, it would be na\"{i}ve to assume the buyer's value for each car is independently distributed. The example in Figure~\ref{fig:car} demonstrates how a MRF can better capture the customer's joint value distribution for different cars.

\begin{figure}[h]
  \centering
 \includegraphics[width=0.55\textwidth]{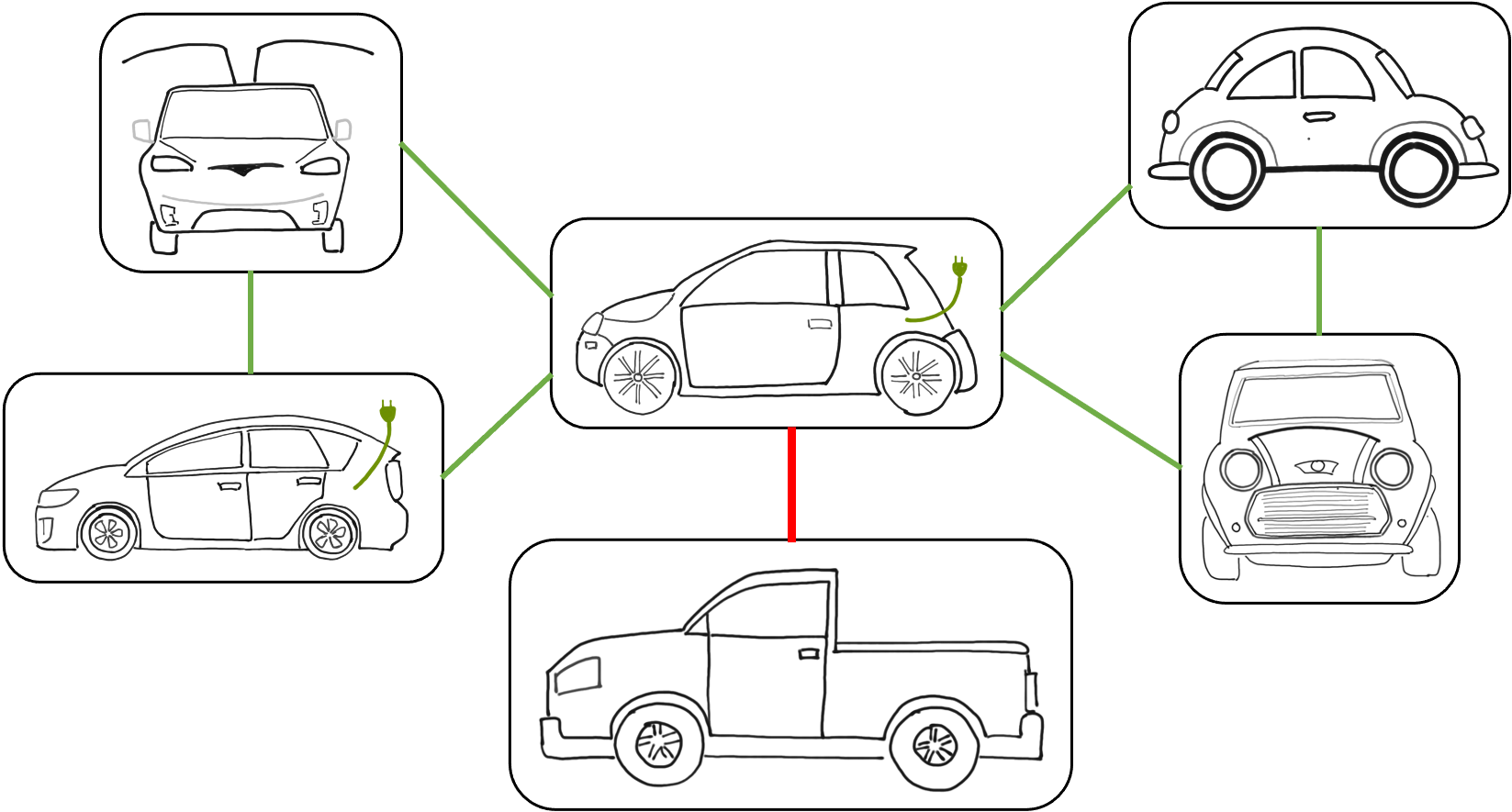}
  \caption{We draw a green edge (or a red edge) between two cars if their values are positively correlated (or negative correlated). The car in the center is an electric coupe with a retro design. Its value is positively correlated with the two electric cars on its left and the two small coupes on its right, but is negatively correlated with the pick-up truck.} \label{fig:car}
\end{figure}

\notshow{
\argyrisnote{MRF can model cases where items (or group of times) are "compatible" or "incompatible" with each other.
For example, we can think that we manage a car dealership where we have one sport car and one luxury car from two different brands (brand A and brand B).
We want to find the optimal pricing for each car, when each buyer is willing to buy only one car (unit-demand valuation).
It would be naive to assume that the buyer values each car independently.
A more realistic model,
would be to assume that the buyer is loyal to one of the brands,
or that the buyer want to buy a sport car.
This model of a correlated buyer can be modeled using a MRF as depicted in Figure~\ref{fig:car_example}.
In Figure~\ref{fig:car_example},
note that each item is represented by a node and each edge encodes the correlation between the values of the items in its endpoints.
For example the edge between the two cars of brand A is a result of the case where the buyer prefer brand A and the edge between the sport cars encodes the case where the buyer is mostly interested in buying a sport car.
}

\begin{figure}
  \centering
 \includegraphics[width=0.5\textwidth]{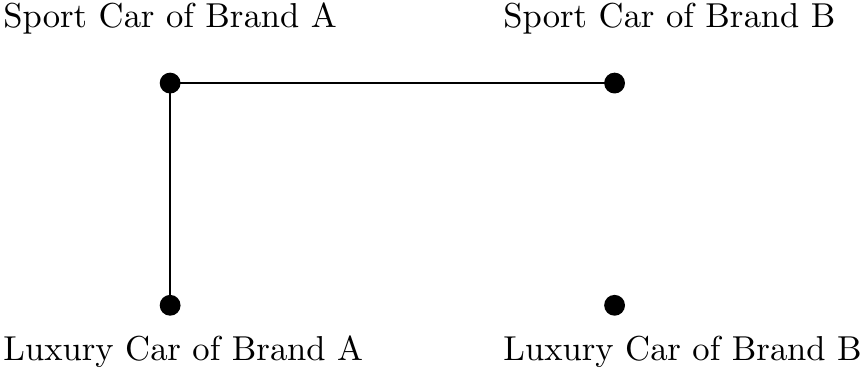}
  \caption{A MRF that captures the value of a buyer for a set of cars, when her value for each car is correlated.}
  \label{fig:car_example}
\end{figure}
}


\subsection{Main Results and Techniques}
We focus on the single buyer case and allow the buyer's valuation to be as general as a XOS function  (a.k.a. a fractionally subadditive function).~\footnote{The class of XOS functions is a super-class of submodular functions, and is contained in the class of subadditive functions. } We consider the two most extensively studied forms of simple mechanisms: \emph{selling the items separately} and \emph{selling the grand bundle}. We use $\srev$ and $\brev$ to denote the optimal revenue obtainable by these two types of mechanisms respectively. In a sequence of papers, it was shown  that $\max\{\srev, \brev\}$ is a constant factor approximation to the optimal revenue for a single additive or unit-demand buyer under the item-independence assumption~\cite{ChawlaHK07,BabaioffILW14,CaiDW16}.~\footnote{More specifically, $\srev$ denotes the optimal expected revenue achievable by any posted price mechanism. When the buyer has a unit-demand valuation,~\cite{ChawlaHK07, CaiDW16} show that $\srev$ is already a constant factor approximation of the optimal revenue.} Our first main result extends the above approximation to any MRFs. The approximation ratio degrades with the \emph{maximum weighted degree $\Delta$} that captures the degree of dependence among the item values.

\paragraph{Parameter I: Maximum Weighted Degree $\Delta$.} The formal definition can be found in Definition~\ref{def:weighted degree}. As we mentioned, a MRF can be thought of as a graph (or a hypergraph) where each node represents a random variable. The weight of an edge is related to the maximum absolute value the corresponding potential function can take and represents the ``strength'' of the dependence between two incident random variables. The weighted degree of a random variable is simply the sum of weights from all incident edges. 
If the maximum weighted degree $\Delta$ of a MRF is small, then no random variable can depend strongly on many other random variables. Note that $\Delta=0$ when the random variables are independent, and the instance constructed by Hart and Nisan~\cite{HARTN_2019} corresponds to a MRF~with~$\Delta=\infty$.

\medskip \hspace{.5cm}\begin{minipage}{0.93\textwidth}
\begin{enumerate}
\item[{\bf Result I:}] For a single additive or unit-demand buyer whose type is generated by a MRF with maximum weighted degree $\Delta$, $\max\{\srev, \brev\}=\Omega\left(\frac{\opt}{\exp(O(\Delta))}\right)$, where $\opt$ is the optimal revenue.
\end{enumerate}
\end{minipage}

\medskip The formal statement of the result is in Theorem~\ref{thm:MRF unit-demand} and~\ref{thm:MRF additive}.  We further show that the dependence on $\Delta$ is necessary. For any sufficiently large number $C$, there exists a MRF with $\Delta= O(C)$ such that $\max\{\srev, \brev\}$ is no more than $\frac{\opt}{C^{1/7}}$ (Theorem~\ref{thm:poly_dependence_Delta}) using a modification of the Hart-Nisan construction~\cite{HARTN_2019}. Although there is still an exponential gap between our upper and lower bounds, it shows that whenever Result I fails to provide a constant factor approximation (independent of the number of items), no constant factor approximation is possible without further restrictions on the dependency. We leave it as an open question to close the gap between our upper and lower bounds. The main tool we use is a generalization of the prophet inequality to the case where the rewards are sampled from a MRF (Lemma~\ref{lemma:MRF_prophet}). The overall analysis is similar to the one used by Cai et al.~\cite{CaiDW16} for the item-independent case. We show that the exponential dependence on $\Delta$ is unavoidable for this type of analysis in Theorem~\ref{thm:COPIES_example}. More specifically, a key step of the analysis involves approximating the optimal revenue in a single-dimensional setting, known as the copies setting, using $\srev$. Theorem~\ref{thm:COPIES_example} constructs an instance such that the optimal revenue in the copies setting is at least $\exp(\Delta)$ times larger than $\max\{\srev,\brev\}$.

Rubinstein and Weinberg~\cite{RubinsteinW15} show that, under the item-independence assumption, $\max\{\srev,\brev\}$ is still a constant factor approximation to the optimal revenue for a buyer whose valuation is a subadditive function. Our second main result extends their result to any MRFs when the buyer's valuation is a XOS function. The approximation ratio depends on $\Delta$ and the \emph{spectral gap of the Glauber Dynamics $\gamma$}.

\vspace{-.1in}
\paragraph{Parameter II: Spectral Gap of the Glauber Dynamics $\gamma$.} A common way to generate a sample from a high-dimensional distribution is via a Markov Chain Monte Carlo method known as the Glauber dynamics (see Definition~\ref{def:Glauber dynamics}). The spectral gap $\gamma$ of the Glauber dynamics is the difference between the largest eigenvalue $\lambda_1 =1$ and the second largest eigenvalue $\lambda_2$ of the transition matrix of the Glauber dynamics. It is well-known that $\lambda_2$ is strictly less than $1$ for any MRFs~\cite{LPW09}, so $\gamma$ is always strictly positive.


\medskip \hspace{.5cm}\begin{minipage}{0.93\textwidth}
\begin{enumerate}
\item[{\bf Result II:}] For a single XOS buyer whose type is generated by a MRF, $\max\{\srev, \brev\}=\Omega\left(\frac{\opt}{\exp({O(\Delta)})+\frac{1}{\sqrt{n\gamma}}}\right)$, where $n$ is the number of items, $\gamma$ is the spectral gap of the Glauber Dynamics, and $\Delta$ is the maximum weighted degree.~\footnote{Although the approximation ratio depends on $n$, the ratio indeed improves if we increase $n$ and fix $\gamma$~.} 
\end{enumerate}
\end{minipage}

\medskip 
Some remarks are in order. First, our approximation ratio holds for any MRFs. Second, for any $n$-dimensional random vector $X=(x_1,\ldots, x_n)$, the $X_i$'s are considered \emph{weakly dependent} if the spectral gap $\gamma=\Omega(\frac{1}{n})$. For example, when the $x_i$'s are independent, $\gamma \geq \frac{1}{n}$. Finally, the condition $\gamma \geq \Omega(\frac{1}{n})$ is extensively studied in probability theory. The condition is satisfied under the  Dobrushin uniqueness condition (see Section~\ref{sec:Dobrushin} for details), a well-known sufficient condition that ensures weak dependency; it implies rapid mixing of the Glauber dynamics (i.e., they mix in time $O(n\log n)$); it also guarantees that polynomial functions concentrate in Ising models~\cite{Gheissari18,daskalakis2019testing}.  

The formal statement of Result II can be found in Theorem~\ref{thm:XOS}. 
The analysis follows the same general framework by Cai and Zhao~\cite{CaiZ17}. The major new challenge is to prove that any XOS function $g(X)$ concentrates, when $X$ is a drawn from a high-dimensional distribution $D$.  Proving concentration inequalities for non-linear functions over dependent random variables is a non-trivial task that lies at the heart of many high-dimensional statistical problems. We prove a parametric concentration inequality for XOS functions that depends on the spectral gap of the Glauber dynamics for $D$ (Lemma~\ref{lem:self_bounding plus Poincare}). The proof is based on a combination of the Poincar\'e inequality and a special property of XOS functions -- the self-boundingness. We believe this concentration inequality may be of independent interest. 
An interesting question is whether the approximation ratio needs to depend on both $\Delta$ and $\gamma$. We show that the dependence on $\Delta$ is crucial, as no approximation can be obtained with only restriction on the spectral gap even for a single additive or unit-demand buyer (Theorem~\ref{lem:inaprx_dob}).~\footnote{Indeed, we prove an even stronger result that shows no finite approximation ratio is possible under only the Dobrushin uniqueness condition, which implies that $\gamma=\Omega(\frac{1}{n})$ (Lemma~\ref{lem:Dobrushin implies large spectral gap}).} We do not know whether it is possible to obtain an approximation that only depends on $\Delta$ for a XOS buyer and leave it as an open question.~\footnote{A na\"{i}ve approach is to directly bound $\gamma$ using a function of $\Delta$. However, this approach can at best provide an approximation ratio that is exponential in $n$, as $\frac{1}{\gamma}$ could be exponential in $n$ even when $\Delta$ is upper bounded by some absolute constant~\cite{randall2006slow}.} We suspect such an improvement requires proving a parametric concentration inequality for XOS functions that only depend on the maximum weighted degree $\Delta$, which we believe will have further applications.

\paragraph{Our Results under Other Weak Dependence Conditions.}\label{sec:other measure}
There are several alternative ways to parametrize the degree of dependency in a high-dimensional distribution. We focus on two prominent ones -- the \emph{Dobrushin coefficient} and the \emph{inverse temperature of a MRF}, and discuss how our approximation results change under these conditions. We first consider the Dobrushin coefficient and its relaxations. An important concept is the influence matrix. 

\paragraph{Influence Matrix and the Dobrushin Condition} For any $n$-dimensional random vector $\mathbf{X}=(X_1,\ldots, X_n)$, we define the influence of variable $j$ on variable $i$ as $$\alpha_{i,j}:=\sup_{\substack{x_{-i-j} \\x_j\neq x'_j}} d_{TV}\left(F_{X_i\mid X_j=x_j, X_{-i-j}=x_{-i-j}}, F_{X_i\mid X_j=x'_j, X_{-i-j}= x_{-i-j}}\right),~\footnote{$d_{TV}(\cdot,\cdot)$ denotes the total variation distance between two distributions, hence $\alpha_{i,j}$ measures the maximum total variation distance we can have between two conditional distributions of variable $i$ that only differ on the value of variable $j$.}$$ where $F_{X_i\mid X_{-i}=x_{-i}}$ denotes the conditional distribution of $X_i$ given $X_{-i}=x_{-i}$.
Let $\alpha_{i,i}:=0$ for each $i$. We define the influence matrix $A:=\left(\alpha_{i,j}\right)_{i,j\in [n]}$. When the $X_i$'s are weakly dependent, the entries of $A$ should have small values. The \emph{Dobrushin Coefficient}, defined as $||A||_\infty=\max_{i\in [n]} \sum_{j\in [n]} \alpha_{i,j}$, was  originally introduced by Dobrushin~\cite{Dobruschin68} in the study of Gibbs measures. The Dobrushin coefficient less than $1$ is known as the \emph{Dobrushin uniqueness condition}, under which the Gibbs distribution has a unique equilibrium, hence the name. The condition can be viewed as a sufficient condition that guarantees weak dependence and has been extensively studied in statistical physics and probability literature (see e.g.~\cite{DobrushinS87,StroockZ92}). As the spectral radius of any matrix is no more than its $L_\infty$ norm, a relaxation of the Dobrushin uniqueness condition is to restrict the spectral radius $\rho$ of $A$ to be less than $1$. We show that $n\gamma\geq 1- \rho$ (Lemma~\ref{lem:Dobrushin implies large spectral gap}), so we can replace the dependence on $n\gamma$ with $1-\rho$ in Result II when the item values are weakly dependent (Theorem~\ref{thm:XOS Dobrushin}). 
We also show that the dependence on $\Delta$ is necessary. Without any restriction on $\Delta$, the gap between $\max\{\srev, \brev\}$ and the optimal revenue could be unbounded even under the Dobrushin uniqueness condition for an additive or unit-demand buyer (Theorem~\ref{lem:inaprx_dob}). 
Next, we consider how the approximation guarantee degrades in terms of the inverse temperature of a MRF. 

\paragraph{Inverse Temperature $\beta$ of a MRF} The inverse temperature is related to both the maximum weighted degree and the Dobrushin coefficient. See Definition~\ref{def:Markov influence} for the formal definition. Intuitively, as the inverse temperature increases (or temperature drops), the dependence between the different random variables strengthens. When the inverse temperature is $0$, the MRF represents a product distribution.  The \emph{high temperature} regime is when the inverse temperature is less than $1$. This parameter often controls when phase transitions in the behavior of MRFs happen, and hence the name. The Dobrushin coefficient always upper bounds the inverse temperature. 
Recently, MRFs in the high temperature regime have been applied to model weakly dependent random variables~\cite{DaganDDJ19}.

We show that if the MRF is in the high temperature regime, then its maximum weighted degree $\Delta<1$ and the spectral gap $\gamma$ of the Glauber dynamics has value at least $\frac{1-\beta}{n}$. As a corollary of Result II, we have  

\medskip \hspace{.8cm}\begin{minipage}{0.9\textwidth}
\begin{enumerate}
\item[{\bf Result III:}] For a single XOS buyer, $\max\{\srev, \brev\}=\Omega\left(\sqrt{1-\beta}\cdot \opt\right)$,  where $\beta<1$ is the inverse temperature.
\end{enumerate}
\end{minipage}

\medskip The result states that as long as the inverse temperature is bounded away from $1$ by any constant,  $\max\{\srev, \brev\}$ achieves a constant fraction of the optimal revenue. Theorem~\ref{thm:XOS_high_temperature} contains the formal statement of the result. 

We summarize our results in Table~\ref{table 1} and  the relationship between the parameters in Figure~\ref{fig:parameters}. 

\begin{table}[t]
\centering
\resizebox{0.7\textwidth}{!}{
\begin{minipage}{\textwidth}	\centering
		\hspace*{-8pt}\makebox[\linewidth][c]{
		\begin{tabular}{c || c | c| c|c| }
		\hline \hline
		 & \thead{Maximum Weighted Degree $\Delta$}& \thead{Maximum Weighted Degree $\Delta$\\ and \\ Spectral Gap $\gamma$} & \thead{Spectral Gap $\gamma$\\ or\\ Dobrushin Coefficient $\alpha<1$} & \thead{Inverse Temperature $\beta<1$} \\ 
\hline
\thead{Additive\\ or\\ Unit-Demand} & \thead{UB: $\Omega\left(\frac{\opt}{\exp(O(\Delta))}\right)$ (Theorem~\ref{thm:MRF additive} and~\ref{thm:MRF unit-demand}) \\\\ LB: $O\left(\frac{\opt}{\Delta^{1/7}}\right)$ (Theorem~\ref{thm:poly_dependence_Delta})\\} &  \thead {UB: $\Omega\left(\frac{\opt}{\exp(O(\Delta))}\right)$ ($\leftarrow$)\\\\ LB: $O\left(\frac{\opt}{\Delta^{1/7}}\right)$ ($\leftarrow$) \\}&\thead{Unbounded (Theorem~\ref{lem:inaprx_dob}) }  &  \thead{UB: $\Omega\left(\sqrt{1-\beta}\cdot \opt\right)$ ($\downarrow$) \\\\ LB: open\\} \\ \cline{1-5}

	\thead{XOS}	& \thead{UB: open \\\\ LB: $O\left(\frac{\opt}{\Delta^{1/7}}\right)$ ($\uparrow$) \\}& \thead {UB: $\Omega\left(\frac{\opt}{\exp({O(\Delta)})+\frac{1}{\sqrt{n\gamma}}}\right)$ (Theorem~\ref{thm:XOS})\\\\ LB:  $O\left(\frac{\opt}{\Delta^{1/7}}\right)$ ($\leftarrow$)} &\thead{Unbounded ($\uparrow$)} 
		 &  \thead{UB: $\Omega\left(\sqrt{1-\beta}\cdot \opt\right)$ (Theorem~\ref{thm:XOS_high_temperature}) \\\\ LB: open} 
\\ \cline{1-5}
\hline
		\end{tabular}}
				\end{minipage}}
				\caption{The table contains our upper bounds and lower bounds of the approximation ratio of $\max\{\srev,\brev\}$ in various settings. The results are listed based on (i) the valuation of the buyer and (ii) the parameters the approximation ratio can depend on. In our table, an arrow means the result follows from the result that the arrow points to. }
				\label{table 1}
				\end{table}

\begin{figure}
  \centering
 \includegraphics[width=0.5\textwidth]{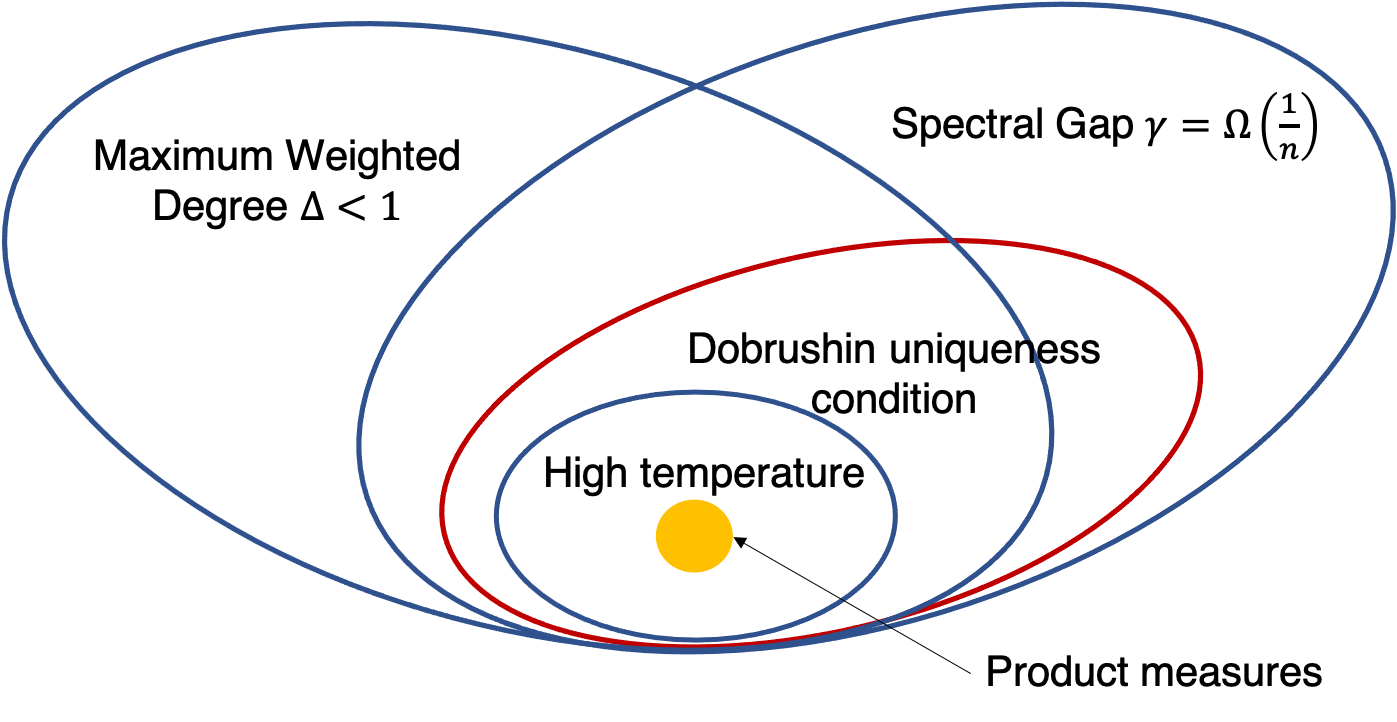}
  \caption{The relationship between the parameters: inverse temperature, Dobrushin coefficient, maximum weighted degree, and spectral gap of the Glauber dynamics.}\label{fig:parameters}
\end{figure}

\subsection{Related Work}
\paragraph{Simple vs. Optimal Auctions} There has been a large body of work on multi-item auction design  focusing on either approximation results under item-independence ~\cite{ChawlaHK07,ChawlaHMS10,Alaei11, CaiH13, LiY13, BabaioffILW14,Yao15,RubinsteinW15,CaiDW16,ChawlaM16, CaiZ17} or impossibility to approximate under arbitrary dependence~\cite{briest2015pricing, HARTN_2019}. Two types of models have been studied for items with limited dependence. The first model considers a specific type of dependence where each item's value is a linear combination of ``independent features''~\cite{ChawlaMS10,bateni2015revenue}. Unlike MRFs, this model cannot express arbitrary structure of dependence. Indeed, the values of any two items can only be positively correlated under this model. The second model considers the smoothed complexity of the problem~\cite{psomas2019smoothed}. Their result applies to arbitrary dependence structure between the item values, but only achieves an approximation ratio that is exponential in the number of items. Our paper is the first to consider a model general enough to capture arbitrary structure of dependence and obtain  parametric approximation ratios that are independent of the number~of~items.  

\vspace{-.05in}
\paragraph{MRFs and Weakly Dependent Random Variables} There has been growing interest in understanding the behavior of weakly dependent random variables that can be captured by a MRF in the high temperature regime or under the Dobrushin uniqueness condition~\cite{Gheissari18, daskalakis2019testing,DaganDDJ19}. In mechanism design, Brustle et al.~\cite{BrustleCD20} is the first to propose modeling dependent item values using MRFs in multi-item auctions, but they focus on the sample complexity of learning nearly optimal auctions.

\section{Preliminaries}\label{sec:prelim}

\paragraph{Basic Notation} We consider an auction where a seller is selling $n$ heterogeneous items to a single buyer. We denote the buyer's type $\vT$ as $\langle t_{i}\rangle_{i=1}^n$, where $t_{i}$ is the buyer's private information about item $i$. We use $D$ to denote the distribution of $\vT$, $D_i$ to denote the marginal distribution of $t_i$, and $D_{i\mid c_{-i}}$ to denote the distribution of $t_i$ conditioned on $t_{-i}=c_{-i}$. We use $\supp(\mathcal{F})$ to denote the support of distribution $\mathcal{F}$, and $T_i= \supp(D_i)$ and $T=\supp(D)$. Moreover, we use $f(c)$ to denote  $\Pr_{\vT\sim D}[\vT= c]$. For any item $i$ and any $c_i\in T_i$ and $c_{-i}\in T_{-i}$, we use $f_i(c_i)$ to denote $\Pr_{t_i\sim D_i}[t_i = c_i]$, $f_i(c_i\mid c_{-i})$ to denote $\Pr_{\vT\sim D}[t=(c_i,c_{-i})]\over \Pr_{\vT\sim D}[t_{-i}=c_{-i}]$, $f_{-i}(c_{-i})$ to denote $\Pr_{\vT\sim D}[t_{-i}=c_{-i}]$, and $f_{-i}(c_{-i}\mid c_i)$ to denote $\Pr_{\vT\sim D}[t=(c_i,c_{-i})]\over \Pr_{\vT\sim D}[t_{i}=c_{i}]$. We also define $F_i(c_i)=\Pr_{{t}_i\sim D_i}\left[{t}_i\leq c_i\right]$ and $F_i(c_i\mid c_{-i})= \Pr_{{t}_i\sim D_{i\mid c_{-i}}}\left[{t}_i\leq c_i\right]$. Finally, when the buyer's type is $\vT$, her valuation for a set of items $S$ is denoted by $v(\vT,S)$.  

We investigate the performance of simple mechanisms for several well-studied valuation classes.

\begin{definition}[Valuation Classes]\label{def:valuation classes}
We define several classes of valuations formally. 	
\begin{itemize}
	\item  \textbf{Constrained Additive:} interpret $t_i$ as the value of item $i$, and $v(\vT,S) =\max_{R\subseteq S, R\in \mathcal{I}}\sum_{i\in R} t_i$, where $\mathcal{I}\subseteq 2^{[m]}$ is a downward closed set system over the items specifying the feasible bundles. When $\mathcal{I}= 2^{[m]}$, the valuation is called \textbf{Additive}. When $\mathcal{I}$ contains all the singletons and the empty set, the valuation is called \textbf{unit-demand}.
\item \textbf{XOS/Fractionally Subadditive:}
interpret $t_i$ as $ \{t_{i}^{(k)}\}_{k\in[K]}$ that encodes all the possible values associated with item $i$, and $v(t,S)=\max_{k\in[K]}\sum_{i\in S}t_{i}^{(k)}$.
\end{itemize}

It is well known that the class of XOS valuations contains all constrained additive valuations.
\end{definition}
\vspace{-.1in}
\paragraph{Mechanism} A mechanism $M$ is specified by an allocation rule and a payment rule. We use $\pi$ to denote the allocation rule, and $\pi_i(\vT)$ is the probability that the buyer receives item $i$ when she reports type $\vT$. We also use $p(\vT)$ to denote the buyer's payment when she reports type $\vT$.
We assume the buyer has quasi-linear utility. We say a mechanism $M$ is Incentive Compatible (IC) if the buyer cannot increase their expected utility by misreporting their type, and Individual Rational (IR) if the buyer has non-negative expected utility when they report their type truthfully to the mechanism.

Given $D$, valuation function $v(\cdot,\cdot)$, we use \textbf{$\rev(M,v, D)$} to denote the expected revenue of an IC and IR mechanism $M$. We slightly abuse notation to use \textbf{$\rev(D)$} to denote the optimal revenue achievable by any IC and IR mechanism under distribution $D$. 

Throughout the paper, we use the following notations for the simple mechanisms we consider.

\noindent \textbf{- $\srev(v,D)$} denotes the optimal expected revenue achievable by any posted price mechanism, and we use $\srev$ for short if there is no ambiguity.
\noindent \textbf{- $\brev(v,D)$} denotes the optimal expected revenue achievable by selling a grand bundle and we use $\brev$ for short if there is no ambiguity.


\subsection{Markov Random Fields}\label{sec:MRF}
\begin{definition}[Markov Random Fields]
\label{def:MRF}
A Markov Random Field (MRF) is defined by a hypergraph $G=(V,E)$. Associated with every vertex $v \in V$ is a random variable $X_v$ taking values in some alphabet $\Sigma_v$, as well as a potential function $\psi_v: \Sigma_v \rightarrow \mathbb{R}$. Associated with every hyperedge $e \subseteq V$ is a potential function $\psi_e: \Sigma_e \rightarrow \mathbb{R}$. In terms of these potentials, we define a probability distribution $\pi$ associating to each vector $\bm{c} \in \bigtimes_{v\in V} \Sigma_v$ probability $\pi(\bm{c})$ satisfying:
$\pi(\bm{c}) \propto \prod_{v\in V} e^{\psi_v(c_v)} \prod_{e \in E} e^{\psi_e(\bm{c}_e)}$,where $\Sigma_e$ denotes $\times_{v\in e} \Sigma_v$ and  $\bm{c}_e$ denotes $\{c_v\}_{v\in e}$.

We refer the interested readers to~\cite{lauritzen1996,kindermann1980markov} and the references therein for more details about MRFs. Throughout the paper, when we say the type distribution $D$ is a MRF over a hypergraph $G=(V,E)$, if $V=[n]$, $t_i=x_i$, $T_i=\Sigma_i$, and there exists a collection of potential functions $\{\psi_i(\cdot)\}_{i\in [n]}$ and $\{\psi_e(\cdot)\}_{e\in E}$ so that the corresponding distribution $p(\cdot)$ equals to $D$. If there are only pairwise potentials, then $G$ is a graph.
We say that a random variable $\vT$ is generated by a MRF, if $\vT$ is sampled from a distribution that is represented by the MRF. 
\end{definition}


Next, we define two ways to measure the degree of dependence in a MRF.

\notshow{
\begin{definition}[\cite{BrustleCD2019,DaganDDJ19}]
A random variable $\bm{T}=(T_1,\ldots,T_n)\in \mathbb{R}^n$ is a Markov Random Field over a hypergraph $G=(V,E)$,
where each vertex in $G$ is associated with a random variable $T_i$,
and each hyperedge $e$ is associated with a potential function $\psi_e: \mathbb{R}^{|e|}\rightarrow \mathbb{R}$,
if there exists functions $\varphi_i: \mathbb{R}\rightarrow \mathbb{R}$ such that for all $\bm{t'}$:

\begin{align*}
    \Pr_{\bm{t} \sim T}[\bm{t}=t'] = \prod_{1\leq i \leq n} e^{\varphi_i(t_i')} \prod_{e \in E} e^{\psi_e(\bm{t}_e')}
\end{align*}

where by $\bm{t}_e'$ we denote the vector that consists only of the values of $\bm{t}$ with respect to the hyperedge $e$.
\end{definition}
}

\begin{definition}\label{def:Markov influence}
Let random variable $\vT$ be generated by a Markov Random Field over a hypergraph $G=([n],E)$,
we define the \textbf{Markov influence} between item $i$ and $j$ to be:    $\beta_{i, j}(\vT): = \max_{\bm{x} \in \times_{\ell \in [n]} T_\ell}{ \left| \sum_{\substack{e \in E :\\ i,j \in e}} \psi_e(\bm{x}_e) \right|}$.
We further define the \textbf{inverse temperature} of the MRF as $\beta(\vT):= \max_{i\in [n]}  \sum_{j \neq i} \beta_{i, j} (\vT)$. We say random variable/type $\vT$ is in the \textbf{high temperature regime} if $\beta(\vT)<1$. 
\end{definition}

\begin{definition}\label{def:weighted degree}
Given a random variable/type $\vT$  generated by a Markov Random Field over a hypergraph $G=([n],E)$,
we define the \textbf{weighted degree} of item $i$ as:
    $d_i (\vT) := \max_{\bm{x}\in \times_{i\in[n]}T_i} \left| \sum_{e \in E : i \in e} {\psi_e(\bm{x}_e)} \right|$, and the \textbf{maximum weighted degree} as
    $\Delta(\vT) := \max_{i\in [n]} d_i(\vT)$.
\end{definition}

\begin{remark}
Both $\beta(\vT)$ and $\Delta(\vT)$ capture the degree of dependence between the items. Note that $\Delta(\vT)\leq \beta(\vT)$ for any MRF $\vT$, and it is possible that $\beta(\vT)=\Omega(d\cdot \Delta(\vT))$, where $d$ is the size of the largest hyperedge in $G$. When $\vT$ is drawn from a product measure, both $\beta(\vT)$ and $\Delta(\vT)$ are $0$. In general, restricting $\beta(\vT)$ and $\Delta(\vT)$ to be small ensures that the item values are only weakly dependent. 
\end{remark}

To achieve our results, we need another important concept -- the Glauber dynamics. In Section~\ref{sec:MRF XOS}, we relate the approximation ratio achievable by simple mechanisms to the spectral gap of the Glauber dynamics of the MRF.

\begin{definition}[Glauber Dynamics]\label{def:Glauber dynamics}
Let $X_1,\ldots, X_n$ be an $n$-dimensional random vector drawn from distribution $\pi$. Let $\Omega$ be the support of $\pi$. The Glauber dynamics for $\pi$ is a reversible Markov chain with state space $\Omega$.  The Glauber chain moves from state $x \in \Omega$ as follows: an index $i$ is chosen uniformly at random from $[n]$, and a new state $y$ is chosen so that (i) $y_j=x_j$ for all $j\neq i$; (ii) draw $y_i$ from the conditional distribution $\pi\mid X_{-i}=x_{-i}$. It is not hard to verify that the Glauber dynamics is a reversible Markov chain with stationary distribution $\pi$. 
\end{definition}

\begin{remark}
When $\pi$ is the distribution that can be represented by a MRF $G=(V,E)$,  the Glauber dynamics has state space $\bigtimes_{v\in V} \Sigma_v$. The Glauber chain moves from state $x \in \bigtimes_{v\in V} \Sigma_v$ as follows: a vertex $v$ is chosen uniformly at random from $V$, and a new state $y$ is chosen so that (i) $y_u=x_u$ for all $u\neq v$; (ii) for any $c\in \Sigma_v$, $y_v= c$ w.p. $\frac{\exp(\psi_v(c))\Pi_{e:v\in e}\exp(\psi_e\left(c,x_{e/\{v\}}\right))}{\sum_{c'\in \Sigma_v}\exp(\psi_v(c'))\Pi_{e:v\in e}\exp(\psi_e \left(c',x_{e/\{v\}}\right))}$, in other words, sample $y_v$ according to the distribution conditioned on $y_{-v}=x_{-v}$. Note that for a MRF, the Glauber dynamics is an irreiducible Markov chain, so $\pi$ is its only stationary distribution. The Glauber dyanamics is a standard method for generating samples from a MRF, as it does not require computing the partition function, which is often a computationally intractable task. 
\end{remark}
\section{Markov Random Fields: Basic Properties and Tools}\label{sec:basic MRF properties}
We first present some basic properties of a MRF. Roughly speaking, we show that the conditional distribution can be approximated by the corresponding marginal distribution of $D$, and the approximation quality only depends $\Delta(\vT)$. 
\begin{lemma}\label{lemma:cond_MRF}
Let random variable $\vT$ be generated by a MRF. Then for any $t_i\in T_i,t_{-i} \in T_{-i}$:

$$    \frac{\exp(\psi_i(t_i))}{\sum_{t'_i \in T_i} \exp(\psi_i(t_i'))}\exp(-2 \Delta(\vT)) \leq f_i(t_i\mid t_{-i}) \leq \frac{\exp(\psi_i(t_i))}{\sum_{t'_i \in T_i} \exp(\psi_i(t_i'))}\exp(2 \Delta(\vT))$$
and 
$$f_i(t_i) \cdot\exp(-4\Delta(\vT)) \leq f_i(t_i\mid t_{-i}) \leq f_i(t_i)\cdot \exp(4\Delta(\vT)).$$
\end{lemma}

\begin{proof}
Note that for any $t_i'\in T_i$, 
${f\left((t_i,t_{-i})\right)\over f\left((t_i',t_{-i})\right)}={\exp(\psi_i(t_i))\over \exp(\psi_i(t'_i))}\cdot {\exp\left(\sum_{e\in E, i\in e} \psi_e\left( t_i, t_{-i}\right)_e\right) \over \exp\left(\sum_{e\in E, i\in e} \psi_e\left( t'_i, t_{-i}\right)_e\right) }$.


Clearly, $${\exp(\psi_i(t_i))\over \exp(\psi_i(t'_i))}\cdot \exp(-2\Delta(\vT))\leq {f\left((t_i,t_{-i})\right)\over f\left((t_i',t_{-i})\right)}\leq {\exp(\psi_i(t_i))\over \exp(\psi_i(t'_i))}\cdot \exp(2\Delta(\vT)).$$

Since $f_i(t_i\mid t_{-i}) = {f\left((t_i,t_{-i})\right)\over\sum_{t_i'\in T_i}f\left((t_i',t_{-i})\right)}$, $$\frac{\exp(\psi_i(t_i))\cdot\exp(-2 \Delta(\vT))}{\sum_{t'_i \in T_i} \exp(\psi_i(t_i'))} \leq f_i(t_i\mid t_{-i}) \leq \frac{\exp(\psi_i(t_i))\exp(2 \Delta(\vT))}{\sum_{t'_i \in T_i} \exp(\psi_i(t_i'))}.$$
By Law of Total Probability,
\begin{align*}
    f_i(t_i) = & \sum_{t_{-i} \in T_{-i}} f_i(t_i \mid t_{-i}) f_{-i} (t_{-i})\in\left[\frac{\exp(\psi_i(t_i))\cdot\exp(-2\Delta(\vT))}{\sum_{t_i' \in T_i} \exp(\psi_i(t_i'))} ,\frac{\exp(\psi_i(t_i))\cdot \exp(2\Delta(\vT))}{\sum_{t_i' \in T_i} \exp(\psi_i(t_i'))}\right].
\end{align*}

\end{proof}

\begin{lemma}\label{lemma:tv_MRF}
Let {random variable} $\vT$ be generated by a MRF. For any $i$ and any set $\EE\subseteq T_i$ and set $\EE'\subseteq T_{-i}$:
$$    \exp(-4\Delta(\vT)) \leq { \Pr_{\vT\sim D}\left[t_i\in \EE \land t_{-i}\in \EE' \right] \over \Pr_{t_i\sim D_i}[t_i\in \EE] \Pr_{t_{-i}\sim D_{-i}}\left[t_{-i}\in \EE'\right]}\leq \exp(4\Delta(\vT)).$$
\end{lemma}

\begin{proof}
Note that $\Pr_{\vT\sim D}\left[t_i\in \EE \land t_{-i}\in \EE' \right] =\sum_{\vT\in \EE \times \EE'} f_i(t_i\mid t_{-i})\cdot f_{-i}(t_{-i})$ and  $\Pr[t_i\in \EE] \Pr\left[t_{-i}\in \EE'\right]= \sum_{\vT\in \EE \times \EE'} f_i(t_i)\cdot f_{-i}(t_{-i})$.
Hence
$$
{ \Pr_{\vT\sim D}\left[t_i\in \EE \land t_{-i}\in \EE' \right] \over \Pr_{t_i\sim D_i}[t_i\in \EE] \Pr_{t_{-i}\sim D_{-i}}\left[t_{-i}\in \EE'\right]}  = \frac{\sum_{\vT\in \EE \times \EE'} f_i(t_i\mid t_{-i})\cdot f_{-i}(t_{-i})}{\sum_{\vT\in \EE \times \EE'} f_i(t_i)\cdot f_{-i}(t_{-i})}
$$
Using Lemma~\ref{lemma:cond_MRF} we get that:
$$
 \frac{\sum_{\vT\in \EE \times \EE'} f_i(t_i\mid t_{-i})\cdot f_{-i}(t_{-i})}{\sum_{\vT\in \EE \times \EE'} f_i(t_i)\cdot f_{-i}(t_{-i})} \leq \exp(4\Delta) \frac{\sum_{\vT\in \EE \times \EE'} f_i(t_i)\cdot f_{-i}(t_{-i})}{\sum_{\vT\in \EE \times \EE'} f_i(t_i)\cdot f_{-i}(t_{-i})} = \exp(4\Delta)
$$
and
$$
 \frac{\sum_{\vT\in \EE \times \EE'} f_i(t_i\mid t_{-i})\cdot f_{-i}(t_{-i})}{\sum_{\vT\in \EE \times \EE'} f_i(t_i)\cdot f_{-i}(t_{-i})} \geq \exp(-4\Delta) \frac{\sum_{\vT\in \EE \times \EE'} f_i(t_i)\cdot f_{-i}(t_{-i})}{\sum_{\vT\in \EE \times \EE'} f_i(t_i)\cdot f_{-i}(t_{-i})} = \exp(-4\Delta).
$$

\notshow{Similarly we can prove that  $\Pr[t_i] \geq \frac{\exp(\psi_i(t_i))}{\sum_{t_i' \in T_i} \exp(\psi_i(t_i'))}\exp(-2\Delta)$

Noting that:

\begin{align*}
\frac{\exp(\psi_i(t_i))}{\sum_{t_i' \in T_i} \exp(\psi_i(t_i'))}\exp(-2\Delta) \leq \Pr[t_i] \leq \frac{\exp(\psi_i(t_i))}{\sum_{t_i' \in T_i} \exp(\psi_i(t_i'))}\exp(2\Delta) \\
\frac{\exp(\psi_i(t_i))}{\sum_{t_i' \in T_i} \exp(\psi_i(t_i'))}\exp(-2\Delta) \leq \Pr[t_i=t\mid t_{-i}] \leq \frac{\exp(\psi_i(t_i))}{\sum_{t_i'\in T_i} \exp(\psi_i(t_i'))}\exp(2\Delta)
\end{align*}

We can infer that:

\begin{align*}
\Pr[t_i] \exp(-4\Delta) \leq \Pr[t_i\mid t_{-i}] \leq \Pr[t_i] \exp(4\Delta)
\end{align*}

Therefore we have that:

\begin{align*}
    \Pr[t > t_i \mid t_{-i}] =  \sum_{t_i' > t_i} \Pr[t_i' \mid t_{-i}] \leq \sum_{t_i' > t_i} \Pr[t_i'] \exp(4\Delta) = \Pr[t > t_i] \exp(4\Delta)
\end{align*}

We can similarly prove that $\Pr[t > t_i\mid t_{-i}] \geq \Pr[t > t_i] \exp(-4\Delta) $
}
\end{proof}

\notshow{

\begin{lemma}
\label{lemma:MRF_tv_product}
For any $i$ and any event $\EE_i\subseteq T_i$ and event $\EE_{-i}\subseteq T_{-i}$ we have that:

\begin{align*}
    \exp(-2\Delta) f_i(t_i) \Pr_{t_{-i}\sim D_{-i}}\left[t_{-i}\in \EE\right] \leq \Pr_{\vT'\sim D}\left[t'_i=t_i \land t'_{-i}\in \EE\right] \leq \exp(2\Delta)  f_i(t_i) \Pr_{t_{-i}\sim D_{-i}}[t_{-i}\in \EE]
\end{align*}
\end{lemma}

\begin{proof}
Fix an item $i$.
We can write

\begin{align*}
\Pr[t_i \cap C(T_{-i})] = \sum_{\substack{t_{-i} \in T_{-i} \\ \cap C(T_{-i})}} \Pr[t_i \mid t_{-i}] \cdot \Pr[t_{-i}] 
\end{align*}

By Lemma~\ref{lemma:cond_MRF} we have that:

\begin{align*}
\sum_{\substack{t_{-i} \in T_{-i}  \\ \cap C(T_{-i})}} \Pr[t_i \mid t_{-i}] \cdot \Pr[t_{-i}] 
\leq & \sum_{\substack{t_{-i} \in T_{-i} \\ \cap C(T_{-i})}} \Pr[t_i] \exp(2\Delta) \cdot \Pr[t_{-i}]  \\
= & \Pr[t_i] \exp(2\Delta) \sum_{\substack{t_{-i} \in T_{-i} \\ {\cap C(T_{-i})}}} \Pr[t_{-i}]  \\
= & \Pr[t_i] \exp(2\Delta) \Pr[C(T_{-i})] 
\end{align*}

This implies that:

\begin{align*}
    \Pr[t_i \cap C(T_{-i})] \geq \Pr[t_i] \exp(2\Delta) \Pr[C(T_{-i})] 
\end{align*}
\end{proof}
}

\paragraph{Prophet Inequality for MRF} Equipped with Lemma~\ref{lemma:tv_MRF}, we provide a generalization of the Prophet inequality when the rewards in different stages are dependent and generated by a MRF. We can think of the prophet inequality problem,
as finding a good policy for a gambler in a multi-round game.
At the $i$-th round, the gambler is given the choice to accept a reward or to continue to the next round. The goal of the gambler is to find a policy that obtains high expected reward,
given the distributions of the rewards at each round.
Prophet inequalities have been obtained when the rewards between stages are independent~\cite{krengel1978semiamarts,samuel1984comparison,KleinbergW12} or can be expressed a a linear combination of some independent random variables~\cite{Immorlica0W20}.

\begin{lemma}
\label{lemma:MRF_prophet}
Let $\vT=(t_1,\ldots, t_n)$ be an $n$-dimensional random vector generated by a MRF. There are totally $n$ rounds, and the reward of round $i$ is $g_i(t_i)$, where $g_i$ is an arbitrary function. The total reward of the prophet is $\E_{\vT}\left[\max_{i\in [n]} g_i(t_i)\right]$.
We denote by $\rew_{\vT}\left[\{g_i\}_{i\in [n]}, \tau\right]$ the expected of reward of the following policy -- accept any reward that is at least $\tau$. The following inequality holds if we choose $\tau^*=\median_{\vT}\left(\max_{i\in [n]} g_i(t_i)\right)$ (i.e., $\Pr[\max_{i \in [n]}g_i(t_i) \geq \tau^*]=1/2$),

\begin{equation*}
    \frac{\exp(-4\Delta(\vT))}{2} \E_{\vT}\left[\max_{i\in [n] }g_i(t_i)\right] \leq  \rew_{\vT}\left[\{g_i\}_{i\in [n]}, \tau^*\right].
\end{equation*}

\notshow{\begin{equation*}
    \frac{\exp(-4\Delta(\vT))}{2} \E_{\vT}\left[\max_{i\in [n] }g_i(t_i)\right] \leq {\tau^* + \sum_{i\in [n]}\E_{t_i\sim D_i}\left[(g_i(t_i) - \tau^*)^+\right]\over 2\exp(4\Delta(\vT))}\leq  \rew_{\vT}\left[\{g_i\}_{i\in [n]}, \tau^*\right].
\end{equation*}}

\end{lemma}

\begin{proof}
The proof is similar to the case when all $t_i$ are independent.

It is not hard to see that $$\E_{\vT}\left[\max_{i\in [n] }g_i(t_i)\right]\leq \tau^*+\sum_{i\in [n]} \E_{t_i\sim D_i}\left[\left(g_i(t_i)-\tau^*\right)^+\right].$$

We provide a lower bound on $\rew_{\vT}\left[\{g_i\}_{i\in [n]}, \tau^*\right]$.
\begin{align*}
	\rew_{\vT}\left[\{g_i\}_{i\in [n]}, \tau^*\right]\geq \Pr_{\vT}\left[\max_{i\in[n]} g_i (t_i)  \geq \tau^*\right] \cdot \tau^* + \sum_{i\in [n]}\E_{\vT}\left[(g_i(t_i) - \tau^*)^+ \cdot \mathds{1}[\max_{j\neq i} g_j(t_j) \leq \tau^*] \right]
\end{align*}  
For every $i\in [n]$, we define the set $\EE_{i}$ as $\{t_i\in T_i : g_i(t_i)> \tau^*\}$ and $\EE'_{i}$ as $\{t_{-i}\in T_{-i}: \max_{j\neq i} g_j(t_j)\leq\tau^* \}$. Note that 
\begin{align*}
\E_{\vT}\left[(g_i(t_i) - \tau^*)^+ \cdot \mathds{1}[\max_{j\neq i} g_j(t_j) < \tau^*] \right]
&=\sum_{\vT\in \EE_i\times \EE_i'}\left(g_i(t_i)-\tau^*\right)\cdot f(\vT)\\	
&\geq \exp(-4\Delta(\vT)) \sum_{\vT\in \EE_i\times \EE_i'}\left(g_i(t_i)-\tau^*\right)\cdot f_i(t_i)f_{-i}(t_{-i})\\
&= \exp(-4\Delta(\vT))\E_{t_i\sim D_i}\left[\left(g_i(t_i)-\tau^*\right)^+\right]\Pr_{t_{-i}\sim D_{-i}}[t_{-i}\in \EE'_i]
\end{align*}

The inequality is due to Lemma~\ref{lemma:tv_MRF}. Putting everything together, we know that 
$$\rew_{\vT}\left[\{g_i\}_{i\in [n]}, \tau^*\right]\geq {1\over 2} \cdot \tau^* + \sum_{i\in [n]}\E_{t_i\sim D_i}\left[(g_i(t_i) - \tau^*)^+\right]\cdot {\exp(-4\Delta(\vT))\over 2},
$$ which is at least ${\exp(-4\Delta(\vT))\over 2}$ of the upper bound we provide for $\E_{\vT}\left[\max_{i\in [n] }g_i(t_i)\right]$.
\footnote{When $\max_{i \in [n]}g_i(t_i)$ is a discrete distribution, we may not be able to pick a $\tau^*$ such that $\Pr[\max_{i \in [n]}g_i(t_i) \geq \tau^*]=1/2$, but it is folklore that this can be resolved by carefully picking a tie-breaking rule. We do not include the details here.}
%
\notshow{
Fix an item $i$.

Using Lemma~\ref{lemma:tv_MRF} with condition function $C(T_{-i})=\mathds{1}[\max_{j\neq i} g_j(t_j) < \tau]$ we have that:

\begin{align*}
    \Pr[t_i \cap \max_{j \neq i} t_j < \tau] \geq \Pr[t_i] \exp(-2\Delta) \Pr[\max_{j\neq i} g_j(t_j) < \tau] 
\end{align*}

An upper bound on $E[\max_i g(t_i)]$ is:

\begin{align*}
    \tau + \sum_{i}E \left[ (g_i(t_i) - \tau)^+ \right]
\end{align*}

A lower bound on the prophet reward is:

\begin{align*}
    &\Pr[\max_i g_i (x_i) \geq \tau] \cdot \tau + E\left[\sum_{i}(g_i(t_i) - \tau)^+ \cdot \mathds{1}[\max_{j\neq i} g_j(t_j) < \tau] \right] \\
    = & \Pr[\max_i g_i (t_i) \geq \tau] \cdot \tau + \sum_{i}E\left[(g_i(t_i) - \tau)^+ \cdot \mathds{1}[\max_{j\neq i} g_j(t_j) < \tau] \right]
\end{align*}

We have that:

\begin{align*}
    &E \left[(g_i(t_i) - \tau)^+ \cdot \mathds{1}[\max_{j\neq i} g_j(t_j) < \tau] \right] \\
    = & \sum_{\substack{t_i \in T_i: \\ t_i \geq \tau }}  \left( g_i(t_i) - \tau\right) \Pr\left[ t_i \cap  \max_{j \neq i} t_j < \tau \right] \\
    \geq & \sum_{\substack{t_i \in T_i \\ t_i \geq \tau }}  \left( g_i(t_i) - \tau\right) \Pr\left[ t_i \right] \exp(-2\Delta) \Pr[\max_{j\neq i} g_j(t_j) < \tau] \\
    = & \exp(-2\Delta) \Pr[\max_{j\neq i} g_j(t_j) < \tau] \cdot \sum_{\substack{t_i \in T_i \\ t_i \geq \tau }}  \left( g_i(t_i) - \tau\right) \Pr\left[ t_i \right] \\
    = & \exp(-2\Delta) \Pr[\max_{j\neq i} g_j(t_j) < \tau] E \left[ (g_i(t_i) - \tau)^+ \right]
\end{align*}

Therefore a lower bound on the prophet's reward is:

\begin{align*}
    &\tau \cdot \Pr[\max_{i} g_i(t_i) \geq \tau] + \exp(-2\Delta) \Pr[\max_{j\neq i} g_j(t_j) < \tau] \sum_i E \left[ (g_i(t_i) - \tau)^+ \right]
\end{align*}

By choosing a value of $\tau$ such that $\Pr[\max_i g_i(x_i) \geq \tau] = \frac{1}{2}$,
then we also have that for any $i$: $\Pr[\max_{j\neq i} g_j(x_j)< \tau] \geq \Pr[\max_{j} g_j(x_j)< \tau] = \frac{1}{2}$.

The lower bound on Prophet's reward is:

\begin{align*}
\frac{1}{2} \tau + \frac{\exp(-2\Delta)}{2}\sum_i E \left[ (g_i(t_i) - \tau)^+ \right] \geq 
\frac{\exp(-2\Delta)}{2}\left( \tau + \sum_i E \left[ (g_i(t_i) - \tau)^+ \right]\right) 
\end{align*}

And the statement is proved.

}
\end{proof}

\section{Simple Mechanisms for a Unit-Demand or Additive Buyer under MRF}\label{sec:simple vs. optimal UD and additive}

In this section, we first use the duality framework from~\cite{CaiDW16, CaiZ17} to construct an upper bound of $\rev(D)$. Next, we prove that if the buyer has either unit-demand or additive valuation across the items, $\max\{\srev,\brev\}$ is a $O(\exp(12\Delta(\vT)))$-approximation or a $O(\exp(4\Delta(\vT)))$-approximation of $\rev(D)$, respectively.

\subsection{Benchmark of the Optimal Revenue for Constrained Additive Valuations}\label{sec:constrained additive benchmark}
In this section, we use the duality framework from~\cite{CaiDW16, CaiZ17} to construct an upper bound of $\rev(D)$. We describe a benchmark of the optimal revenue for all constrained additive valuations. Deriving a benchmark for XOS valuations requires some extra care, and we provide details of the derivation in Section~\ref{sec:duality for XOS} when we study XOS valuations. We first remind the readers the partition of type space used in~\cite{CaiDW16, CaiZ17}.
\begin{definition}[Partition of the Type Space for Constrained Additive Valuations~\cite{CaiDW16, CaiZ17}]
\label{def:partition}
	We partition the type space $T$ into $n$ regions, where $R_i=\{\vT\in T: i~\text{is the smallest index in}~\argmax_{i'\in[n]}t_{i'}\}$. If $\vT\in R_i$, we call item $i$ the \textbf{favorite item} of type $\vT$.
	\end{definition}
	
To handle the dependence across the items, we introduce some new notations to specify the benchmark.	
\begin{definition}[Ironed Virtual Value]\label{def:ironed virtual values}
Let $D$ be the type distribution. For any  $\bm{t} \in R_i$, we use $\phi_i(t_i)$ to denote the \textbf{ironed Myerson's virtual value} for distribution $D_i$, $\phi_i(t_i\mid t_{-i})$ to denote the \textbf{ironed Myerson's virtual value} when we ironed $D_{i\mid t_{-i}}$ over interval $[\max_{j\neq i} t_j,\max \supp(D_{i\mid t_{-i}})]$. 

If $D_{i\mid t_{-i}}$ is a regular distribution and $t'_i= \argmin \{\hat{t} \in \supp(D_{{i}\mid t_{-i}}): \hat{t}>t_i \}$, $$
\phi_i(t_i\mid t_{-i})= t_{i} - \frac{ \left( t_{i}' -t_{i} \right) \cdot \Pr_{\hat{\vT}\sim D}\left[ \hat{t}_i> t_i \land \hat{t}_{-i}=t_{-i} \right] }{f(\vT)} =  t_{i} - \frac{ \left( t_{i}' -t_{i} \right) \cdot \left(1- F_i(t_i\mid t_{-i})\right)}{f_i( t_i \mid t_{-i})}.$$

Moreover, $\phi_i(t_i\mid t_{-i})$ always satisfies the following property: 
$$\max_{p\geq \max_{j\neq i} t_j} p\cdot (1- F_i(p\mid t_{-i}))=\sum_{\substack{t_i:~ (t_i,t_{-i})\in R_i}} f_i(t_i\mid t_{-i}) \cdot \phi_i(t_i\mid t_{-i})^+,$$ where $x^+=\max \{x,0\}$.
\notshow{Where in the independent case, we had that:
\begin{align*}
\phi_j  \left( t \right) = &t_{j} - \frac{ \left( t_{j}' -t_{j} \right) \cdot Pr\left[ t> t_j \right] }{\Pr[ t_j ]} \\
\end{align*}}
\end{definition}
Lemma~\ref{lem:benchmark constrained additive} contains the benchmark we use. See Appendix~\ref{appx:dual_add_ud} for more details about Lemma~\ref{lem:benchmark constrained additive}.
\begin{lemma}[Benchmark of Optimal Revenue for Constrained Additive Valuations]\label{lem:benchmark constrained additive}
	Given a distribution $D$ over the type space $T$, and a mechanism $M=(\pi,p)$, if the buyer's valuation $v$ is constrained additive, then we have the following benchmark:
	\begin{align*}
\rev(M,v, D)\leq &\sum_{\vT \in T} \sum_{i\in[n]} f(\vT) \cdot \pi _i (\vT) \cdot	\phi_i(t_i\mid t_{-i}) \cdot	 \mathds{1} \left[ \vT  \in R_i \right] \quad (\single)\\
& \qquad\qquad+\sum_{\vT \in T}\sum_{i\in[n]} f(\vT) \cdot  \pi _i (\vT) \cdot t_i \cdot  \mathds{1} \left[ \vT  \notin R_i \right] (\nf) \\
\leq &\sum_{\vT \in T} \sum_{i\in[n]} f(\vT) \cdot \pi _i (\vT) \cdot	\phi_i(t_i\mid t_{-i}) \cdot	 \mathds{1} \left[ \vT  \in R_i \right] \quad (\single)\\
&\qquad\qquad +\sum_{i\in[n]}\sum_{t_i > r} f_i(t_i) \cdot t_i\cdot \Pr_{\vT'\sim D}\left[ \vT'  \notin R_i \mid t'_i=t_i \right] \quad(\tail) \\
 & \qquad\qquad\qquad\qquad+\sum_{i\in[n]} \sum_{t_i \leq r} f_i(t_i) \cdot t_i \quad(\core),
 	\end{align*}
	where $r=\srev(v,D)$.
\end{lemma}

\paragraph{Single-Dimensional Copies Setting:} In the analysis of unit-demand bidders with independent items~\cite{ChawlaHMS10, CaiDW16}, the optimal revenue is upper bounded by the optimal revenue in the single-dimensional copies setting defined in~\cite{ChawlaHMS10}. We make use of the same technique in our analysis. There is a single item for sale, and we construct $n$ agents, where agent $i$ has value $t_i$ for winning the item. Notice that this is a single-dimensional setting, as each agent's type is specified by a single number. 

\subsection{A Unit-Demand Buyer}\label{sec:MRF unit-demand}
In this section, we show that a simple posted price mechanism can extract $O(\exp(12\Delta(\vT)))$ fraction of the optimal revenue when the type distribution $D$ is a MRF. We first use the revenue of the Ronen's lookahead auction~\cite{Ronen_2001} to upper bound the benchmark from Lemma~\ref{lem:benchmark constrained additive}.~\footnote{ Ronen's lookahead auction considers the setting where the seller is selling a single  item to a set of buyers, whose values for the item are correlated.}
Ronen's auction first identifies the highest bidder,
and offers a take it or leave it price to the highest bidder to maximize the revenue conditioned on the other bidders' types. The proof follows from the definition of Ronen's lookahead auction and basic properties of MRF presented in Lemma~\ref{lemma:cond_MRF} and~\ref{lemma:tv_MRF}. We postpone the proof to Appendix~\ref{sec:appx UD}.

\begin{lemma}
\label{lemma:single_ronen}
Let the type distribution $D$ be represented by a MRF, $M$ be any IC and IR mechanism for a unit-demand buyer, and $\rcopies$ be revenue of the Ronen's lookahead auction~\cite{Ronen_2001} in the $\copies$ settings with respect to $D$.
The following inequalities hold:
\begin{align*}
\max\{\single,\nf\} &\leq \rcopies \\
\rcopies &\leq \exp(8\Delta(\vT)) \E_{\vT}\left[\max_{i\in[n]} \phi_i(t_i)^+\right].
\end{align*}
\end{lemma}

Equipped with Lemma~\ref{lemma:single_ronen}, we can apply the prophet inequality for MRF to show that a posted-price mechanism can obtain expected revenue that is at least $\Omega\left(\frac{\rcopies}{\exp(12\Delta(\vT))}\right)$. We delay the proof to Appendix~\ref{sec:appx UD}.

\begin{theorem}\label{thm:MRF unit-demand}
Let the type distribution $D$ be represented by a MRF. If the buyer's valuation is unit-demand, then there exists a  posted-price mechanism $M$ with prices $\{p_i\}_{i\in [n]}$ that obtains expected revenue at least $\rev(D)\over 8\exp(12\Delta(\vT))$.
\end{theorem}

Is it possible to improve the dependence on $\Delta$? In Theorem~\ref{thm:COPIES_example}, we show that if we use the optimal revenue in the COPIES setting as a benchmark of the optimal revenue in the original setting, then the exponential dependence on $\Delta(\vT)$ is unavoidable.

\notshow{In the next theorem, we show that if we use the optimal revenue in the COPIES setting as a benchmark of the optimal revenue in the original setting, then the exponential dependence on $\Delta(\vT)$ is unavoidable. The proof is postponed to Appendix~\ref{sec:LB_copies}.

\begin{theorem}\label{thm:COPIES_example}
For any value of $n\in \mathbb{N}$ and $\beta \in \mathbb{R}_+$
there exists an MRF $D$ with $n+1$ nodes and only pairwise potentials,
such that the Markov influence $\beta_{i,j}\leq \beta$ for every pair $i$, $j$ and it holds that the expected optimal revenue in the COPIES settings can be arbitrarily close to 
$\frac{1}{2}\exp(2\beta n)$,
while $\max\{\brev, \srev\} < 2$. Note that $\Delta(t)\leq \beta n$.
\end{theorem}
}
\notshow{
\begin{lemma}
Let $T=(T_1 , \ldots , T_n)$ be a MRF with degree $\Delta$.
We can find a sequential posted priced mechanism $(SEQ)$ such that:
\begin{align*}
    E [\max_i \phi_i(x_i)] \leq 2 \exp(8\Delta) (SEQ)
\end{align*}
\end{lemma}

\begin{proof}
    According to the Prophet Inequality for MRF,
    there exists a threshold $\tau$, such that the reward of the Prophet with respect to $\phi(\cdot)$ is at least $\frac{\exp(-2\Delta)}{2}E [\max_i \phi_i(t_i)]$.
    For the value of $\tau$ stated in the lemma, 
    we set the price for item $i$ to $p_i=\phi^{-1}_i(\tau)$ and we compare the revenue of the sequential posted mechanism with prices $p_i$ to the reward of the prophet with threshold $\tau$.
    
    We reorder the items such that $p_1\leq p_2 \leq \ldots \leq p_n$.
    We assume that when the buyer has positive utility for multiple items,
    she always buys the item with the smallest price.
    This couples the item that the buyer buys with the item the prophet chooses.
    We assume that both the prophet and the buyer have visit the first $i-1$ items and have bought nothing.
    We assume that the realized values of the item so far are $t_{<i}$.
    The expected contribution of the $i$-th item to the prophet reward is, using Lemma~\ref{lemma:tv_MRF} is:
    
    \begin{align*}
        \sum_{\substack{t_i \in T_i \\ t_i \geq p_i}} \Pr[t_i\mid t_{<i}] \phi_i(t_i)
        \leq & \exp(4\Delta) \sum_{\substack{t_i \in T_i \\ t_i \geq p_i}} \Pr[t_i] \phi_i(t_i) \\
        = & \exp(4\Delta) p_i \Pr[t \geq p_i] \\
        \leq & \exp(8\Delta) p_i \Pr[t \geq p_i \mid t_{<i}]
    \end{align*}

The last term is $\exp(8\Delta)$ times the expected revenue obtained by item $i$ conditioned on the realization of the previous values.

This implies that the total reward of the prophet is upper bounded by $\exp(8\Delta) (SEQ)$.

This concludes the lemma.

\end{proof}

}

\subsection{An Additive Buyer}\label{sec:MRF additive}
In this section, we show that $\max\{\srev, \brev\}$ is a $O(\exp(4\Delta(\vT)))$ approximation of the optimal revenue when the type distribution $D$ is a MRF. We denote by $r_i$ the revenue of Myerson's auction for selling item $i$ only. We use $r=\sum_{i\in[n]} r_i$ to denote $\srev$, as the revenue collected from item $i$ only depends on the marginal distribution $D_i$. We first upper bound the terms \single~and \tail~by $\exp(4\Delta(t))\cdot \srev$. The proof follows from a combination of the standard analysis of the terms \single~and \tail~from \cite{CaiDW16,CaiZ17} with properties of MRFs (Lemma~\ref{lemma:tv_MRF}). We postpone the proof to Appendix~\ref{sec:appx_additive}.

\begin{lemma}\label{lem:bounding single and tail additive}
Let the type distribution $D$ be a MRF and $M$ be any IC and IR mechanism for an additive  buyer. The following inequalities holds:
$\single \leq \exp(4\Delta(\vT))\cdot \srev$ and $\tail\leq \exp(4\Delta(\vT))\cdot \srev$.
\end{lemma}

\notshow{
Similarly, we can bound the term \tail.

\begin{lemma}\label{lem:bounding tail additive}
Let the type distribution $D$ be a MRF and the buyer has additive valuation. The following inequality holds: $\tail\leq \exp(4\Delta(\vT))\cdot \srev$.
\end{lemma}

\begin{prevproof}{Lemma}{lem:bounding tail additive}
$\sum_{i\in[n]}\sum_{t_i > r} f_i(t_i) \cdot t_i\cdot \Pr_{\vT'\sim D}\left[ \vT'  \notin R_i \mid t'_i=t_i \right]$

First, note that $\Pr_{\vT'\sim D}\left[ \vT'  \notin R_i \mid t'_i=t_i \right] \leq \Pr_{\vT'\sim D}\left[ \exists k \neq i : t'_k \geq t_i \mid t'_i= t_i \right]$.
Therefore

\begin{align*}
     \tail \leq & \sum_{i\in[n]}\sum_{t_i > r} f_i(t_i) \cdot t_i\cdot \Pr_{\vT'\sim D}\left[ \exists k \neq i : t'_k \geq t_i \mid t'_i= t_i \right]\\
     \leq & \exp(4\Delta(\vT))\cdot \sum_{i\in[n]}\sum_{t_i > r} f_i(t_i) \cdot t_i\cdot \Pr_{t'_{-i}\sim D_{-i}}\left[ \exists k \neq i : t'_k \geq t_i \right]\\
     \leq &  \exp(4\Delta(\vT))\cdot \sum_{i\in[n]}\sum_{t_i > r} f_i(t_i) \cdot t_i\cdot \left(\sum_{k\neq i}\Pr_{t'_{k}\sim D_{k}}\left[t'_k \geq t_i \right]\right)\\
     \leq &  \exp(4\Delta(\vT))\cdot \sum_{i\in[n]}\sum_{t_i > r} f_i(t_i) \cdot \sum_{k\neq i} r_k\\
        \leq &  \exp(4\Delta(\vT))\cdot \sum_{i\in[n]}r\cdot \sum_{t_i > r} f_i(t_i) \\
        \leq &  \exp(4\Delta(\vT))\cdot \sum_{i\in[n]}r_i\\
        =& \exp(4\Delta(\vT))\cdot r
     \end{align*}

The second inequality is due to Lemma~\ref{lemma:tv_MRF}. The third inequality follows from the union bound. The fourth and sixth inequalities hold because $r_k\geq t_i\cdot \Pr_{t'_{k}\sim D_{k}}\left[t'_k \geq t_i \right]$ and $r_i\geq r\cdot (1-F_i(r))$. 
\notshow{
Using Lemma~\ref{lemma:MRF_tv_product} with condition function $C(T_{-i})=\mathds{1}[\exists k \neq j : t_k \geq t_j]$ we have that:

\begin{align*}
    Pr\left[ \exists k \neq j : t_k \geq t_j \mid t_j \right] \leq \exp(2\Delta) \Pr[t_j] Pr\left[ \exists k \neq j : t_k \geq t_j \right]
\end{align*}

This implies that:

\begin{align*}
 & \sum_j \sum_{t_j > r} \Pr[t_j] \cdot t_j Pr\left[ \exists k \neq j : t_k \geq t_j \mid  t_j \right] \\
 \leq &
\sum_j \sum_{t_j > r} \Pr[t_j] \cdot t_j \exp(2\Delta) Pr\left[ \exists k \neq j : t_k \geq t_j \right] \\ 
= &
\exp(2\Delta) \sum_j \sum_{t_j > r} \Pr[t_j]\cdot   t_j Pr\left[ \exists k \neq j : t_k \geq t_j \right] \\
\end{align*}

If we set consider the mechanism that posts price $t_j$ at each item expect item $j$,
then its expected revenue is at least $t_j \Pr\left[ \exists k \neq j : t_k \geq t_j \right]$,
which is at most $r$.
This implies that:

\begin{align*}
& \exp(2\Delta) \sum_j \sum_{t_j > r} \Pr[t_j]\cdot   t_j \Pr\left[ \exists k \neq j : t_k \geq t_j \right] \\
 \leq &
 \exp(2\Delta)\sum_j \sum_{t_j > r} \Pr[t_j] r  \\
= &
 \exp(2\Delta) \sum_j \Pr_{t\sim T_j}[t>r] r  \\
 \leq &
 \exp(2\Delta) \sum_j r_i  \\
 = &
 \exp(2\Delta) r 
\end{align*}
}

\end{prevproof}
}
Finally, we analyze the \core. We define new random variables $C_i=t_i \cdot \mathds{1}[t_i \leq r]$.
Let $C = \sum_{i=1}^n C_i$. Note that $\E[C]= \core$. We first provide an upper bound on $\Var[C]$, and show that if we sell the grand bundle at an appropriate price, its revenue is close to the \core. Note that under the item-independence assumption, it is not hard to show that $\Var[C]$ is upper bounded by $2r^2$~\cite{BabaioffILW14,CaiDW16}. However, this analysis does not extend to the case where the buyer type is generated by a MRF. We first obtain a new upper bound of $\Var[C]$. As $C=\sum_{i=1}^n C_i$, we have $\Var[C]=\sum_{i\in[n]} \Var[C_i]+\sum_{i\neq j} \Cov[C_i,C_j]$. We further bound $\sum_{i\in[n]} \Var[C_i]$ by $2r^2$ using the standard analysis in~\cite{BabaioffILW14,CaiDW16} and each covariance $\Cov[C_i,C_j]$ using properties of MRF (Lemma~\ref{lemma:tv_MRF}). The proof is postponed to Appendix~\ref{sec:appx_additive}.

\begin{lemma}
\label{lemma:cov}
Let the type distribution $D$ be represented by a MRF. For any $i,j\in [n]$, $\Cov[C_i, C_j] \leq (\exp(4\Delta(\vT))-1) \E[C_i] \E[C_j]$.
Moreover, $\Var[C]\leq 2 r^2 + (\exp(4\Delta(\vT))-1)\E[C]^2$.
\end{lemma}

In the item-independence case, the standard analysis~\cite{BabaioffILW14,CaiDW16} applies Chebyshev's inequality to show that the seller can sell the grand bundle at price $\E[C]-2r$ with probability at least $1/2$, which implies that $\core$ is $O(\brev+\srev)$. As our upper bound on $\Var[C]$ is a lot larger, Chebyshev's inequality only gives a vacuous bound on the sell probability.~\footnote{In particular, if we set the price to be $a\cdot \E[C]-\kappa\cdot r$ for any constant $a\in [0,1]$ and $\kappa$, Chebyshev's inequality tells us that the probability that the buyer cannot afford the grand bundle is at most $\frac{\Var[C]}{((1-a)\E[C]+\kappa\cdot r)^2}$. However, our upper bound of $\Var[C]$ will be larger than $((1-a)\E[C]+\kappa\cdot r)^2$, if $\exp(4\Delta(\vT))>2$ and $\E[C]$ is much larger than $r$. In this case, $\frac{\Var[C]}{((1-a)\E[C]+\kappa\cdot r)^2}$ is larger than $1$ making the bound useless.}
To show that selling the grand bundling is a good approximation of the \core, we set the price of the grand bundle differently and use the Paley-Zygmund inequality to prove that either the sell probability is high or the $\core$ is within a constant factor of $r$. The proof of Theorem~\ref{thm:MRF additive} can be found in Appendix~\ref{sec:appx_additive}.

\begin{theorem}\label{thm:MRF additive}
Let the type distribution $D$ be represented by a MRF. If the buyer's valuation is additive, then $$\left(2\exp(4\Delta(\vT))+\sqrt{2}\right)\cdot \srev+8\left(\exp(4\Delta(t))+1\right)\cdot\brev\geq \rev(D).$$
\end{theorem}

In the following theorem, we show that the approximation ratio must have polynomial dependence on $\Delta(\vT)$. Our proof is based on a modification of the hard instance by Hart and Nisan~\cite{HARTN_2019}. They construct a joint distribution over two items with support size $m$ and show that the optimal revenue is at least  $m^{1/7} \cdot \max\{\brev,\srev\}$. Unfortunately, their construction requires $\Delta$ to be infinite. We show how to modify their construction so that the new distribution has maximum weighted degree $\Delta=O(m)$, and the gap between the optimal revenue and $\max\{\brev,\srev\}$ remains to be $m^{1/7}$. The key is to show that under the new distribution, no type shows up too rarely, and the optimal revenue, $\srev$, and $\brev$ remain roughly the same.  The proof is postponed to Appendix~\ref{sec:LB_Delta_poly}.
\begin{theorem}\label{thm:poly_dependence_Delta}
For any sufficiently large $m \in \mathbb{N}$,
there exists a type distribution over two items represented by a MRF $D$ such that (i) the maximum weighted degree $\Delta$ is at most $C\cdot m$,
where $C$ is an absolute constant;  
(ii) for an additive buyer whose type is sampled from $D$, there exists an absolute constant $C'>0$ such that $\rev(D)\geq C' m^{1/7}\cdot {\max\{\brev(D),\srev(D)\}} $.
\end{theorem}

\notshow{
\begin{lemma}
We can bound the (CORE) by:
\begin{align*}
    (CORE) \leq 
    \begin{cases}
 \left(\exp(8\Delta) + \exp(4\Delta)\right)B_{Rev} + \sqrt{2}r \quad &\textit{, When  $\Delta\geq 0.2$} \\
 (2 + 4\sqrt{x} + (2\sqrt{2} + 11)x + 11x^2)B_{Rev} + \sqrt{2} r &\textit{, When $\Delta \leq 0.2$}
 \end{cases}
\end{align*}

\end{lemma}

\begin{proof}
Using Cantelli's inequality and using Lemma~\ref{lemma:var_MRF_additive} we have that:

\begin{align*}
&Pr\left[ C > (CORE) - \theta\left(\sqrt{2} r + \sqrt{(\exp(4\Delta)-1)} (CORE) \right) \right] \\
\geq & Pr\left[ C > (CORE) - \theta\sqrt{2 r^2 + (\exp(4\Delta)-1) (CORE)^2 } \right] \\
\geq &
Pr\left[ C > (CORE) - \theta\sigma \right] \geq \frac{\theta^2}{1+\theta^2}
\end{align*}

By setting $\theta=\frac{1}{1+\sqrt{\exp(4\Delta)-1}}$, we can note that $1 - \theta \sqrt{\exp(4\Delta)-1}=\theta$ and we have that:

\begin{align*}
     &(CORE) - \theta\left(\sqrt{2} r + \sqrt{(\exp(4\Delta)-1)} (CORE) \right) \\
     = & (CORE) \left(1 - \theta \sqrt{\exp(4\Delta)-1}\right) - \theta \sqrt{2} r\\
     = & \theta \left( (CORE)  - \sqrt{2}r\right)
\end{align*}

Therefore with the choice of $\theta$ as mentioned above,
and selling the grand bundle with price $\theta \left( (CORE)  - \sqrt{2}r\right)$,
we sell the grand bundle with probability at least $\frac{\theta^2}{1+\theta^2}$.
This implies that:

\begin{align*}
&\theta \left( (CORE)  - \sqrt{2}r\right) \cdot \frac{\theta^2}{1+\theta^2} \leq B_{Rev} \\
\Leftrightarrow & (CORE) \leq \frac{1}{\theta^3}B_{Rev} + \frac{1}{\theta} B_{Rev} + \sqrt{2} r
\end{align*}

Now we are going to provide a lower bound for the value of $\theta= \frac{1}{1+\sqrt{\exp(4\Delta)-1}}$ by proving a upper bound on $1+\sqrt{\exp(4\Delta)-1}$.
We note that:

{\color{red}
------------------------
I NEED TO VERIFY AGAIN
------------------------
}

\begin{align*}
 1+ \sqrt{\exp(4\Delta)-1} \leq
 \begin{cases}
 \exp(4\Delta) \quad &\textit{, When  $\Delta\geq 0.18$} \\
 1+ \sqrt{4x +  11 x^2}  &\textit{, When $\Delta \leq 0.2$} \\
 1+ 2 \sqrt{x} + \sqrt{11} x  &\textit{, When $\Delta \leq 0.2$} \\
 \exp(\sqrt{4x +  11 x^2})  &\textit{, When $\Delta \leq 0.2$}
 \end{cases}
\end{align*}

Therefore we have that:

\begin{align*}
    \theta = \frac{1}{1+\sqrt{\exp(4\Delta)-1}} \leq \sqrt{\exp(4\Delta)} = \exp(2\Delta)
\end{align*}

Therefore we conclude that:

\begin{align*}
    (CORE) \leq 
    \begin{cases}
 \left(\exp(12\Delta) + \exp(4\Delta)\right)B_{Rev} + \sqrt{2}r \quad &\textit{, When  $\Delta\geq 0.18$} \\
 (2 + 4\sqrt{x} + (2\sqrt{2} + 11)x + 11x^2)B_{Rev} + \sqrt{2} r &\textit{, When $\Delta \leq 0.2$}
 \end{cases}
\end{align*}

\end{proof}

}

\section{ Simple Mechanisms for a XOS Buyer }\label{sec:MRF XOS}

\subsection{Duality Framework for XOS Valuations}\label{sec:duality for XOS}

The benchmark is obtained using essentially the same approach as in \cite{CaiZ17}. Suppose the buyer has a XOS valuation function $v(\bm{t},S)$.
We denote by $V_i(\bm{t}) = v(\bm{t}, \{i\}) $.
We abuse this notation and we also define for $t_i \in T_i$,
$V_i(t_i) = v((\mathbf{0},\ldots,t_i,\ldots,\mathbf{0}),\{i\})$, where $\mathbf{0}$ is the all $0$ vector. We summarize the benchmark for a XOS buyer in the following Lemma. More details can be found in Appendix~\ref{sec:appx_XOS_duality}.


\begin{lemma}
\label{lemma:benchmark_XOS}
Partition the type space $T$ into $n$ regions,
where
\vspace{-.1in}
\begin{align*}
 R_i := \{ \bm{t} \in T : \text{$f(\bm{t}) >0$ and $i$ is the smallest index that belongs in $\argmax_{i\in [n]}V_i(\vT)$  }\} \end{align*}
Let $r=\srev$ be the revenue of the optimal posted price mechanism that allows the buyer to purchase at most one item.
Let $C(\bm{t}) := \{ i : V_i(\bm{t}) < 2 r \}$. For any IC and IR Mechanism $M$, we can bound its revenue by:

	\begin{align*}
\rev(M,v, D)\leq &
2\sum_{\bm{t}\in T}f(\bm{t}) \sum_{i\in [n]}  \pi_i(\bm{t})\phi(V_i(t_i) \mid \bm{t}_{-i})\mathds{1}[t \in R_i] \quad (\single)\\
&\qquad\qquad + 4\sum_{i \in [n]} \sum_{\substack{t_i \in T_i \\ {V_i(t_i)} \geq 2r}} f(t_{i}) \cdot  V_i(t_i) \Pr_{\bm{t}'\sim D}\left[\bm{t}'\notin R_i \mid t_i'=t_i \right](\tail) \\
 & \qquad\qquad\qquad\qquad+  4\sum_{\bm{t}\in T}f(\bm{t}) v(\vT,C(\bm{t}))\quad(\core)
 	\end{align*}
 	
\end{lemma}

\subsection{Approximating the Benchmark of a XOS Buyer}
In this section, we show how to approximate the optimal revenue of a buyer with a XOS valuation. We first upper bound the term \single~  and \tail~. The analysis of both terms follows from the combination of the analysis in~\cite{CaiZ17} and Lemma~\ref{lemma:tv_MRF}.
\begin{lemma}\label{lemma:XOS_single_tail}
Let the type distribution $D$ be represented by a MRF. If $M$ is an IC and IR mechanism for a buyer with a XOS valuation, then the following inequalities hold $$\single \leq 4\exp(12\Delta(\bm{t}))\cdot \srev$$
and
$$\tail \leq \exp(8\Delta(\bm{t})) \cdot \srev, $$ where $\srev$ is the revenue of the optimal posted price auction, in which the buyer is allowed to purchase at most one item.

\end{lemma}


\subsubsection{Bounding the \core~using the Poincar\'{e} Inequality}\label{sec:XOS_core}

In this section, we show how to bound the \core~for a XOS buyer. The $\core$ is the expectation of the random variable $v(\vT,C(\vT))$. To show that bundling can achieve a good approximation of the \core, we need to upper bound the variance of $v(\vT,C(\vT))$. This is the main task of this section. As $v(\vT,\cdot)$ is not additive across the items, our method for the additive valuation (see Lemma~\ref{lemma:cov}) no longer applies. We provide a new approach that is based on the \emph{Poincar\'{e} Inequality} and the \emph{self-boundingness of XOS functions}. We first state the  \emph{Poincar\'{e} Inequality}.

\begin{lemma}[The Poincar\'{e} Inequality (adapted from Lemma 13.12 of~\cite{LPW09})]\label{lem:Poincare ineq}
Let $P$ be a reversible transition matrix on state space $\Omega$ with stationary distribution $\pi$. For any function $g: \Omega \rightarrow \mathbb{R}$, let $$\mathcal{E}(g):= \frac{1}{2}\sum_{x,y\in \Omega}[g(x)-g(y)]^2\pi(x)P(x,y).$$ If $\Var_{x\sim \pi}[g(x)]>0$, then $$\frac{\mathcal{E}(g)}{\Var_{x\sim \pi}[g(x)]}\geq \gamma,$$ where $\gamma$ is the  spectral gap of $P$.~\footnote{It is well-known that the largest eigenvalue of $P$ is $1$, and the spectral gap of $P$ is the difference between $P$'s largest and second largest eigenvalues.} Moreover, there exists a function $g^*:\Omega\rightarrow \mathbb{R}$, such that $$\frac{\mathcal{E}(g^*)}{\Var_{x\sim \pi}[g^*(x)]}= \gamma.$$
\end{lemma}

Next, we apply Lemma~\ref{lem:Poincare ineq} to the Glauber dynamics of the MRF that generates the buyer's type.

\begin{lemma}~\label{lem:Poincare on MRF}
Let $D$ be the joint distribution of random variables $\vT=(t_1,\ldots,t_n)$ and $P$ be the transition matrix of the Glauber dynamics for $D$. For any function $g: T\rightarrow \mathbb{R}$, we have $$    n\gamma \cdot \Var_{\vT\sim D}[g(\vT)] \leq \sum_{i \in [n]}  \E_{\vT \sim D} \left[ \left( g(t_i,\vT_{-i}) - \E_{t'_i \sim D_{i \mid \vT_{-i}}}[g(t'_i,\vT_{-i})] \right)^2 \right] ,$$ where $\gamma$ is the spectral gap of $P$. Moreover, there exists a function $g^*:T\rightarrow \mathbb{R}$, such that the inequality is tight.

\end{lemma}
\begin{remark}
Lemma~\ref{lem:Poincare on MRF} is a generalization of the well-known Efron-Stein inequality to dependent random variables. Indeed, when $D$ is a product measure, $\gamma$ is at least $1/n$ and we recover the Efron-Stein inequality. As we demonstrate in Section~\ref{sec:Dobrushin}, $\gamma$ is at least  $\Omega(1/n)$ under many well-studied conditions of weak dependence, such as the Dobrushin uniqueness condition. 
\end{remark}
\begin{prevproof}{Lemma}{lem:Poincare on MRF}
According to the definition of the Glauber dynamics, $P$ is a reversible transition matrix on state space $T$ with stationary distribution $D$. Lemma~\ref{lem:Poincare ineq} states that 
\begin{equation}\label{eq:Poincare on MRF}
\gamma\cdot \Var_{\vT\sim D}[g(x)]\leq \frac{1}{2} \sum_{\vT,\vT'\in T}[g(\vT)-g(\vT')]^2\cdot f(\vT)\cdot P(\vT,\vT').
\end{equation}

\notshow{
and that there exists a function $g^*:T\rightarrow \mathbb{R}$, such that \begin{equation*}
\gamma\cdot \Var_{\vT\sim D}[g^*(x)]= \frac{1}{2} \sum_{\vT,\vT'\in T}[g^*(\vT)-g(^*\vT')]^2\cdot f(\vT)\cdot P(\vT,\vT').
\end{equation*}
}

By the definition of the Glauber dynamics, the RHS of Inequality~\eqref{eq:Poincare on MRF} is equivalent to \begin{align*}
    &\frac{1}{2}\E_{\vT\sim D}\left[\frac{1}{n}\sum_{i\in[n]} \E_{t'_i\sim D_{i \mid \vT_{-i}}}\left[g(\vT)-g(t'_i,\vT_{-i})\right]^2\right]\\
    =& \frac{1}{n}\sum_{i\in[n]} \E_{\vT_{-i}\sim D_{-i}}\left[\E_{t_i,t'_i \sim D_{i \mid \vT_{-i}}} \left[ \frac{1}{2}\left( g(t_i,\vT_{-i}) - g(t'_i,\vT_{-i}) \right)^2 \right]\right]\\
    =&  \frac{1}{n}\sum_{i\in[n]} \E_{\vT_{-i}\sim D_{-i}}\left[\E_{t_i \sim D_{i \mid \vT_{-i}}} \left[\left( g(t_i,\vT_{-i}) - \E_{t'_i \sim D_{i \mid \vT_{-i}}}[g(t'_i,\vT_{-i})] \right)^2 \right]\right]\\
    =& \frac{1}{n} \sum_{i \in [n]}  \E_{\vT \sim D} \left[ \left( g(t_i,\vT_{-i}) - \E_{t'_i \sim D_{i \mid \vT_{-i}}}[g(t'_i,\vT_{-i})] \right)^2 \right].
\end{align*}

The second equality is because $t_i$ and $t'_i$ are two i.i.d. samples from $D_{i\mid \vT_{-i}}$.

Hence, $$ n\gamma \cdot \Var_{\vT\sim D}[g(\vT)] \leq \sum_{i \in [n]}  \E_{\vT \sim D} \left[ \left( g(t_i,\vT_{-i}) - \E_{t'_i \sim D_{i \mid \vT_{-i}}}[g(t'_i,\vT_{-i})] \right)^2 \right].$$

Note that if we choose $g(\cdot)$ to be the function $g^*(\cdot)$ in Lemma~\ref{lem:Poincare ineq}, Inequality~\eqref{eq:Poincare on MRF} becomes an equality.

\end{prevproof}

Recall that to bound the $\core$, we need to upper bound the variance of the random variable $v(\vT,C(\vT))$. By choosing $g(\vT)$ to be $v(\vT,C(\vT))$  and applying Lemma~\ref{lem:Poincare on MRF}, we can instead upper bound the RHS of the inequality in Lemma~\ref{lem:Poincare on MRF}. A priori, it is not clear that the RHS would be easier to bound. In the following sequence of Lemmas, we show that the RHS is indeed more amenable to analysis. We first argue that the function $v(\vT,C(\vT))$ has a key property known as \emph{self-boundingness}, using which we then upper bound the RHS by $O(\srev\cdot \core)$ and show that $\srev$ and $\brev$ can approximate the $\core$.

\begin{definition}[Self-Bounding Functions~\cite{BoucheronLM00}]\label{def:self-bounding}
Let $\mathbb{S}$ be an arbitrary set and ${A}$ be a subset of $\mathbb{S}^n$. We say that a function $g(\vT): A \rightarrow \mathbb{R}$ 
is \emph{$C$-self-bounding} with some constant $C\in \mathbb{R}_+$ if there exists a collection of functions $g_i:A_{-i}\rightarrow \mathbb{R}$ for each $i\in[n]$ with $A_{-i}:=\{\vT_{-i}: \exists t_i,~(t_i, \vT_{-i})\in A\}$, such that for each $\vT \in A$ the followings hold:
\begin{itemize}
    \item $0 \leq g(\vT)-g_i(\vT_{-i}) \leq C$ for all $i\in[n]$.
    \item $\sum_{i\in[n]} \left(g(\vT)-g_i(\vT_{-i})\right)\leq g(\vT)$. 
\end{itemize}
\end{definition}

We next argue that for a self-bounding function, the RHS of the inequality in Lemma~\ref{lem:Poincare on MRF} is upper bounded by its mean.

\begin{lemma} \label{lemma:sum_bound}
Let $D$ be the joint distribution of random variables $\vT=(t_1,\ldots,t_n)$. If $g(\cdot)$ is a $C$-self-bounding function, then
\begin{align*}
    \sum_{i \in [n]} \E_{\vT \sim D} \left[ \left( g(\vT) - \E_{t'_i \sim D_{i \mid \vT_{-i}}}[g(t'_i,\vT_{-i})] \right)^2 \right] \leq C \E_{\vT \sim D} \left[ g(\vT) \right].
\end{align*}
\end{lemma}

\begin{proof}
 Recall the following property of the variance: For any real-value random variable $X$, $\Var[X]=\min_{a\in R} \E[(X-a)^2]$. In other words, $\Var[X]\leq \E[(X-a)^2]$ for any $a$. Therefore, for any $\vT_{-i}$, $$\E_{t_i \sim D_{i \mid \vT_{-i}}} \left[ \left( g(\vT) - \E_{t'_i \sim D_{i \mid \vT_{-i}}}[g(t'_i,\vT_{-i})] \right)^2\right] = \Var[g(\vT)\mid \vT_{-i}]\leq \E_{t_i \sim D_{i \mid \vT_{-i}}} \left[ \left( g(\vT) - g_i(\vT_{-i}) \right)^2\right].$$ 
 
Using this relaxation, we proceed to prove the claim.
\begin{align*}
    &  \sum_{i \in [n]} \E_{\vT \sim D} \left[ \left( g(\vT) - \E_{t'_i \sim D_{i \mid \vT_{-i}}}[g(t'_i,\vT_{-i})] \right)^2 \right]\\
    \leq & \sum_{i \in [n]}\E_{\vT \sim D} \left[ \left( g(\vT) - g_i(\vT_{-i}) \right)^2 \right]\\
    \leq &  C \sum_{i \in [n]}\E_{\vT \sim D} \left[ \left(g(\vT) - g_i(\vT_{-i})\right) \right]\\
    \leq & C \E_{\vT \sim D} \left[ g(\vT) \right]
\end{align*}
The first inequality follows from the relaxation. The second and last inequality follow from the first and second property of a self-bounding function respectively.
\end{proof}

Combining Lemma~\ref{lem:Poincare on MRF} and~\ref{lemma:sum_bound}, we have the following Lemma.

\begin{lemma}\label{lem:self_bounding plus Poincare}
Let $D$ be the joint distribution of random variables $\vT=(t_1,\ldots,t_n)$ and $P$ be the transition matrix of the Glauber dynamics for $D$. For any $C$-self-bounding function $g: T\rightarrow \mathbb{R}$, we have $$ \frac{n\gamma}{C}\cdot \Var_{\vT\sim D}[g(\vT)] \leq  \E_{\vT \sim D} \left[ g(\vT) \right],$$ where $\gamma$ is the spectral gap of $P$.
\end{lemma}

Definition~\ref{def:self-bounding} may seem obscure at first, but many natural functions are indeed self-bounding. For example, if $A$ is $[0,1]^n$ and $g(\cdot)$ is the additive function, then $g(\cdot)$ is $1$-self-bounding. We show that the function $g(\vT):= v(\vT,C(\vT))$ is $2\srev$-self-bounding and its variance is no more than $\frac{2\srev\cdot\core}{n\gamma}$.  Here, we first specialize our analysis to MRFs. The main difference is that the Glauber dynamics for a MRF is irreducible, so the spectral gap is strictly positive (Lemma 12.1 of~\cite{LPW09}). The proof is postponed to Appendix~\ref{sec:appx_XOS_single_tail}.

\begin{lemma}
\label{lemma:self_bounding_MRF}
Let $C(\bm{t}) := \{ j : V_j(\bm{t}) < 2 \srev\}$. The function $g(\vT):=v(\vT,C(\vT))$ is $2\srev$-self-bounding and $\Var_{\vT\sim D}[g(\vT)]\leq \frac{2\srev\cdot\core}{n\gamma}$, where $\gamma>0$ is the spectral gap of the transition matrix of the Glauber dynamics of the MRF that generates the buyer's type.
\end{lemma}

\notshow{Now we are ready to present the Poincar\'e inequality. 
At a high level,
Poincar\'{e} inequalities connect the variance of a function $g(\cdot)$ with the $\ell_2$-norm of the gradient of $g(\cdot)$, and usually take the form of $\Var_{X\sim D}[g(X)]\lesssim\E_{X\sim D}[||\text{gradient}(g(X))||^2]$. These inequalities are called \emph{Poincar\'e inequalities} because Poincar\'{e} first published an inequality of this type for the uniform distribution on a bounded domain in $\mathbb{R}^n$, ca. 1890. Many similar inequalities have since been discovered. One notable example is the Efron-Stein inequality~\cite{efron1981jackknife}, which has numerous applications in combinatorics, statistics, and statistical learning~\cite{massart2007concentration,boucheron2013concentration}. The Efron-Stein inequality only applies to \emph{independent random variables}. The following result by Wu~\cite{wu2006} provides a generalization of the Efron-Stein inequality for \emph{weakly dependent random variables}.

\begin{lemma}[Poincar\'{e} Inequality for Weakly Dependent Random Variables \cite{wu2006}]
\label{thm:var_d}
Let $\vT=(t_1,\ldots,t_n)$ be an $n$-dimensional random vector drawn from distribution $D$ that is supported on $\E^n$. For any metric $d(\cdot,\cdot)$ on $\E$, let $A$ be the $d$-Dobrushin interdependence matrix for $\vT$. Let $\rho_d(\vT)$ be the spectral radius of $A$ (which is the dominant eigenvalue of $A$ by the Perron-Frobenius Theorem).
If $\E_{\vT \sim D}\left[ \sum_{i\in[n]} d(t_i,y_i)^2 \right] < + \infty$ for some fixed $y \in \E^n$ and  $\rho_d(\vT) < 1$, then for any square integrable function $g(\cdot)$ w.r.t. distribution $D$, the following holds:
$$    (1-\rho_d(\vT))\Var_{\vT\sim D}[g(\vT)] \leq \sum_{i \in [n]} \E_{\vT_{-i} \sim D_{-i}} \left[ \E_{t_i \sim D_{i \mid \vT_{-i}}} \left[ \left( f(t_i,\vT_{-i}) - \E_{t_i \sim D_{i \mid \vT_{-i}}}[f(t_i,\vT_{-i})] \right)^2 \right] \right].$$
\end{lemma}

\begin{remark}
In the independent case, $\rho_d(\vT)=0$, which recovers the Efron-Stein inequality. The condition that $\rho_d(\vT)<1$ is a natural condition. Firstly, it is implied by the \textbf{Dobrushin uniqueness condition} (see Section~\ref{subsec:core_dob} for more discussion). Secondly, the condition implies rapid mixing of the Glauber dynamics (see Lemma~\ref{lemma:rapid mixing}). A special case of Lemma~\ref{lemma:rapid mixing}, where $d(x,y)=\ind_{x\neq y}$, was shown by Hayes~\cite{Hayes06}. 

\end{remark} 
}

Now, we show how to approximate $\core$ using $\srev$ and $\brev$.

\begin{lemma}\label{lemma:XOS_core}
Let the buyer's type distribution $D$ be represented by a MRF, $P$ be the transition matrix of the Glauber dynamics of the MRF, and $\gamma>0$ be the spectral gap of $P$. We have 
$$\core \leq \max\left( \frac{4\srev}{\sqrt{n\gamma}} ,\left(7 + \frac{4}{\sqrt{n\gamma}} \right)\brev \right).$$
\end{lemma}
\begin{proof}
According to Lemma~\ref{lemma:self_bounding_MRF}, $v(\vT,C(\vT))$ is a $2\srev$-self-bounding function and  $\Var[v(\vT,C(\vT)]\leq \frac{2\srev\cdot \core}{n\gamma}$. If $\core \leq \frac{4\srev}{\sqrt{n\gamma}}$,
then the statement holds. If $\core > \frac{4\srev}{\sqrt{n\gamma}}$,
then $\Var[v(\vT,C(\vT))] \leq \frac{2\srev\cdot \core}{n\gamma} <  \frac{(\core)^2}{2\sqrt{n\gamma}} = \frac{ \E[v(\vT,C(\vT))]^2}{2\sqrt{n\gamma}}$.
By Paley-Zygmund inequality we have that:
    $$ \Pr\left[ v(\bm{t},C(\bm{t})) \geq \frac{\core}{3}  \right] 
    \geq \frac{4}{9} \frac{1}{1 + \frac{\Var[v(\vT,C(\vT))]}{\E[v(\vT,C(\vT))]^2}} 
    \geq  \frac{4}{9} \frac{1}{1 + \frac{1}{2\sqrt{n\gamma}}}.$$
Therefore we have that: $\Pr\left[ v(\bm{t},C(\bm{t})) \geq \frac{\core}{3}  \right]\cdot \frac{\core}{3} \leq \brev$, which implies that the statement. 

\end{proof}

Finally, we combine our analysis of \single, \tail, and \core~to obtain the approximation guarantee for a XOS buyer.
\begin{theorem}\label{thm:XOS}
 Let the buyer have a XOS valuation and her type distribution $D$ be represented by a MRF. We use $\gamma$ to denote the spectral gap of matrix $P$ -- the transition matrix of the Glauber dynamics of the MRF. Then 
    $\rev(D)\leq  12\exp(12\Delta(\vT)) \cdot \srev + \left(28 + \frac{16}{\sqrt{n\gamma}}\right)\max\{\srev,\brev\}$. 
\end{theorem}
\begin{prevproof}{Theorem}{thm:XOS}
The statement follows from the combination of  Lemma~\ref{lemma:benchmark_XOS} ,~\ref{lemma:XOS_single_tail}, and~\ref{lemma:XOS_core}.
\end{prevproof}
\section{Connection to other Weak Dependence Conditions}~\label{sec:Dobrushin}
A common way to measure the degree of dependence of a high-dimensional distribution is by considering its \emph{Dobrushin Interdependence Matrix}. In this section, we show that for several natural sufficient conditions that guarantee weak dependence in the distribution, the spectral gap $\gamma$ of the Glauber dynamics transition matrix is $\Omega(1/n)$. We begin by defining the \emph{Dobrushin interdependence matrix}.

\begin{definition}[$d$-Dobrushin Interdependence Matrix~\cite{wu2006}]\label{def:Dobrushin's interdependence matrix}
Let $(\E,d)$ be a metrical, complete and separable space.
For two distributions $\mu$ and $\nu$ supported on $\E$,
their $L^1$-Wasserstein distance is defined as:    $W_{1,d}(\mu,\nu) = \inf_{\pi \in \Pi} \int\int_{\E \times \E} d(x,y) \pi(dx,dy)$, where $\Pi$ is the set of valid coupling such that its marginal distributions are $\mu$ and $\nu$.

Let $X=(x_1,\ldots, x_n)$ be a $n$-dimensional random vector supported on $\E^n$ and $\mu_i(\cdot \mid x_{-i})$ be the conditional distribution of $x_i$ knowing $x_{-i}$. 
Define the \textbf{$d$-Dobrushin Interdependence Matrix} $A=(\alpha_{i,j})_{i,j\in[n]}$ by
$$\alpha_{i,j}:= \sup_{\substack{x_{-i-j}=y_{-i-j}\\ x_j\neq y_j}}\frac{W_{1,d}(\mu_i(\cdot\mid x_{-i}),\mu_i(\cdot\mid y_{-i}))}{d(x_j,y_j)}~\text{for all $i\neq j$},$$
and $\alpha_{i,i}=0$ for all $i\in [n]$.
\end{definition}
\begin{remark}
$\alpha_{i,j}$ captures how strong the value of $x_j$ affects the conditional distribution of
$x_i$ when all other coordinates are fixed. Higher $\alpha_{i,j}$ value implies stronger dependence between $x_i$ and $x_j$. When all the coordinates of $X$ are independent, $A$ is the all zero matrix.  
\end{remark}

\paragraph{\textbf{Dobrushin uniqueness condition:}} If we choose $d(x,y)$ to be the trivial metric $\ind_{x\neq y}$, then $W_{1,d}(\cdot,\cdot)$ is exactly the total variation distance. The influence matrix mentioned in Section~\ref{sec:intro} is exactly the Dobrushin interdependence matrix with respect to the trivial metric. To remind the audience, the \textbf{Dobrushin Coefficient} is defined as $\alpha(\vT):=||A||_\infty=\max_{i\in [n]} \sum_{j\neq i} \alpha_{i,j}$ when $A$ is the influence matrix. If $\alpha(\vT)<1$, we  say $\vT$ satisfies the \textbf{Dobrushin uniqueness condition}. As $||A||_\infty$ is at least as large as $A$'s spectral radius $\rho_d(\vT)$,~\footnote{$\rho_d(\vT)$ is the dominant eigenvalue of $A$ by the Perron-Frobenius Theorem.} a weaker condition than the Dobrushin uniqueness condition is that the spectral radius $\rho_d(\vT)$ is strictly less than $1$.

We argue that even the weaker condition that $\rho_d(\vT)<1$ implies that the spectral gap of the transition matrix of the Glauber dynamics $\gamma=\Omega(1/n)$. 

\begin{lemma}\label{lem:Dobrushin implies large spectral gap}
Let $d(\cdot,\cdot)$ be any metric, for any $n$-dimensional random vector $\vT$, $n\gamma\geq 1-\rho_d(\vT)$, where $\gamma$ is the spectral gap of the transition matrix of the Glauber dynamics for $\vT$.
\end{lemma}

\begin{proof}
By Lemma~\ref{lem:Poincare on MRF},
there exists a function $g^*:\Omega\rightarrow \mathbb{R}$, such that
$$   \frac{\sum_{i \in [n]}  \E_{\vT \sim D} \left[ \left( g^*(t_i,\vT_{-i}) - \E_{t'_i \sim D_{i \mid \vT_{-i}}}[g^*(t'_i,\vT_{-i})] \right)^2 \right]}{\Var_{\vT\sim D}[g^*(\vT)]} = n\gamma.$$

The following result by Wu~\cite{wu2006} provides a generalization of the Efron-Stein inequality for \emph{weakly dependent random variables}.

\begin{lemma}[Poincar\'{e} Inequality for Weakly Dependent Random Variables - Theorem~2.1 in \cite{wu2006}]
\label{thm:var_d}
Let $\vT=(t_1,\ldots,t_n)$ be an $n$-dimensional random vector drawn from distribution $D$ that is supported on $\E^n$. For any metric $d(\cdot,\cdot)$ on $\E$, let $A$ be the $d$-Dobrushin interdependence matrix for $\vT$. Let $\rho_d(\vT)$ be the spectral radius of $A$. If $\E_{\vT \sim D}\left[ \sum_{i\in[n]} d(t_i,y_i)^2 \right] < + \infty$ for some fixed $y \in \E^n$ and  $\rho_d(\vT) < 1$, then for any square integrable function $g(\cdot)$ w.r.t. distribution $D$, the following holds:
$$    (1-\rho_d(\vT))\Var_{\vT\sim D}[g(\vT)] \leq \sum_{i \in [n]} \E_{\vT \sim D} \left[  \left( g(t_i,\vT_{-i}) - \E_{t_i' \sim D_{i \mid \vT_{-i}}}[g(t_i,\vT_{-i})] \right)^2 \right] .$$
\end{lemma}

If we choose $g(\cdot)$ to be $g^*(\cdot)$ in Lemma~\ref{thm:var_d},
we have that :
$$    1-\rho_d(\vT)\leq \frac{\sum_{i \in [n]} \E_{\vT \sim D} \left[  \left( g^*(t_i,\vT_{-i}) - \E_{t_i' \sim D_{i \mid \vT_{-i}}}[g^*(t_i,\vT_{-i})] \right)^2 \right]}{\Var_{\vT\sim D}[g^*(\vT)] }$$

Combining the two inequalities conclude the proof
\end{proof}

Combining Lemma~\ref{lem:Dobrushin implies large spectral gap} with Theorem~\ref{thm:XOS}, we immediately have the following Theorem.~\footnote{A major benefit of using $\rho_d$ or the Dobrushin coefficient rather than $\gamma$ is that these parameters are easier to estimate than $\gamma$ given the joint distribution.}
 
\notshow{
\todo{Change the mixing time proof for this. Use Theorem 13.1 of~\cite{LPW09}.}

\begin{proof}
Given a metric  $d$ such that the spectral radius of the $d$-Dobrushin interdependence matrix is $\rho_d(\vT)$,
first we are going to prove that there exists another metric $d'$ such that the $d'$-Dobrushin interdependent matrix has norm  one equal to $\rho_d(\vT)$.
W.l.o.g. we assume that $T_i \cap T_j = \empty$ for $i \neq j$.
If for some $i\neq j$,
$T_i \cap T_j \neq \empty$,

Let $(x,y)\in T \times T$,
where $x = (x_1,\ldots,x_n)$ and $y=(y_1,\ldots,y_n)$.
Let $d(x,y)= [d(x_1,y_1),\ldots,d(x_n,y_n)]$ be the vector that contains the distance between every element of $x$ with the corresponding element of $y$ w.r.t. metric $d$
and by $d_{-i}(x,y)$ we the vector that results from $d(x,y)$ is we drop the $i$-th element,
that is $d_{-i}(x,y)=[d(x_1,y_1),\ldots,d(x_{i-1},y_{i-1}),d(x_{i+1},y_{i+1}),\ldots,d(x_n,y_n)]$.
We consider the metric $d_1(x,y) = ||d(x,y)||_2$.
Let $X=(X_1,\ldots,X_n)$ (and $Y=(Y_1,\ldots,Y_n)$ resp.) be the one step evolution of the Glauber dynamics when the chain starts at state $x$ (and $y$ resp) under the coupling that minimizes the expected value of $d(X,Y)$.
We are going to prove that there exists a coupling between random variables $X$ and $Y$ such that $\E\left[d_1(X,Y)\right] \leq \left( 1 - \frac{1-\rho_d(\vT)}{n} \right)d_1(x,y)$ using the Path Coupling Technique.

We consider the intermediate following types $\{z^i=(z^i_1,\ldots,z^i_n)\}_{0\leq i \leq n}$. 
First $z^0 = x$ and for $i\in[n]$:
\begin{align*}
    z^i_k = 
    \begin{cases}
    y^i_i \quad &\text{if $k=i$} \\
    z^{i-1}_k &\text{otherwise}
    \end{cases}
\end{align*}

Observe that $z^0=x$, $z^n=y$ and for any $0 \leq i\leq n-1$,
$z^i$ and $z^{i+1}$ can differ only in the $i+1$ coordinate.
Thus $d(z^i,z^{i+1})$ is the all $0$ vector,
but the in the $i+1$ coordinate has value $d(x_{i+1},y_{i+1})$.
Thus:

\begin{align*}
    d(x,y) = \sum_{0 \leq i\leq n-1}d(z^i,z^{i+1})
\end{align*}

Let $\{Z^i=(Z^i_1,\ldots,Z^i_n)\}_{0 \leq i \leq n}$ be the one step evolution of the Glauber dynamics starting from $\{z^i\}_{0 \leq i \leq n}$ under the following coupling.
First we sample an item $i^*$ uniformly at random in $[n]$.
In all $\{Z^i\}_{0 \leq i \leq n}$ we are going to update the value of the $i^*$-th item (that implies that $Z^i_{-i^*} = z^i_{-i^*}$ for all $i$).
For all $i$,
$Z^i$ is chosen by starting at $z^i$,
and resampling item $i^*$ from the distribution of the $i^*$-th element condition on $z^i_{-i^*}$.
We consider the following coupling among the resampled types:
We first just sample $Z^0_{i^*}$ and we for each $i$ we couple the distribution of $Z^{i+1}_{i^*}$ and $Z^i_{i^*}$ that minimizes $\E[d(Z^i_{i^*},Z^{i+!}_{i^*})]$.

Condition on the fact that item $i^*$ was chosen in this coupling procedure we have that:
\begin{align*}
    \E[d(X_{i^*},Y_{i^*}) \mid \text{item $i^*$ was chosen}] \leq & \sum_{0\leq i\leq n-1}\E[d(Z^i_{i^*},Z^{i+1}_{i^*}) \mid \text{item $i^*$ was chosen}] \\
    \leq & \sum_{j\in[n]}a_{i^*,j}d(x_{j},y_{j})
\end{align*}
and for items $j\neq i^*$ we have that: 
\begin{align*}
    \E[d(X_{j},Y_{j}) \mid \text{item $i^*$ was chosen}] = d(x_j,y_j)
\end{align*}

Since item $i^*$ is sampled with probability $\frac{1}{n}$ we have that:

\begin{align*}
    \E[d(X_{i^*},Y_{i^*})] = & \left(1 - \frac{1}{n}\right)d(x_{i^*},y_{i^*}) +\frac{1}{n}\E[d(X_{i^*},Y_{i^*}) \mid \text{item $i^*$ was chosen}] \\
    \leq &\left(1-\frac{1}{n}\right)d(x_{i^*},y_{i^*}) +  \frac{1}{n}\sum_{j\in[n]}a_{i^*,j}d(x_{j},y_{j})
\end{align*}

Let $A$ be the $d$-Dobrushin matrix of $\vT$ w.r.t. metric $d(\cdot,\cdot)$.
Note that $a_{i^*,j}d(x_{j},y_{j})$ is the inner product of the $i^*$-th row of $A$ and $d(x,y)$.
Thus we have that:

\begin{align*}
    \E[d(X,Y) ] \leq & \left( 1-\frac{1}{n}\right) d(x,y) + \frac{1}{n}A d(x,y) \\
    = & \left( \left( 1-\frac{1}{n}\right) I + \frac{1}{n}A \right) d(x,y)
\end{align*}

Therefore we have have that:

\begin{align*}
    \left| \left| \E[d(X,Y)] \right|\right|_2
    \leq & \left| \left|\left( \left( 1-\frac{1}{n}\right) I + \frac{1}{n}A \right) d(x,y)\right|\right|_2 \\
    \leq &\left| \left| \left( 1-\frac{1}{n}\right) I + \frac{1}{n}A \right|\right|_2 \left|\left|d(x,y)\right|\right|_2 \\
    = &\left(1 -  \frac{1-\rho_d(\bm{t})}{n} \right)d_1(x,y)
\end{align*}

Where the last inequality follows from the fact that if matrix $A$ has spectral radius $\rho_d(\bm{t})$,
then matrix $\left( 1-\frac{1}{n}\right) I + \frac{1}{n}A$ has spectral radius $1-\frac{1}{n}+\frac{\rho_d(\bm{t})}{n}$.

Now we are going to prove that $\E\left[ d_1(X,Y)\right] = \E\left[ \left| \left| d(X,Y) \right|\right|_2\right] \leq \left| \left| \E[d(X,Y)] \right|\right|_2 $.
Let $\bar{D_i} = \E  \left[d(X_i,Y_i) \mid \text{item $i$ was chosen}\right]$.

\begin{align*}
    \E\left[ \left| \left| d(X,Y) \right|\right|_2\right] =& \sum_{i \in [n]}\frac{1}{n}\E \left[\sqrt{\sum_{j \neq i }d(x_i,y_i)^2 + d(X_i,Y_i)^2}\mid \text{item $i$ was chosen}\right] \\
    \leq & \sum_{i \in [n]}\frac{1}{n}\left(\sqrt{\sum_{j \neq i }d(x_i,y_i)^2} + \E \left[\sqrt{ d(X_i,Y_i)^2}\mid \text{item $i$ was chosen}\right]\right) \\
    =& \sum_{i \in [n]}\frac{1}{n}\left(||d_{-i}(x,y)||_2 +\E  \left[d(X_i,Y_i) \mid \text{item $i$ was chosen}\right]\right) \\
    = & \sum_{i \in [n]}\frac{1}{n}\left(||d_{-i}(x,y)||_2 +\bar{D_i}\right)
\end{align*}

Also according to the coupling described above we have that $d(X_i,Y_i) = \left( 1 - \frac{1}{n}\right)d(x_i,y_i) + \frac{1}{n}\bar{D_i}$.
We have that:

\begin{align*}
     \left| \left| \E\left[d(X,Y)\right] \right|\right|_2 = &\sqrt{ \sum_{i \in [n]}\left( \left(1-\frac{1}{n}\right)d(x_i,y_i) +\frac{1}{n}\bar{D_i}\right)^2} \\
    = &\frac{1}{n}\sqrt{ \sum_{i \in [n]}\left( (n-1)d(x_i,y_i) +\bar{D_i}\right)^2} \\
    \geq & \frac{1}{n} \sum_{i \in [n]}\left( (n-1)d(x_i,y_i) +\bar{D_i}\right) \\
    = & \frac{1}{n} \sum_{i \in [n]}\left( ||d_{-i}||_1 +\bar{D_i}\right)
\end{align*}

where the first inequality follows from Cauchy Schwartz.
Thus the inequality $\E\left[ d_1(X,Y)\right] = \E\left[ \left| \left| d(X,Y) \right|\right|_2\right] \leq \left| \left| \E[d(X,Y)] \right|\right|_2 $ is equivalent to:

\begin{align*}
    \frac{1}{n} \sum_{i \in [n]}\left( ||d_{-i}||_2 +D_i\right) \leq \frac{1}{n} \sum_{i \in [n]}\left( ||d_{-i}||_1 +D_i\right)
\end{align*}

which is true.
Thus $d_1(X,Y) \leq \left( 1 - \frac{1-\rho_d(\vT)}{n} \right)d_1(x,y)$

For fixed $0\leq i \leq n$ we are going to bound $\E[d(Z^i,Z^{i+1})]$ under the coupling we mentioned above.
With probability $1-\frac{1}{n}$ we dot choos
By the definition of $d$-Dobrushin matrix,
if item $i^*$ was chosen,
then:
\end{proof}
}

\begin{theorem}\label{thm:XOS Dobrushin}
 Let the buyer have a XOS valuation, her type distribution $D$ be represented by a MRF, and $\rho_d(\vT)$ be the spectral radius of the $d$-Dobrushin interdependence matrix of $\vT$ under some metric $d(\cdot,\cdot)$. If $\rho_d(\vT)<1$, then 
    $\rev(D)\leq  12\exp(12\Delta(\vT)) \cdot \srev + \left(28 + \frac{16}{\sqrt{1-\rho_d(\vT)}}\right)\max\{\srev,\brev\}$. 
\end{theorem}
\begin{prevproof}{Theorem}{thm:XOS Dobrushin}
The statement follows from the combination of Lemma~\ref{lem:Dobrushin implies large spectral gap} and Theorem~\ref{thm:XOS}.
\end{prevproof}

\paragraph{High Temperature MRFs.}\label{subsec:core_dob}
Using Theorem~\ref{thm:XOS Dobrushin}, we show that when the MRF is in the \emph{high temperature regime}, i.e., $\beta(\vT) < 1$ (see Definition~\ref{def:Markov influence}), $\max\{\srev,\brev\}$ is a constant factor approximation to the optimal revenue. By the definition of $\beta(\vT)$, it clear that $\Delta(\vT)\leq \beta(\vT)$. Next, we show that $\beta(\vT)$ is also an upper bound of $\rho_d(\vT)$ for the trivial metric $d(x,y)=\ind_{x\neq y}$.

\begin{lemma}\label{lemma:Dobruthin properties}
Let $d(\cdot,\cdot)$ be the trivial metric $d(x,y)=\ind_{x\neq y}$. For any MRF $\vT$, $\rho_d(\vT)\leq \alpha(\vT)\leq \beta(\vT)$. Moreover, $\alpha_{i,j}\leq \beta_{i,j}(\vT)$ for all $i,j\in[n]$.
\end{lemma}

\begin{proof}
$\rho_d(\vT)\leq \alpha(\vT)$ follows the elementary fact that the spectral radius is upper bounded by the infinity norm. To prove $\alpha(\vT)\leq \beta(\vT)$, we first need the following definition and lemma.
\begin{definition}\label{def:log influence}\cite{DaganDDJ19}
	Let $\bm{x}=(x_1,\ldots, x_n)$ be a random variable over $\times_{i\in[n]} \Sigma_i$, and $P_{\bm{x}}$ denote its probability distribution. Assume $P_{\bm{x}}>0$ on all $\times_{i\in[n]} \Sigma_i$. For any $i\neq j\in [n]$, define the log influence between $i$ and $j$ as $$I_{i,j}^{\log}(\bm{x})={1\over 4}\max_{\substack{x_{-i-j}\in \Sigma_{-i-j}\\ x_i, x'_i \in \Sigma_i \\ x_j,x'_j\in \Sigma_j}}\log { P_{\bm{x}}(x_i x_j x_{-i-j}) P_{\bm{x}}(x'_i x'_j x_{-i-j}) \over P_{\bm{x}}(x'_i x_j x_{-i-j}) P_{\bm{x}}(x_i x'_j x_{-i-j})}.$$
\end{definition}

\begin{lemma}[Adapted from Lemma 5.2 of~\cite{DaganDDJ19arxiv}]
	Let $\bm{x}=(x_1,\ldots, x_n)$ be a random variable , for any $i,j\in [n]$, $\alpha_{i,j}\leq I_{i,j}^{\log}(\bm{x})$. \end{lemma}

We only need to prove that $I_{i,j}^{\log}(\vT)$ is no more than $\beta_{i,j}(\vT)$. Since {random variable} $\vT$ is generated by a MRF,
\begin{align*}
	I_{i,j}^{\log}(\vT)&={1\over 4}\max_{\substack{t_{-i-j}\in \Sigma_{-i-j}\\ t_i, t'_i \in \Sigma_i \\ t_j,t'_j\in \Sigma_j}}
 \sum_{\substack{e \in E :\\ i,j \in e}} \psi_e\left(\left(t_i,t_j,t_{-i-j}\right)_e\right) + \sum_{\substack{e \in E :\\ i,j \in e}} \psi_e\left(\left(t'_i,t'_j,t_{-i-j}\right)_e\right)\\
 &\qquad\qquad\qquad\qquad - \sum_{\substack{e \in E :\\ i,j \in e}} \psi_e\left(\left(t'_i,t_j,t_{-i-j}\right)_e\right) - \sum_{\substack{e \in E :\\ i,j \in e}} \psi_e\left(\left(t_i,t'_j,t_{-i-j}\right)_e\right),
\end{align*}
which is clearly no greater than $\beta_{i,j}(\vT)$.

Since for any $i,j\in[n]$ $\alpha_{i,j}\leq \beta_{i,j}(\vT)$, we have $\alpha(\vT) \leq \beta(\vT) $.
\end{proof}

\begin{theorem}\label{thm:XOS_high_temperature}
Let the buyer's type distribution $D$ be represented by a MRF. If the buyer's a XOS valuation and her type $\vT$ is in the high temperature regime, i.e., $\beta(\vT)<1$,
$$\rev(D)\leq  12\exp(12\beta(\vT)) \cdot \srev + \left(28 + \frac{16}{\sqrt{1-\beta(\vT)}}\right)\max\{\srev,\brev\}=O\left(\frac{\max\{\srev,\brev\}}{\sqrt{1-\beta(\vT)}}\right).$$
\end{theorem}
\begin{prevproof}{Theorem}{thm:XOS_high_temperature}
The statement follows from Theorem~\ref{thm:XOS Dobrushin} and Lemma~\ref{lemma:Dobruthin properties}.
\end{prevproof}


\notshow{
\begin{theorem}\label{lem:inaprx_dob}
For any positive real number $N$ and any choice of $0< \alpha < 1$,
there exists a type distribution over $2$ items generated by a MRF $D$ with $\alpha(\vT) = \alpha$ and finite inverse temperature,
such that for an additive buyer whose type is sampled from $D$,
$\frac{\rev(D)}{\max\{\brev(D),\srev(D)\}}>\frac{\alpha}{2}\cdot N$.
\end{theorem}
}
\section{Impossibility Results}~\label{sec:impossibility results}
In this section we present some of our impossibility results.
In Section~\ref{sec:inapprox_Dobrushin}, we show that the Dobrushin Uniqueness condition alone is insufficient to guarantee any multiplicative approximation of the optimal revenue using $\srev$ and $\brev$.
In Section~\ref{sec:LB_copies} we construct a MRF such that the optimal revenue in the COPIES setting is $\exp(\Delta)$ times larger than $\max\{\srev,\brev\}$.

\subsection{Inapproximability under Only the Dobrushin Uniqueness Condition}\label{sec:inapprox_Dobrushin}
Readers may wonder whether it is possible to prove an approximation ratio that only relies on either the spectral radius $\rho_d(\vT)$, the Dobrushin coefficient $\alpha(\vT)$, or the spectral gap of the Glauber dynamics $\gamma$, but independent of the maximum weighted degree $\Delta(\vT)$. We show that this is impossible. Indeed, we prove that for any $\alpha<1$, and any ratio $N$, there exists a MRF with $\rho_d(\vT)\leq\alpha(\vT)\leq\alpha$ such that the ratio between the optimal revenue and $\max\{\brev(D),\srev(D)\}$ is at least $\frac{\alpha}{2}\cdot N$. Our result is based on a modification of the Hart-Nisan construction~\cite{HARTN_2019}.

\begin{theorem}\label{lem:inaprx_dob}
For any positive real number $N$ and any choice of $0< \alpha < 1$,
there exists a type distribution $D$ over $2$ items generated by a MRF with Dobrushin coefficient $\alpha(\vT) = \alpha$ and finite inverse temperature,
such that for an additive buyer whose type is sampled from $D$,
$\frac{\rev(D)}{\max\{\brev(D),\srev(D)\}}>\frac{\alpha}{2}\cdot N$.
\end{theorem}

First we present the main building block of our construction. 

\begin{lemma}
\label{lemma:mix_dobrushin}
Let $D'$ be a correlated valuation distribution over $2$ items with Dobrushin coefficient $\alpha$.
Let $D$ be a product distribution that has the same marginal distributions as $D'$.
Then for any $0\leq \alpha' \leq 1$,
we consider the distribution $D'':=\alpha'\cdot D'+(1-\alpha')\cdot D$, that is, if we want to sample from $D''$, we can take a sample from $D'$ with probability $\alpha'$ and take a sample from $D$ with probability $1-\alpha'$.
Distribution $D''$ can be modeled as a MRF with finite inverse temperature
{such that $\Delta = \beta(\vT) \leq |\log((1-\alpha) p^2)|$,
where $p = \inf_{\vT \in \supp(D')} \Pr_{\vT' \sim D'}[\vT'=\vT]$}
and $D''$ has Dobrushin coefficient $\alpha' \cdot \alpha$.
Furthermore, $D''$ has the same marginal distribution as $D$ and $\rev(D'') \geq \alpha'\rev(D')$.
\end{lemma}

\begin{proof}
Assume that $D=D_1\times D_2$, where $D_1$ and $D_2$ are the marginals of  $D'$. 
Let $T_i = \supp(D_i)$
and $T= T_1 \times T_2$.
Let $A^{D'}=\{\alpha_{i,j}^{D'}\}_{i,j\in[2]}$ be the influence matrix of $D'$ and $A^{D''}=\{\alpha_{i,j}^{D''}\}_{i,j\in[2]}$ be the influence matrix of $D''$.
Note that the diagonal entries of $A^{D'}$ and of $A^{D''}$ are zero.
We have that:
$$\alpha_{i,j}^{D''}
=  \max_{  \substack{ t_j,t'_j \in T_j } }  
	d_{TV} \left( D''_{i \mid t_j},
	D''_{i \mid t_j'}\right) 
    = \alpha' \max_{ t_j,t'_j \in T_j  }  
	d_{TV} \left( D'_{i \mid t_j},
	D'_{i \mid t_j'}\right) 
	= \alpha' \alpha_{i,j}^{D'}$$
\notshow{
\begin{align*}
\alpha_{i,j}^{D''}
= & \max_{  \substack{ t_j,t'_j \in T_j } }  
	d_{TV} \left( D''_{i \mid t_j},
	D''_{i \mid t_j'}\right) \\
    = &\alpha' \max_{ t_j,t'_j \in T_j  }  
	d_{TV} \left( D'_{i \mid t_j},
	D'_{i \mid t_j'}\right) \\
	=& \alpha' \alpha_{i,j}^{D'}
\end{align*}
}

In the statement of the lemma,
we assumed that $\alpha = ||A^{D'}||_\infty=\max(\alpha^{D'}_{1,2},\alpha^{D'}_{2,1})$.
We can easily infer the following: $||A^{D''}||_\infty = \max(\alpha_{1,2}^{D''},\alpha_{2,1}^{D''}) = \alpha' \max(\alpha_{1,2}^{D'},\alpha_{2,1}^{D'}) = \alpha' \cdot \alpha$. This concludes the fact that $D''$ has $\alpha(\vT) =\alpha' \cdot \alpha$.

Now we prove that $D''$ can be modeled as a MRF with finite inverse temperature.
We consider a MRF with potential functions $\psi_1(t_1) = 1,\forall t_1 \in T_1$, $\psi_2(t_2) = 1,\forall t_2 \in T_2$ and $\psi_{1,2}(\vT) = \ln\left( \Pr_{\vT' \sim D''}(\vT' =\vT)\right)$. Since Distribution $D''$ samples from the product distribution $D$ with probability $1-\alpha$,
we have that for each $\vT \in T$, {$\Pr_{\vT' \sim D''}(\vT' =\vT) \geq  (1-\alpha)p^2$. This is true because with probability $1-\alpha$ we sample from the product distribution $D$,
and with probability at least $p^2$ we sample the type $\vT$}, 
therefore the MRF we just described has finite inverse temperature {and $\Delta = \beta(\vT) = \max_{\vT \in \supp(D')} |\log\left(\Pr_{\vT' \sim D'}[\vT'=\vT] \right)|\leq |\log\left((1-\alpha)p^2 \right)|$}. In the case where the distribution is over two items, $\beta(\vT)= \max_{\vT \in T} |\psi_{1,2}(\vT)|$.
Moreover we can easily verify that the MRF we just described has the same joint distribution as $D''$.
Therefore $D''$ can be modeled as a MRF with finite inverse temperature.

We now prove that $\rev(D'') \geq \alpha' \rev(D')$.
This is true because we can simply use the optimal mechanism that induces $\rev(D')$ on $D''$.
Since $D''$ take a sample from $D'$ with probability $\alpha'$,
we are guaranteed that this mechanism has revenue at least $\alpha'\rev(D')$ on $D''$.

The fact that $D''$ has the same marginal distributions as $D'$ follows from the sampling procedure.
\end{proof}

We also need the following important result from \cite{HARTN_2019}.

\begin{lemma}[Theorem~A from \cite{HARTN_2019}]
\label{lemma:bad_distribution}
For any positive number $N$, there exists a two item correlated distribution $D$,
such that for an additive buyer whose type is sampled from $D$, $\frac{\rev(D)}{\max\{\brev(D),\srev(D)\}}>N$.
\end{lemma}

Equipped with Lemma~\ref{lemma:mix_dobrushin} and~\ref{lemma:bad_distribution}, we are ready to prove Theorem~\ref{lem:inaprx_dob}.

\begin{prevproof}{Theorem}{lem:inaprx_dob}
Let $D'$ be the distribution over two items that is guaranteed to exist by Lemma~\ref{lemma:bad_distribution}. Since $D'$ is a two dimensional distribution, its Dobrushin coefficient is at most $1$.

Apply Lemma~\ref{lemma:mix_dobrushin} to $D'$ with parameter $\alpha'=\alpha$ to create another distribution $D$ which has the same marginals as $D'$ but with a Dobrushin coefficient at most $\alpha$. Moreover, $D$ can be expressed as a MRF with finite inverse temperature. Clearly, $\rev(D)\geq \alpha\cdot \rev(D')$, as one can simply achieve the RHS under distribution $D$ using the optimal mechanism designed for $D'$. Also, $\srev(D')=\srev(D)$ as the two distributions have the same marginals. Finally, $\brev(D')\leq 2\srev(D')$. Suppose $b$ is the optimal price for the bundle, then we can set the two items separately each at price $b/2$. Clearly, whenever the bundle is sold, at least one item is sold. To conclude, $\frac{\rev(D)}{\max\{\brev(D),\srev(D)\}}\geq \frac{\rev(D)}{2\srev(D)}\geq \frac{\alpha\cdot\rev(D')}{2\srev(D')}>\frac{\alpha}{2}\cdot N$.

\end{prevproof}


\subsection{Lower Bound for the Copies Setting}\label{sec:LB_copies}


In this section, we show that if the analysis uses the optimal revenue in the COPIES setting as part of the benchmark for the optimal revenue in the original setting (as in our analysis), the exponential dependence on the maximum weighted degree $\Delta$ in the approximation ratio is unavoidable. Note that we also showed that the approximation ratio must have polynomial dependence on $\Delta$ no matter what approach is used (Theorem~\ref{thm:poly_dependence_Delta}).

\begin{theorem}\label{thm:COPIES_example}
For any value of $n\in \mathbb{N}$ and $\beta \in \mathbb{R}_+$
there exists a type distribution $D$ over $n+1$ items, such that $D$ can be represented by a MRF with only pairwise potentials and maximum weighted degree $\Delta\leq \beta\cdot n$.
Moreover, for an additive or unit-demand buyer, the expected optimal revenue in the COPIES settings w.r.t. $D$ can be arbitrarily close to 
$\frac{1}{2}\exp(2\beta n)$,
while $\max\{\brev, \srev\} < 2$. 
\end{theorem}

\begin{prevproof}{Theorem}{thm:COPIES_example}
We construct the MRF in the following way.
The first item has support $T_1=\{2^0,2^1,2^2,\ldots,2^{k^n-1}\}$,
where $k\in \mathbb{N}$ is going to be defined later.
Let $\varepsilon_1,\ldots,\varepsilon_k$ be some tiny non-negative values, and  the support of the other items' distributions is $R=\{\varepsilon_i\}_{i\in[k]}$.
We consider the following node potential for the first item:
\begin{align*}
    \psi_1(2^i) =
    \begin{cases}
    \ln(\frac{1}{2^{i+1}}) \quad &\text{if $0 \leq i \leq k^n-2$} \\
    \ln(\frac{1}{2^{i}}) & \text{if $i = k^n-1$}
    \end{cases}
\end{align*}

The node potentials for the other items is:
$\psi_i (a) =\ln\left(\frac{1}{\exp(\beta)+ (k-1)\exp(-\beta)}\right)$ 
for all $i\in[2,n+1]$ and $a\in R$.

Note that $|T_1|=k^n$ and $|R^n| = k^n$,
therefore for each $t_1 \in T_1$,
we can map it to a unique $t_{-1} \in R^n$.
Formally, we consider a bijective function $c: T_1 \rightarrow R^n$.

\notshow{
\begin{align*}
    \psi_{1,j}(t_1,t_j) =
    \begin{cases}
    \beta  &\text{if $t_j =c(t_1)_j$} \\
    -\beta  &\text{if $t_j \neq c(t_1)_j$} \\
    \end{cases}
\end{align*}
}

We define pair-wise potentials
between the first item and the $j$-th item:
\begin{align*}
    \psi_{1,j}(2^i,\varepsilon_\ell) =
    \begin{cases}
    \beta  &\text{if $\varepsilon_\ell =c(2^i)_j$} \\
    -\beta  &\text{if $\varepsilon_\ell \neq c(2^i)_j$} \\
    \end{cases}
\end{align*}

It is easy to verify that $\Delta\leq \beta\cdot n$ for the constructed MRF.

Let $Z$ be the normalizing constant so that the MRF with potentials $\{\psi_i\}_{i \in[n+1]},\{\psi_{1,i}\}_{2 \leq i \leq n+1}$ is a valid distribution.
That is $Z = \sum_{t \in \supp(D)} \prod_{i \in [n+1]}\exp(\psi_i(t_i)) \prod_{2 \leq i \leq n+1} \exp(\psi_{1,i}(t_1,t_i))$.
For any $t_1 \in T_1$ we have that: $\Pr_{t'\sim D}\left[t_1' = t_1 \land t_{-1}'= c(t_1) \right] = 
    \frac{1}{Z}\exp(\psi_1(t_1)) \frac{\exp{(n \beta)}}{\left(\exp(\beta)+(k-1)\exp(-\beta)\right)^n}$.
\begin{align*}
    &\Pr_{t'\sim D}\left[t_1' = t_1 \land t_{-1}' \neq c(t_1) \right] \\ =
    &\frac{1}{Z}\exp(\psi_1(t_1)) \frac{1}{(\exp(\beta)+(k-1)\exp(-\beta))^n} \sum_{\substack{t_{-1} \in T_{-1}:\\ t_{-1} \neq c(t_1)}} \prod_{i \in [2,n+1]} \exp\left(\psi_{1,i}(t_1, \{t_{-1}\}_i)\right) \\
    = & \frac{1}{Z}\exp(\psi_1(t_1)) \frac{1}{(\exp(\beta)+(k-1)\exp(-\beta))^n} \sum_{i \in [1,n]} {n \choose i} \left(k-1 \right)^i\left(\exp(-\beta) \right)^i \left(\exp(\beta)\right)^{n-i} \\
    = & \frac{1}{Z} \exp(\psi_1(t_1)) \frac{(\exp(\beta)+(k-1)\exp(-\beta))^n - \exp(n \beta)}{(\exp(\beta)+(k-1)\exp(-\beta))^n} 
\end{align*}

Thus for any $t_1 \in T_1$, the marginal probability: $f_1(t_1) = \frac{1}{Z} \exp(\psi_1 (t_1))$. {Note that $Z = \sum_{t_1 \in T_1} \exp(\psi_1(t_1)) = \sum_{i=0}^{k^n-2}\frac{1}{2^{i+1}} + \frac{1}{2^{k^n-1}}= 1$ and $f_1(t_1) = \exp(\psi_1(t_1))$.}
Therefore the marginal distribution of the first item is an Equal Revenue Distribution, which means that the revenue of any posted price mechanism for the first item, cannot be more than $1$.
Moreover, if we choose $\varepsilon_1,\ldots, \varepsilon_k$ to be sufficiently small so that $\max_{x \in R} \leq \frac{1}{2n}$,
then any posted price mechanisms for the rest $n$ items has revenue less or equal than $\frac{1}{2}$.
Thus $\srev < 2$.

Now we consider the following Mechanism in the copies settings.
We first collect the values for all buyers except the first one $t_{-1}$, then let the first buyer decide whether she wants to purchase the item at price $c^{-1}(t_{-1})$. This is essentially Ronen's lookahead auction~\cite{Ronen_2001}.
A lower bound on the revenue of this mechanism in the COPIES settings is:
\begin{align*}
    \sum_{t_1 \in T_1} t_1 \Pr_{t' \sim D} \left[ c^{-1}(t'_{-1}) = t_1 \land t_1' =t_1 \right]
    = & \sum_{t_1 \in T_1} t_1 \exp(\psi_1(t_1)) \frac{\exp{(n \beta)}}{(\exp(\beta)+(k-1)\exp(-\beta))^n} \\
    = & \left(\frac{1}{1 + (k-1)\exp(-2\beta)} \right)^n\sum_{t_1 \in T_1} t_1 \exp(\psi_1(t_1))  \\
    \geq & \left(\frac{1}{1 + (k-1)\exp(-2\beta)} \right)^n\sum_{t_1 \in T_1}\frac{1}{2}  \\
    = & \frac{1}{2} \left(\frac{1}{1 + (k-1)\exp(-2\beta)} \right)^n|T_1| \\
    = & \frac{1}{2} \left(\frac{k}{1 + (k-1)\exp(-2\beta)} \right)^n 
\end{align*}

Where the last inequality follows from the definition of $\psi_1(t_1)$.

Note that if we fix $\beta$ and $n$, and let $k \rightarrow \infty$,
then $\lim_{k\rightarrow \infty}\left(\frac{k}{1 +(k-1)\exp(-2\beta)}\right)^n= \exp(2\beta n)$.

Therefore as $k\rightarrow \infty$,
the lower bound of the revenue of the proposed mechanisms becomes $\frac{\exp(2\beta n)}{2}$.
{Since we assumed that the value of the agent for each item except the first is less or equal than $\frac{1}{2n}$,
then the value of the agent for all but the first item is less or equal than $\frac{1}{2}$.
This implies that if the agent buys the whole bundle at price $p$,
then she also buys the first item at price $p-\frac{1}{2}$.
Let $\rev_1$ be the revenue of the posted price mechanism on the first item.
Since the marginal of the fist item is the Equal Revenue Distribution,
then $\rev_1\leq 1$.
Moreover by the argument described above,
we have that $\brev \leq \rev_1 + \frac{1}{2} < 2$.
Thus $\max\{\srev, \brev\} < 2$}.
\end{prevproof}

\paragraph{Acknowledgement.}
The authors would like to thank Jessie Huang for creating Figure~\ref{fig:car} and to thank Costis Daskalakis for helpful discussions in the early stages of the paper. 

\bibliographystyle{plain}
\bibliography{Yang}

\begin{thebibliography}{10}

\bibitem{Alaei11}
Saeed Alaei.
\newblock {Bayesian Combinatorial Auctions: Expanding Single Buyer Mechanisms
  to Many Buyers}.
\newblock In {\em the 52nd Annual IEEE Symposium on Foundations of Computer
  Science (FOCS)}, 2011.

\bibitem{AlaeiFHHM12}
Saeed Alaei, Hu~Fu, Nima Haghpanah, Jason Hartline, and Azarakhsh Malekian.
\newblock {Bayesian Optimal Auctions via Multi- to Single-agent Reduction}.
\newblock In {\em the 13th ACM Conference on Electronic Commerce (EC)}, 2012.

\bibitem{AlaeiFHH13}
Saeed Alaei, Hu~Fu, Nima Haghpanah, and Jason~D. Hartline.
\newblock The simple economics of approximately optimal auctions.
\newblock In {\em 54th Annual {IEEE} Symposium on Foundations of Computer
  Science, {FOCS} 2013, 26-29 October, 2013, Berkeley, CA, {USA}}, pages
  628--637. {IEEE} Computer Society, 2013.

\bibitem{BabaioffILW14}
Moshe Babaioff, Nicole Immorlica, Brendan Lucier, and S.~Matthew Weinberg.
\newblock A simple and approximately optimal mechanism for an additive buyer.
\newblock In {\em 55th {IEEE} Annual Symposium on Foundations of Computer
  Science, {FOCS} 2014, Philadelphia, PA, USA, October 18-21, 2014}, pages
  21--30, 2014.

\bibitem{bateni2015revenue}
MohammadHossein Bateni, Sina Dehghani, MohammadTaghi Hajiaghayi, and Saeed
  Seddighin.
\newblock Revenue maximization for selling multiple correlated items.
\newblock In {\em Algorithms-ESA 2015}, pages 95--105. Springer, 2015.

\bibitem{BhalgatGM13}
Anand Bhalgat, Sreenivas Gollapudi, and Kamesh Munagala.
\newblock Optimal auctions via the multiplicative weight method.
\newblock In Michael~J. Kearns, R.~Preston McAfee, and {\'{E}}va Tardos,
  editors, {\em {ACM} Conference on Electronic Commerce, {EC} '13,
  Philadelphia, PA, USA, June 16-20, 2013}, pages 73--90. {ACM}, 2013.

\bibitem{BoucheronLM00}
St{\'{e}}phane Boucheron, G{\'{a}}bor Lugosi, and Pascal Massart.
\newblock A sharp concentration inequality with applications.
\newblock {\em Random Struct. Algorithms}, 16(3):277--292, 2000.

\bibitem{briest2015pricing}
Patrick Briest, Shuchi Chawla, Robert Kleinberg, and S~Matthew Weinberg.
\newblock Pricing lotteries.
\newblock {\em Journal of Economic Theory}, 156:144--174, 2015.

\bibitem{BrustleCD20}
Johannes Brustle, Yang Cai, and Constantinos Daskalakis.
\newblock Multi-item mechanisms without item-independence: Learnability via
  robustness.
\newblock EC '20, page 715–761, New York, NY, USA, 2020. Association for
  Computing Machinery.

\bibitem{CaiD11b}
Yang Cai and Constantinos Daskalakis.
\newblock {Extreme-Value Theorems for Optimal Multidimensional Pricing}.
\newblock In {\em the 52nd Annual IEEE Symposium on Foundations of Computer
  Science (FOCS)}, 2011.

\bibitem{CaiDW12a}
Yang Cai, Constantinos Daskalakis, and S.~Matthew Weinberg.
\newblock {An Algorithmic Characterization of Multi-Dimensional Mechanisms}.
\newblock In {\em the 44th Annual ACM Symposium on Theory of Computing (STOC)},
  2012.

\bibitem{CaiDW12b}
Yang Cai, Constantinos Daskalakis, and S.~Matthew Weinberg.
\newblock {Optimal Multi-Dimensional Mechanism Design: Reducing Revenue to
  Welfare Maximization}.
\newblock In {\em the 53rd Annual IEEE Symposium on Foundations of Computer
  Science (FOCS)}, 2012.

\bibitem{CaiDW13b}
Yang Cai, Constantinos Daskalakis, and S.~Matthew Weinberg.
\newblock {Understanding Incentives: Mechanism Design becomes Algorithm
  Design}.
\newblock In {\em the 54th Annual IEEE Symposium on Foundations of Computer
  Science (FOCS)}, 2013.

\bibitem{CaiDW16}
Yang Cai, Nikhil~R. Devanur, and S.~Matthew Weinberg.
\newblock A duality based unified approach to bayesian mechanism design.
\newblock In {\em the 48th Annual ACM Symposium on Theory of Computing (STOC)},
  2016.

\bibitem{CaiH13}
Yang Cai and Zhiyi Huang.
\newblock {Simple and Nearly Optimal Multi-Item Auctions}.
\newblock In {\em the 24th Annual ACM-SIAM Symposium on Discrete Algorithms
  (SODA)}, 2013.

\bibitem{CaiZ16_arxiv}
Yang Cai and Mingfei Zhao.
\newblock Simple mechanisms for subadditive buyers via duality.
\newblock {\em CoRR}, abs/1611.06910, 2016.

\bibitem{CaiZ17}
Yang Cai and Mingfei Zhao.
\newblock Simple mechanisms for subadditive buyers via duality.
\newblock In {\em the 49th Annual ACM Symposium on Theory of Computing (STOC)},
  2017.

\bibitem{ChawlaHK07}
Shuchi Chawla, Jason~D. Hartline, and Robert~D. Kleinberg.
\newblock {Algorithmic Pricing via Virtual Valuations}.
\newblock In {\em the 8th ACM Conference on Electronic Commerce (EC)}, 2007.

\bibitem{ChawlaHMS10}
Shuchi Chawla, Jason~D. Hartline, David~L. Malec, and Balasubramanian Sivan.
\newblock {Multi-Parameter Mechanism Design and Sequential Posted Pricing}.
\newblock In {\em the 42nd ACM Symposium on Theory of Computing (STOC)}, 2010.

\bibitem{ChawlaMS10}
Shuchi Chawla, David~L. Malec, and Balasubramanian Sivan.
\newblock The power of randomness in bayesian optimal mechanism design.
\newblock In David~C. Parkes, Chrysanthos Dellarocas, and Moshe Tennenholtz,
  editors, {\em Proceedings 11th {ACM} Conference on Electronic Commerce
  (EC-2010), Cambridge, Massachusetts, USA, June 7-11, 2010}, pages 149--158.
  {ACM}, 2010.

\bibitem{ChawlaM16}
Shuchi Chawla and J.~Benjamin Miller.
\newblock Mechanism design for subadditive agents via an ex-ante relaxation.
\newblock In Vincent Conitzer, Dirk Bergemann, and Yiling Chen, editors, {\em
  Proceedings of the 2016 {ACM} Conference on Economics and Computation, {EC}
  '16, Maastricht, The Netherlands, July 24-28, 2016}, pages 579--596. {ACM},
  2016.

\bibitem{DaganDDJ19}
Yuval Dagan, Constantinos Daskalakis, Nishanth Dikkala, and Siddhartha Jayanti.
\newblock Learning from weakly dependent data under dobrushin's condition.
\newblock In Alina Beygelzimer and Daniel Hsu, editors, {\em Conference on
  Learning Theory, {COLT} 2019, 25-28 June 2019, Phoenix, AZ, {USA}}, volume~99
  of {\em Proceedings of Machine Learning Research}, pages 914--928. {PMLR},
  2019.

\bibitem{DaganDDJ19arxiv}
Yuval Dagan, Constantinos Daskalakis, Nishanth Dikkala, and Siddhartha Jayanti.
\newblock Learning from weakly dependent data under dobrushin's condition.
\newblock {\em CoRR}, abs/1906.09247, 2019.

\bibitem{DaskalakisDW15}
Constantinos Daskalakis, Nikhil~R. Devanur, and S.~Matthew Weinberg.
\newblock Revenue maximization and ex-post budget constraints.
\newblock In Tim Roughgarden, Michal Feldman, and Michael Schwarz, editors,
  {\em Proceedings of the Sixteenth {ACM} Conference on Economics and
  Computation, {EC} '15, Portland, OR, USA, June 15-19, 2015}, pages 433--447.
  {ACM}, 2015.

\bibitem{Daskalakis18}
Constantinos Daskalakis, Nishanth Dikkala, and Gautam Kamath.
\newblock Testing ising models.
\newblock In {\em Proceedings of the Twenty-Ninth Annual ACM-SIAM Symposium on
  Discrete Algorithms}, pages 1989--2007. Society for Industrial and Applied
  Mathematics, 2018.

\bibitem{daskalakis2019testing}
Constantinos Daskalakis, Nishanth Dikkala, and Gautam Kamath.
\newblock Testing ising models.
\newblock {\em IEEE Transactions on Information Theory}, 65(11):6829--6852,
  2019.

\bibitem{Dobruschin68}
PL~Dobruschin.
\newblock The description of a random field by means of conditional
  probabilities and conditions of its regularity.
\newblock {\em Theory of Probability \& Its Applications}, 13(2):197--224,
  1968.

\bibitem{DobrushinS87}
RL~Dobrushin and SB~Shlosman.
\newblock Completely analytical interactions: constructive description.
\newblock {\em Journal of Statistical Physics}, 46(5-6):983--1014, 1987.

\bibitem{Gheissari18}
Reza Gheissari, Eyal Lubetzky, Yuval Peres, et~al.
\newblock Concentration inequalities for polynomials of contracting ising
  models.
\newblock {\em Electronic Communications in Probability}, 23, 2018.

\bibitem{HARTN_2019}
Sergiu Hart and Noam Nisan.
\newblock Selling multiple correlated goods: Revenue maximization and menu-size
  complexity.
\newblock {\em Journal of Economic Theory}, 183:991 -- 1029, 2019.

\bibitem{Immorlica0W20}
Nicole Immorlica, Sahil Singla, and Bo~Waggoner.
\newblock Prophet inequalities with linear correlations and augmentations.
\newblock In P{\'{e}}ter Bir{\'{o}}, Jason Hartline, Michael Ostrovsky, and
  Ariel~D. Procaccia, editors, {\em {EC} '20: The 21st {ACM} Conference on
  Economics and Computation, Virtual Event, Hungary, July 13-17, 2020}, pages
  159--185. {ACM}, 2020.

\bibitem{kindermann1980markov}
Ross Kindermann and Laurie Snell.
\newblock {\em Markov random fields and their applications}, volume~1.
\newblock American Mathematical Society, 1980.

\bibitem{KleinbergW12}
Robert Kleinberg and S.~Matthew Weinberg.
\newblock {Matroid Prophet Inequalities}.
\newblock In {\em the 44th Annual ACM Symposium on Theory of Computing (STOC)},
  2012.

\bibitem{krengel1978semiamarts}
Ulrich Krengel and Louis Sucheston.
\newblock On semiamarts, amarts, and processes with finite value.
\newblock {\em Probability on Banach spaces}, 4:197--266, 1978.

\bibitem{lauritzen1996}
Steffen~L. Lauritzen.
\newblock {\em Graphical Models}.
\newblock Oxford University Press, 1996.

\bibitem{LPW09}
David~A. Levin, Yuval Peres, and Elizabeth~L. Wilmer.
\newblock Markov chains and mixing times.
\newblock {\em American Mathematical Soc., Providence}, 2009.

\bibitem{LiY13}
Xinye Li and Andrew Chi-Chih Yao.
\newblock On revenue maximization for selling multiple independently
  distributed items.
\newblock {\em Proceedings of the National Academy of Sciences},
  110(28):11232--11237, 2013.

\bibitem{psomas2019smoothed}
Alexandros Psomas, Ariel Schvartzman, and S~Matthew Weinberg.
\newblock Smoothed analysis of multi-item auctions with correlated values.
\newblock In {\em Proceedings of the 2019 ACM Conference on Economics and
  Computation}, pages 417--418, 2019.

\bibitem{randall2006slow}
Dana Randall.
\newblock Slow mixing of glauber dynamics via topological obstructions.
\newblock In {\em Symposium on Discrete Algorithms: Proceedings of the
  seventeenth annual ACM-SIAM symposium on Discrete algorithm}, volume~22,
  pages 870--879, 2006.

\bibitem{Ronen_2001}
Amir Ronen.
\newblock On approximating optimal auctions.
\newblock In Michael~P. Wellman and Yoav Shoham, editors, {\em Proceedings 3rd
  {ACM} Conference on Electronic Commerce (EC-2001), Tampa, Florida, USA,
  October 14-17, 2001}, pages 11--17. {ACM}, 2001.

\bibitem{RubinsteinW15}
Aviad Rubinstein and S.~Matthew Weinberg.
\newblock Simple mechanisms for a subadditive buyer and applications to revenue
  monotonicity.
\newblock In {\em Proceedings of the Sixteenth {ACM} Conference on Economics
  and Computation, {EC} '15, Portland, OR, USA, June 15-19, 2015}, pages
  377--394, 2015.

\bibitem{samuel1984comparison}
Ester Samuel-Cahn et~al.
\newblock Comparison of threshold stop rules and maximum for independent
  nonnegative random variables.
\newblock {\em the Annals of Probability}, 12(4):1213--1216, 1984.

\bibitem{StroockZ92}
Daniel~W Stroock and Boguslaw Zegarlinski.
\newblock The logarithmic sobolev inequality for discrete spin systems on a
  lattice.
\newblock {\em Communications in Mathematical Physics}, 149(1):175--193, 1992.

\bibitem{wu2006}
Liming Wu.
\newblock Poincaré and transportation inequalities for gibbs measures under
  the dobrushin uniqueness condition.
\newblock {\em Ann. Probab.}, 34(5):1960--1989, 09 2006.

\bibitem{Yao15}
Andrew~Chi{-}Chih Yao.
\newblock An n-to-1 bidder reduction for multi-item auctions and its
  applications.
\newblock In {\em SODA}, 2015.

\end{thebibliography}

\appendix

\section{Missing Details of the Duality-base Benchmark}\label{appx:dual_add_ud}

We provide the necessary information to derive the benchmark for XOS valuations in Appendix~\ref{sec:appx_XOS_duality}.
Deriving a benchmark for a constrained additive valuation is a simpler task.
We summarize the benchmark for a constrained additive valuation in Definition~\ref{def:benchmark constrained additive}.

\begin{definition}\label{def:benchmark constrained additive}

The duality framework provide the following bound:

\begin{align*}
\rev(M,v, D)\leq &\sum_{\vT \in T} \sum_{i\in [n]} f(\vT) \cdot \pi_i (\vT) \cdot \left(    t_i \cdot  \mathds{1} \left[ \vT  \notin R_i \right] + 	\phi_i(t_i\mid t_{-i}) \cdot	 \mathds{1} \left[ \vT  \in R_i \right]		\right) \\
&=\sum_{\vT \in T} \sum_{i\in[n]} f(\vT) \cdot \pi _i (\vT) \cdot	\phi_i(t_i\mid t_{-i}) \cdot	 \mathds{1} \left[ \vT  \in R_i \right] \quad (\single)\\
& \qquad\qquad+\sum_{\vT \in T}\sum_{i\in[n]} f(\vT) \cdot  \pi _i (\vT) \cdot t_i \cdot  \mathds{1} \left[ \vT  \notin R_i \right] (\nf) \\
\end{align*}

\end{definition}

\notshow{
\begin{definition}
For the  (Non-Favorite) we have that:

\begin{eqnarray*}
 &          \sum_{t \in T} \sum_j \Pr[t] \cdot  \pi_j (t) \cdot t_j \cdot  \mathds{1} \left[ t  \notin R_j \right]  \\
 \leq &     \sum_{t \in T} \sum_j \Pr[t]  \cdot t_j  \cdot  \mathds{1} \left[ t  \notin R_j \right] \\
 = &        \sum_{t \in T} \sum_j \Pr[t_j] \cdot \Pr[t_{-j}\mid t_j]  \cdot t_j  \cdot  \mathds{1} \left[ t  \notin R_j \right] \\
 = & \sum_j \sum_{t_j \in T_j} \Pr[t_j] \cdot t_j \sum_{t_{-j}\in T_{-j}}  Pr\left[t_{-j} \mid t_j \right] \cdot  \mathds{1} \left[ t  \notin R_j \right] \\
 \leq &
 \sum_j \sum_{t_j > r} \Pr[t_j] \cdot t_j \sum_{t_{-j}\in T_{-j}}  Pr\left[t_{-j} \mid t_j \right] \cdot  \mathds{1} \left[ t  \notin R_j \right] \\
 & +
 \sum_j \sum_{t_j \leq r} \Pr[t_j] \cdot t_j \sum_{t_{-j}\in T_{-j}}  Pr\left[t_{-j} \mid t_j \right] \cdot  \mathds{1} \left[ t  \notin R_j \right] \\
 \leq &
  \sum_j \sum_{t_j > r} \Pr[t_j] \cdot t_j Pr\left[ t  \notin R_j \mid t_j \right] (TAIL) \\
 & +
 \sum_j \sum_{t_j \leq r} \Pr[t_j] \cdot t_j (CORE) \\
\end{eqnarray*}

\end{definition}
}

\begin{lemma}\label{lem:nonfavorite to core and tail}
We can bound $\nf$ by $\core$ and $\tail$. More specifically,

\begin{align*}
\nf\leq &	\sum_{i\in[n]}\sum_{t_i > r} f_i(t_i) \cdot t_i\cdot \Pr_{\vT'\sim D}\left[ \vT'  \notin R_i \mid t'_i=t_i \right] \quad(\tail) \\
 & \qquad\qquad+\sum_{i\in[n]} \sum_{t_i \leq r} f_i(t_i) \cdot t_i \quad(\core)
\end{align*}
\end{lemma}
\begin{proof}
	\begin{align*}
 &        \sum_{\vT \in T}\sum_{i\in[n]} f(\vT) \cdot  \pi _i (\vT) \cdot t_i \cdot  \mathds{1} \left[ \vT  \notin R_i \right] \\
 \leq &   \sum_{\vT \in T}\sum_{i\in[n]} f(\vT) \cdot t_i \cdot  \mathds{1} \left[ \vT  \notin R_i \right]\\
 = &        \sum_{\vT \in T} \sum_{i\in[n]} f_i(t_i) \cdot f_{-i}(t_{-i}\mid t_i)  \cdot t_i  \cdot  \mathds{1} \left[ \vT  \notin R_i \right] \\
 = & \sum_{i\in[n]} \sum_{t_i \in T_i} f_i(t_i) \cdot t_i \sum_{t_{-i}\in T_{-i}}  f_{-i}(t_{-i}\mid t_i) \cdot  \mathds{1} \left[ \vT  \notin R_i \right] \\
 = &
 \sum_{i\in[n]} \sum_{t_i > r}f_i(t_i) \cdot t_i \sum_{t_{-i}\in T_{-i}}  f_{-i}(t_{-i}\mid t_i) \cdot  \mathds{1} \left[ \vT  \notin R_i \right]\\
 & \qquad\qquad +
 \sum_{i\in [n]} \sum_{t_i \leq r} f_i(t_i) \cdot t_i \sum_{t_{-i}\in T_{-i}}  f_{-i}(t_{-i}\mid t_i) \cdot  \mathds{1} \left[ \vT  \notin R_i \right]\\
 \leq &
 \sum_{i\in[n]}\sum_{t_i > r} f_i(t_i) \cdot t_i\cdot \Pr_{\vT'\sim D}\left[ \vT'  \notin R_i \mid t'_i=t_i \right] \quad(\tail) \\
 & \qquad\qquad+\sum_{i\in[n]} \sum_{t_i \leq r} f_i(t_i) \cdot t_i \quad(\core)
\end{align*}
\end{proof}
\section{Missing Proofs from Section~\ref{sec:MRF unit-demand}}\label{sec:appx UD}

\begin{prevproof}{Lemma}{lemma:single_ronen}

For the term \single~ we have that:

\begin{align}
&\sum_{\vT \in T} \sum_{i\in[n]} f(\vT) \cdot \pi _i (\vT) \cdot	\phi_i(t_i\mid t_{-i}) \cdot	 \mathds{1} \left[ \vT  \in R_i \right]  \nonumber \\
=&\sum_{i\in [n]} \sum_{t_{-i} \in T_{-i}} f_{-i}(t_{-i})\cdot\sum_{t_i \in T_i} f_i(t_i\mid t_{-i}) \cdot \pi _i (\vT) \cdot	\phi_i(t_i\mid t_{-i}) \cdot	 \mathds{1} \left[ \vT \in R_i \right] \nonumber \\
\leq &\sum_{i\in [n]} \sum_{t_{-i} \in T_{-i}} f_{-i}(t_{-i})\cdot\sum_{t_i \in T_i} f_i(t_i\mid t_{-i}) \cdot \phi_i(t_i\mid t_{-i})^+ \cdot	 \mathds{1} \left[ \vT \in R_i \right] \nonumber\\
= & \sum_{i\in [n]} \sum_{t_{-i} \in T_{-i}} f_{-i}(t_{-i})\cdot\sum_{t_i : ~(t_i,t_{-i})\in R_i} f_i(t_i\mid t_{-i}) \cdot \phi_i(t_i\mid t_{-i})^+  \label{eq:single_UD} 
\end{align}

According to Definition~\ref{def:ironed virtual values}, $\sum_{t_i : ~(t_i,t_{-i})\in R_i} f_i(t_i\mid t_{-i}) \cdot \phi_i(t_i\mid t_{-i})^+=\max_{p\geq \max_{j\neq i} t_j} p\cdot (1- F_i(p\mid t_{-i}))$, which is exactly the revenue of Ronen's auction when bidder $i$ has the highest value and the other bidders have value $t_{-i}$. This completes the proof that $\single \leq \rcopies$.

As the buyer has unit-demand valuation, the \nf~ term is at most the revenue of the second price auction in the \copies ~settings. 
The second price auction can be viewed as a special case Ronen's mechanism, in which the price is always set to be the same as  the second highest bid. Hence, $\nf \leq \rcopies$.
\notshow{We remind the readers that Ronen's auction first identifies the highest bidder,
and offers a take it or leave it price to the highest bidder (condition on the price not lower than the second-highest bid) so that the revenue is maximized. 
The second price auction can be viewed as a special case of this type of mechanism, in which the price is always set to be the same as  the second highest bid. Hence, $\nf \leq \rcopies$.}

Recall that:
\begin{align*}
\rcopies= \sum_{i\in [n]} \sum_{t_{-i} \in T_{-i}} f_{-i}(t_{-i})\cdot \max_{p\geq \max_{j\neq i} t_j} p\cdot (1- F_i(p\mid t_{-i}))  \end{align*}

Due to Lemma~\ref{lemma:tv_MRF}, 
\begin{align*}
\rcopies &\leq \sum_{i\in [n]} \sum_{t_{-i} \in T_{-i}} f_{-i}(t_{-i})\cdot \exp(4\Delta(\vT))\max_{p\geq \max_{j\neq i} t_j} p\cdot (1- F_i(p)) \\
& \leq \sum_{i\in [n]} \sum_{t_{-i} \in T_{-i}} f_{-i}(t_{-i})\cdot \exp(4\Delta(\vT)) \sum_{t_i : ~(t_i,t_{-i})\in R_i} f_i(t_i) \cdot \phi_i(t_i)^+\\
&\leq \exp(8\Delta(\vT)) \sum_{i\in [n]} \sum_{\vT \in R_i} f(\vT) \cdot \phi_i(t_i)^+\\
&\leq \exp(8\Delta(\vT)) \E_{\vT}\left[\max_{i\in[n]} \phi_i(t_i)^+\right]
\end{align*}
\notshow{
Assume bidder $i^*$ has the highest bid.
By Myerson Lemma $\sum_{t_j \in T_j  : \phi_j^c(t)>0} \Pr[t_j\mid t_{-j}]  \cdot \phi_j^c(t) \mathds{1}[t \in R_j]$ is equal to a posted price mechanism on bidder ${i^*}$,
when the distribution from bidder ${i^*}$ is conditioned on the other values.
Let $S_{i^*,t_{-{i^*}}}$ be that posted price,
then using Lemma~\ref{lemma:tv_MRF} we have that:

\begin{align*}
    \sum_{t_j \in T_j  : \phi_j^c(t)>0} \Pr[t_j\mid t_{-j}]  \cdot \phi_j^c(t) \mathds{1}[t \in R_j] = &S_{i^*,t_{-{i^*}}} Pr [t_{i^*} \geq S_{i^*,t_{-{i^*}}} \mid t_{-i^*} ] \\
    \leq & 
    \exp(4\Delta)\cdot S_{i^*,t_{-{i^*}}} Pr [t_{i^*} \geq S_{i^*,t_{-{i^*}}} ]
\end{align*}

Again by Myerson Lemma we have that:

\begin{align*}
    S_{i^*,t_{-{i^*}}} Pr [t_{i^*} \geq S_{i^*,t_{-{i^*}}} ] = \sum_{\substack{t_j \in T_j  \\ \phi_j(t_j)>\phi_j(S_{i^*,t_{-{i^*}}})}} \Pr[t_j]  \cdot \phi_j(t_j)
\end{align*}

Again using Lemma~\ref{lemma:tv_MRF} we have that:

\begin{align*}
\sum_{\substack{t_j \in T_j  \\ \phi_j(t_j)>\phi_j(S_{i^*,t_{-{i^*}}})}} \Pr[t_j]  \cdot \phi_j(t_j)
\leq \exp(4\Delta) \sum_{\substack{t_j \in T_j  \\ \phi_j(t_j)>\phi_j(S_{i^*,t_{-{i^*}}})}} \Pr[t_j \mid t_{-j}]  \cdot \phi_j(t_j)
\end{align*}

Combining the previews results we get that:

\begin{align*}
    Ronen^{COPIES}\leq \exp(8\Delta) \cdot \sum_j \sum_{t_{-j} \in T_{-j}}  \Pr[t_{-j}] \cdot \sum_{t_j \in T_j  : \phi_j^c(t)>0} \Pr[t_j\mid t_{-j}]  \cdot \phi_j^c(t) \mathds{1}[t \in R_j]
\end{align*}

    We can see that the above inequality each time chooses at most one of the virtual valuation  of the items with respect to the virtual valuation function as if the items are independent.

    Therefore the revenue by Ronen can upper bounded by the $\exp(8\Delta)E[\max_i \phi(x_i)]$
}
\end{prevproof}

\begin{prevproof}{Theorem}{thm:MRF unit-demand}
Let $\tau^*=\median_{\vT}(\max_{i\in [n]} \phi_i(t_i)^+)$
. We let price $p_i=\min\{p\in T_i: \phi_i(p)^+\geq \tau^*\}$. We first provide a lower bound of the revenue of the posted-price mechanism under the prices $\{p_i\}_{i\in [n]}$. 

For each item $i\in [n]$, let $\EE_i$ denote the event $\{\vT\in T: t_i\geq p_i\}$ and $\EE'_i$ denote the event $\{\vT\in T: t_j<p_j,~\forall j\neq i\}$. Clearly, the buyer buys item $i$ in event $\EE_i\cap \EE'_i$, so the revenue of the posted-price mechanism is at least 
$$\sum_{i\in[n]} p_i\cdot \Pr_{\vT\sim D}\left[\EE_i\cap \EE'_i\right]\geq \exp(-4\Delta(\vT))\cdot \sum_{i\in[n]} p_i\cdot \Pr_{\vT\sim D}\left[\EE_i\right ] \Pr_{\vT\sim D}\left[ \EE'_i\right].$$ 
Note that $p_i\cdot \Pr_{\vT\sim D}\left[\EE_i\right ] =\sum_{t_i\in T_i} f_i(t_i)\cdot\phi_i(t_i)^+\cdot \ind\left[\phi_i(t_i)^+\geq \tau^*\right]=\tau^*\cdot \Pr_{t_i\sim D_i}[\phi_i(t_i)^+\geq \tau^*]+\E_{t_i\sim D_i}[(\phi_i(t_i)^+-\tau^*)^+]$ and $\Pr_{\vT\sim D}\left[ \EE'_i\right]\geq 1/2$. Hence, the RHS of the inequality above is lower bounded by 
\begin{align*}
	{\exp(-4\Delta(\vT))\over 2}\cdot \sum_{i\in[n]}\tau^*\cdot \Pr_{t_i\sim D_i}[\phi_i(t_i)^+\geq \tau^*]&+\E_{t_i\sim D_i}[(\phi_i(t_i)^+-\tau^*)^+]\\
	&\geq {\exp(-4\Delta(\vT))\over 2}\cdot \left({\tau^*\over 2}+\sum_{i\in[n]}\E_{t_i\sim D_i}[(\phi_i(t_i)^+-\tau^*)^+]\right)
\end{align*}

The inequality is due to the union bound. By Lemma~\ref{lemma:MRF_prophet}, the lower bound is at least ${\exp(-4\Delta(\vT))\over 4}\cdot \E_{\vT}\left[\max_{i\in[n]} \phi_i(t_i)^+\right]$. Combining this conclusion with Lemma~\ref{lemma:single_ronen}, the revenue of the posted-price mechanism is at least $\rev(D)\over 8\exp(12\Delta(\vT))$.

\end{prevproof}
\section{Missing Proofs from Section~\ref{sec:MRF additive}}\label{sec:appx_additive}

\begin{prevproof}{Lemma}{lem:bounding single and tail additive}

\begin{align*}
\single = &\sum_{\vT \in T} \sum_{i\in[n]} f(\vT) \cdot \pi _i (\vT) \cdot	\phi_i(t_i\mid t_{-i}) \cdot	 \mathds{1} \left[ \vT  \in R_i \right]  \\
\leq &\sum_{i\in [n]}\sum_{t_{-i}\in T_{-i}}f_{-i}(t_{-i})\cdot\sum_{\substack{t_i:~ (t_i,t_{-i})\in R_i}} f_i(t_i\mid t_{-i}) \cdot \phi_i(t_i\mid t_{-i})^+\\
= &\sum_{i\in [n]}\sum_{t_{-i}\in T_{-i}}f_{-i}(t_{-i})\cdot\max_{p\geq \max_{j\neq i} t_j} p\cdot (1- F_i(p\mid t_{-i}))
 \\
 \leq &\exp(4\Delta(\vT))\cdot\sum_{i\in [n]}\sum_{t_{-i}\in T_{-i}}f_{-i}(t_{-i})\cdot\max_{p\geq \max_{j\neq i} t_j} p\cdot (1- F_i(p))
 \\
  \leq &\exp(4\Delta(\vT))\cdot\sum_{i\in [n]}\sum_{t_{-i}\in T_{-i}}f_{-i}(t_{-i})\cdot r_i
 \\
\leq &\exp(4\Delta(\vT))\cdot r
\end{align*}

The first equality is due to the definition of $\phi_i(t_i\mid t_{-i})$ (Definition~\ref{def:ironed virtual values}). The second inequality follows from Lemma~\ref{lemma:tv_MRF}. The third and last inequalities follow from the definition of $r_i$ and $r$.
\notshow{ 
By Myerson Lemma, the optimal mechanism is to sell the item at a posted price when the buyer has positive virtual welfare.
If we sample the value for the item from $t_j\mid t_{-j}$, then we have for some value $S_{j,t_{-j}}$ using Lemma~\ref{lemma:tv_MRF} we have that:

\begin{align*}
\sum_{\substack{t_j \in T_j \\ \phi_j^c(t)>0}} \Pr[t_j\mid t_{-j}]  \cdot \phi_j^c(t) \cdot	 \mathds{1} \left[ t  \in R_j \right] 
= & S_{j,t_{-j}}\cdot Pr\left[ t \geq S_{j,t_{-j}}\mid t_{-j} \right]  \\
\leq & \exp(4\Delta) S_{j,t_{-j}}\cdot Pr\left[ t \geq S_{j,t_{-j}}\right]  \\
\leq & \exp(4\Delta) r_j \Pr[t \geq r_j]
\end{align*}

Therefore:

\begin{align*}
(Single) \leq &\sum_j \sum_{t_{-j} \in T_{-j}}  f(t_{-j}) \sum_{\substack{t_j \in T_j  \\ \phi_j^c(t)>0}} \Pr[t_j\mid t_{-j}]  \cdot \phi_j^c(t) \cdot	 \mathds{1} \left[ t  \in R_j \right] \\
\leq & \sum_j \sum_{t_{-j} \in T_{-j}}  \Pr[t_{-j}] \exp(4\Delta) r_j \Pr[t \geq r_j] \\
=  & \exp(4\Delta) \sum_j r_j \Pr[t \geq r_j] \sum_{t_{-j} \in T_{-j}}  \Pr[t_{-j}] \\
=  & \exp(4\Delta) \sum_j r_j \Pr[t \geq r_j]  \\
= & \exp(4\Delta) r 
\end{align*}

}

Similarly, we can bound the term \tail $= \sum_{i\in[n]}\sum_{t_i > r} f_i(t_i) \cdot t_i\cdot \Pr_{\vT'\sim D}\left[ \vT'  \notin R_i \mid t'_i=t_i \right]$.

First, note that $\Pr_{\vT'\sim D}\left[ \vT'  \notin R_i \mid t'_i=t_i \right] \leq \Pr_{\vT'\sim D}\left[ \exists k \neq i : t'_k \geq t_i \mid t'_i= t_i \right]$.
Therefore

\begin{align*}
     \tail \leq & \sum_{i\in[n]}\sum_{t_i > r} f_i(t_i) \cdot t_i\cdot \Pr_{\vT'\sim D}\left[ \exists k \neq i : t'_k \geq t_i \mid t'_i= t_i \right]\\
     \leq & \exp(4\Delta(\vT))\cdot \sum_{i\in[n]}\sum_{t_i > r} f_i(t_i) \cdot t_i\cdot \Pr_{t'_{-i}\sim D_{-i}}\left[ \exists k \neq i : t'_k \geq t_i \right]\\
     \leq &  \exp(4\Delta(\vT))\cdot \sum_{i\in[n]}\sum_{t_i > r} f_i(t_i) \cdot t_i\cdot \left(\sum_{k\neq i}\Pr_{t'_{k}\sim D_{k}}\left[t'_k \geq t_i \right]\right)\\
     \leq &  \exp(4\Delta(\vT))\cdot \sum_{i\in[n]}\sum_{t_i > r} f_i(t_i) \cdot \sum_{k\neq i} r_k\\
        \leq &  \exp(4\Delta(\vT))\cdot \sum_{i\in[n]}r\cdot \sum_{t_i > r} f_i(t_i) \\
        \leq &  \exp(4\Delta(\vT))\cdot \sum_{i\in[n]}r_i\\
        =& \exp(4\Delta(\vT))\cdot r
     \end{align*}

The second inequality is due to Lemma~\ref{lemma:tv_MRF}. The third inequality follows from the union bound. The fourth and sixth inequalities hold because $r_k\geq t_i\cdot \Pr_{t'_{k}\sim D_{k}}\left[t'_k \geq t_i \right]$ and $r_i\geq r\cdot (1-F_i(r))$. 
\notshow{
Using Lemma~\ref{lemma:MRF_tv_product} with condition function $C(T_{-i})=\mathds{1}[\exists k \neq j : t_k \geq t_j]$ we have that:

\begin{align*}
    Pr\left[ \exists k \neq j : t_k \geq t_j \mid t_j \right] \leq \exp(2\Delta) \Pr[t_j] Pr\left[ \exists k \neq j : t_k \geq t_j \right]
\end{align*}

This implies that:

\begin{align*}
 & \sum_j \sum_{t_j > r} \Pr[t_j] \cdot t_j Pr\left[ \exists k \neq j : t_k \geq t_j \mid  t_j \right] \\
 \leq &
\sum_j \sum_{t_j > r} \Pr[t_j] \cdot t_j \exp(2\Delta) Pr\left[ \exists k \neq j : t_k \geq t_j \right] \\ 
= &
\exp(2\Delta) \sum_j \sum_{t_j > r} \Pr[t_j]\cdot   t_j Pr\left[ \exists k \neq j : t_k \geq t_j \right] \\
\end{align*}

If we set consider the mechanism that posts price $t_j$ at each item expect item $j$,
then its expected revenue is at least $t_j \Pr\left[ \exists k \neq j : t_k \geq t_j \right]$,
which is at most $r$.
This implies that:

\begin{align*}
& \exp(2\Delta) \sum_j \sum_{t_j > r} \Pr[t_j]\cdot   t_j \Pr\left[ \exists k \neq j : t_k \geq t_j \right] \\
 \leq &
 \exp(2\Delta)\sum_j \sum_{t_j > r} \Pr[t_j] r  \\
= &
 \exp(2\Delta) \sum_j \Pr_{t\sim T_j}[t>r] r  \\
 \leq &
 \exp(2\Delta) \sum_j r_i  \\
 = &
 \exp(2\Delta) r 
\end{align*}
}

\end{prevproof}

\begin{prevproof}{Lemma}{lemma:cov}
We have that: $$\sum_{\substack{c_i \leq r \\ c_j \leq r}} c_i c_j \Pr_{\vT\sim D}[t_i=c_i \land t_j=c_j] \leq   \exp(4\Delta(t))\cdot \sum_{\substack{c_i \leq r \\ c_j \leq r}} c_i f_i(c_i) \cdot c_j f_j(c_j) 
        = \exp(4\Delta(t)) \E[C_i] \E[C_j]$$

The inequality follows from Lemma~\ref{lemma:tv_MRF}. Therefore, $\Cov[C_i,C_j] \leq  (\exp(4\Delta(t))-1) \E[C_i] \E[C_j]$.

Note that $\Var[C] =  \sum_{i\in [n]} \Var[C_i^2] + \sum_{i\neq j}\Cov(C_i,C_j) \leq \sum_{i\in [n]} \E[C_i^2] + \sum_{i\neq j}\Cov(C_i,C_j)$.

Using Lemma~9 from \cite{CaiDW16}, we can bound $\sum_{i\in [n]} E[C_i^2]$ by $2r^2$.
Hence,
\begin{align*}
    \Var[C] \leq & 2r^2 + (\exp(4\Delta(\vT))-1)\sum_{i\neq j}  \E[C_i]\E[C_j] \\
    \leq & 2r^2 + (\exp(4\Delta(\vT))-1)\left( \sum_{i\in [n]}  \E[C_i] \right)^2 \\
    = & 2r^2 + (\exp(4\Delta(\vT))-1)\E[C]^2 \\
\end{align*}

\end{prevproof}

\begin{prevproof}{Theorem}{thm:MRF additive}
{
First we present the Paley-Zygmund inequality.
For a non-negative random variable $X$, Paley-Zygmund inequality implies that for $\theta\in[0, 1]$ we have that:
\begin{align*}
    \Pr\left[ X > \theta \E[X] \right] \geq (1-\theta)^2 \frac{1}{1 + \frac{\Var[X]}{\E[X]^2}}
\end{align*}
}
	By the Paley-Zygmund inequality and Lemma~\ref{lemma:cov}, we derive the following inequality:
	\begin{equation}\label{eq:additive core}
		\Pr\left[C\geq {\E[C]\over 2}\right]\geq {1\over 4}\cdot {1\over 1+\Var[C]/\E[C]^2}\geq {1\over 4\left(\exp(4\Delta(\vT))+2r^2/\E[C]^2\right)}.	\end{equation}
	
	If $\E[C]\leq \sqrt{2}r$, then according to Lemma~\ref{lem:bounding single and tail additive}, $$\rev(D)\leq \core+\tail+\single\leq \left(2\exp(4\Delta(\vT))+\sqrt{2}\right)\cdot \srev.$$
	Otherwise, Equation~\eqref{eq:additive core} implies that $$\Pr_{\vT\sim D}\left[\sum_{i\in [n]}t_i\geq {\E[C]\over 2}\right]	\geq \Pr\left[C\geq {\E[C]\over 2}\right]\geq {1\over 4\left(\exp(4\Delta(\vT))+1\right)}.$$
Therefore, if we sell the grand bundle at price $\frac{\E[C]}{2}=\frac{\core}{2}$, it will be sold with probability at least ${1\over 4\left(\exp(4\Delta(\vT))+1\right)}$. Thus $8\left(\exp(4\Delta(t))+1\right)\cdot\brev\geq \core$.
	
	Combining everything, we have  $\left(2\exp(4\Delta(\vT))+\sqrt{2}\right)\cdot \srev+8\left(\exp(4\Delta(t))+1\right)\cdot\brev\geq \rev(D).$
\end{prevproof}

\section{Missing Proofs from Section~\ref{sec:MRF XOS}}\label{sec:appx_XOS_single_tail}

\begin{prevproof}{Lemma}{lemma:XOS_single_tail}

We use $r$ to denote $\srev$.
We remind the readers that we use $r$ to denote $\srev$, which is the revenue of the optimal posted price auction, in which we only allow the buyer to purchase at most one item.

We note that \single~ term in the XOS case,
is the same as the the \single~ term in the Unit-Demand case in Section~\ref{sec:MRF unit-demand},
if we consider that the buyer has valuation $V_i(\bm{t})$ for the $i$-th item. In section~\ref{sec:MRF unit-demand}, using Lemma~\ref{lemma:single_ronen} we proved that $\single \leq \rcopies $.
Therefore it is enough to prove that there exists a posted price mechanism that allows the buyer to only pick her favorite item such that its revenue is at least $\frac{\rcopies}{4\exp(12\Delta(\vT))}$. A corollary of Theorem~\ref{thm:MRF unit-demand} is that there exists a posted price Mechanism $M_p$ such that $\rcopies \leq 4\exp(12\Delta(\vT)\rev(M_p)$, which concludes the proof for the term \single.

Next, we consider the term \tail. We remind the readers that $2r$ is a cutoff we use to separate the $(\core)$ and the $(\tail)$ term. The reason we chose this specific value is that we can bound the sum of the marginal probability  that any item has value greater or equal than $2r$.

\begin{lemma}
\label{lemma:XOS_bound_prob}
We have that:
\begin{align*}
    \sum_{i \in [n]} \Pr_{t_i \sim D_i}[V_i(t_i) \geq 2r] \leq \exp(4\Delta(\bm{t}))
\end{align*}
\end{lemma}

\begin{proof}
We can lower bound $\Pr_{\bm{t} \sim D}[\exists i : {V_i(t_i)} \geq 2r]$ as the sum of the following disjoint events:

\begin{align*}
    &\Pr_{\bm{t} \sim D}[\exists i : {V_i(t_i)} \geq 2r]  \\
    \geq & \sum_{i \in [n]} \Pr_{\bm{t} \sim D}[{V_i(t_i)} \geq 2r \land \max_{j \neq i}  \{{V_j(t_j)}\} < 2r] 
\end{align*}

Using Lemma~\ref{lemma:tv_MRF} with sets $\EE = \{ t_i \in T_i : V_i(t_i) \geq 2r\}$ and $\EE' = \{ t_{-i} \in T_{-i} : \max_{j \neq i} \{{V_j(t_j)}\} < 2r\}$, we have that:

\begin{align*}
    & \sum_{i \in [n]} \Pr_{\bm{t} \sim D}[{V_i(t_i)} \geq 2r \land \max_{j \neq i}  \{{V_j(t_j)}\} < 2r]   \\
    \geq & \sum_{i \in [n]} \Pr_{t_i \sim D_i}[{V_i(t_i)} \geq 2r] \exp(-4\Delta(\bm{t}))  \Pr_{t \sim D}[\max_{j \neq i}  \{{V_j(t_j)}\} < 2r]\\
    \geq & \exp(-4\Delta(\bm{t}))  \Pr_{t \sim D}[ \max_{j}  \{{V_j(t_j)}\} < 2r] \sum_{i \in [n]} \Pr_{t_i \sim D_i}[{V_i(t_i)} \geq 2r]
\end{align*}

Note that $\Pr_{\bm{t} \sim D}[\exists i : {V_i(t_i)} \geq 2r] \leq \frac{1}{2}$.
This is true because if we set the price of every item at $2r$,
then if any item is bought with probability greater than $\frac{1}{2}$,
we have revenue greater than $r$,
a contradiction.
Moreover $\Pr_{\bm{t} \sim D}[\max_{j} \{{V_j(t_j)}\} < 2r] = 1 - \Pr_{\bm{t} \sim D}[\exists i: {V_i(t_i)} \geq 2r] \geq \frac{1}{2}$.
By these observations,
we can conclude that:

\begin{align*}
\frac{1}{2}\geq &\Pr_{\bm{t} \sim D}[\exists i : {V_i(t_i)} \geq 2r]  \\
    \geq & \exp(-4\Delta(\bm{t}))  \Pr_{t \sim D}[ \max_{j}\{{V_j(t_j)}\} < 2r] \sum_{i \in [n]} \Pr_{t_i \sim D_i}[{V_i(t_i)} \geq 2r] \\
    \geq & \exp(-4\Delta(\bm{t}))  \frac{1}{2} \sum_{i \in [n]} \Pr_{t_i \sim D_i}[{V_i(t_i)} \geq 2r], \end{align*}
which implies that    $\sum_{i \in [n]} \Pr_{t_i \sim D_i}[V_i(t_i) \geq 2r] \leq \exp(4\Delta(\bm{t}))$.

\end{proof}

Now we are going to bound the term $\tail$.

For any fixed $t_i \in T_i$,
using Lemma~\ref{lemma:tv_MRF} on sets $\EE = \{ t_i \}$ and $\EE' = \{\vT_{-i} \in \bm{T}_{-i} : \exists j \neq i,  V_j(t_j') \geq V_i(t_i')\}$ we have that:

\begin{align*}
    \tail = &\sum_{i \in [n]} \sum_{\substack{t_i \in T_i \\ {V_i(t_i)} \geq 2r}} f(t_{i}) \cdot  V_i(t_i) \Pr_{\bm{t}'\sim D}\left[\bm{t}'\notin R_i \mid t_i'=t_i \right] \\
    \leq & \sum_{i \in [n]} \sum_{\substack{t_i \in T_i \\ {V_i(t_i)} \geq 2r}} f(t_{i}) \cdot  V_i(t_i) 
    \Pr_{\vT'\sim D}[ \exists j \neq i : V_j(t_j') \geq V_i(t_i) \mid t_i'= t_i  ] \\
    \leq & \sum_{i \in [n]} \sum_{\substack{t_i \in T_i \\ {V_i(t_i)} \geq 2r}} f(t_{i}) \cdot  V_i(t_i) \exp(4\Delta(\bm{t}))
    \Pr_{\vT'\sim D}[ \exists j \neq i : V_j(t_j') \geq V_i(t_i)  ] 
\end{align*}

We consider the mechanism that posts price $V_i(t_i)$ at each item except item $i$,
and allows the buyer to get her favorite item.
The expected revenue of this mechanisms is exactly $V_i(t_i) \Pr[ \exists j \neq i : V_j(t_j) \geq V_i(t_i) ]$,
which is at most $r$.
This implies that:

\begin{align*}
& \exp(4\Delta(\bm{t})) \sum_{i \in [n]} \sum_{\substack{t_i \in T_i: \\ V_i(t_i) \geq 2r}} f(t_i) V_i(t_i) \Pr_{\vT_{-i} \sim D_{-i}}\left[ \exists j \neq i : V_j(t_j) \geq V_i(t_i) \right]\\
 \leq &
 \exp(4\Delta(\bm{t}))\sum_{i \in [n]} \sum_{\substack{t_i \in T_i: \\ V_i(t_i) \geq 2r}} f_i(t_i) r \\
= &
 \exp(4\Delta(\bm{t})) r \sum_{i \in [n]} \Pr_{t_i\sim D_i}[V_i(t_i) \geq 2r]   \\
 \leq &
 \exp(4\Delta(\bm{t})) \exp(4\Delta(\bm{t})) r \\
 = &
 \exp(8\Delta(\bm{t}))\cdot r
\end{align*}

Where the last inequality follows from Lemma~\ref{lemma:XOS_bound_prob}.

\end{prevproof}

\begin{prevproof}{Lemma}{lemma:self_bounding_MRF}
Define $g_i(\vT_{-i}) = v(\vT_{-i},C(\vT_{-i}))$, where $C(\vT_{-i})=\{ j : V_j(\bm{t_{-i}}) < 2 \srev\}$. When $i\notin C(\vT)$, $g(\vT)-g_i(\vT_{-i})= 0$. When $i\in C(\vT)$, $C(\vT)=C(\vT_{-i})\cup \{i\}$. Additionally, $g(\vT)-g_i(\vT_{-i})\geq 0$ and $g(\vT)-g_i(\vT_{-i})\leq V_i(\vT)\leq 2\srev$.  Since $v(\cdot,\cdot)$ is a XOS function, there exists non-negative numbers $\{x_\ell\}_{\ell \in C(\vT)}$ such that $g(\vT) = \sum_{\ell\in C(\vT)} x_\ell$ and $g_i(\vT_{-i}) \geq \sum_{\ell\in C(\vT_{-i})} x_\ell$. Therefore, $\sum_{i\in[n]} \left(g(\vT)-g_i(\vT_{-i})\right)\leq \{x_\ell\}_{\ell \in C(\vT)}= g(\vT)$.

Combining Lemma~\ref{lem:Poincare on MRF},~\ref{lemma:sum_bound}, and the fact that $g(\cdot)$ is $2\srev$-self-bounding, we derive the stated upper bound of the variance of $g(\vT)$.




\end{prevproof}
\section{Missing Details of the Revenue Benchmark for a XOS Buyer}\label{sec:appx_XOS_duality}
Similar to \cite{CaiZ17},
we are going to apply the duality framework on a relaxed version of the valuation function.

\begin{definition}[Relaxed Valuation (Definition~5 from \cite{CaiZ16_arxiv})]
We define the relaxed subadditive valuation $v^r(\bm{t},S)$ the following way:

\begin{align*}
    v^r(\bm{t},S) =
    \begin{cases}
    v(\bm{t},S \backslash  \{i\} ) + V_{i}({t}_i) ~~~&\text{if $\vT \in R_i$ and $i \in S$} \\
    v(\bm{t},S)& \text{Otherwise}
    \end{cases}
\end{align*}

\end{definition}

The reason that we consider the relaxed valuation function
is because that $v^r$ is ``additive'' across the favorite item and the rest of the items, and this ``additivity'' plays a crucial role in obtaining an analyzable dual. Due to the non-monotonicity of the optimal revenue in multi-item auctions, it is not clear that the optimal revenue w.r.t. the relaxed valuation is higher than the original optimal revenue. The following Lemma shows that the optimal revenue under $v^r$ is not too much smaller than the original optimal revenue, so it suffices to apply the Cai-Devanur-Weinberg duality~\cite{CaiDW16,CaiZ17} on the relaxed valuation $v^r$.

\begin{lemma} [Lemma~2 from \cite{CaiZ16_arxiv}]\label{lem:relaxed_v}
We define by $\sigma_S(\bm{t})$ the probability that the buyer with type $\bm{t}$ receives exactly the set $S$ in Mechanism $M$.
Then:

\begin{align*}
    \rev(M,v,D) \leq 2 \rev(v^r, D) + 2 \sum_{\bm{t} \in \bm{T}} \sum_{S \subseteq 2^{[n]}} f(\bm{t}) \sigma_S(\bm{t}) \left(v^r(\bm{t},S)- v(\bm{t},S)\right),
\end{align*}
where $\rev(v^r,D)$ is the optimal revenue under the relaxed valuation $v^r$
\end{lemma}
Now we are going to show how to upper bound the $\rev(v^r,D)$ with terms similar to the ones we studied in the cases where the valuation function was additive.

\begin{lemma}[Theorem~1 and Lemma~33 from \cite{CaiZ16_arxiv}]
\label{lem:bound_vir_xos}
For a Mechanism $M=(\sigma_S,p)$ and a flow $\bm{\lambda}: T\times T \rightarrow \mathbb{R}$ that satisfied the partial specification (See Figure~3 from \cite{CaiZ16_arxiv}),
we have that:

\begin{align*}
    \rev(M,v^r,D) \leq \sum_{\bm{t} \in T} f(\bm{t}) \sum_{S \subseteq 2^{[n]}} \sigma_S(\bm{t}) \Phi^r(\bm{t},S)
\end{align*}

Where $\Phi^r(\cdot, \cdot): T \times 2^{[n]} \rightarrow \mathbb{R}$ is the virtual valuation function, defined as:

\begin{align*}
    \Phi^r(\bm{t},S) = \begin{cases}
    v(\vT,S \backslash \{i\}) + V_i(t_i)- \frac{1}{f(\bm{t})}\sum_{\bm{t}' \in \bm{T}}\lambda (\bm{t}',\bm{t}) \left( V_i(t_i')-V_i(t_i) \right)~~~~~&\text{if $\vT \in R_i$ and $i\in S$} \\
    v(\bm{t},S) &\text{Otherwise}
    \end{cases}
\end{align*}

For $\bm{t} \in R_i$,
we set $\Psi^r_i (\bm{t}) = V_i(t_i)- \frac{1}{f(\bm{t})}\sum_{\bm{t}' \in \bm{T}}\lambda (\bm{t}',\bm{t}) \left( V_i(t_i')-V_i(t_i) \right) $.
So we have that:

\begin{align*}
    \Phi^r(\bm{t},S) \leq \begin{cases}
    v_i(t_i,S \backslash \{i\}) + \Psi^r_i(\bm{t})~~~~~&\text{if $\vT \in R_i$ and $i\in S$} \\
    v(\bm{t},S) &\text{Otherwise}
    \end{cases}
\end{align*}

\end{lemma}

The following lemma provides a way to set a flow that satisfies the partial specifications (See Figure~3 from \cite{CaiZ16_arxiv}).

\begin{lemma} [Adapted Claim~1 from \cite{CaiZ16_arxiv}]
\label{lem:flow_xos}
There exists a flow that satisfies the partial specifications (See Figure~3 by \cite{CaiZ16_arxiv}) such that:

\begin{align*}
    \Psi^r_i(\bm{t}) \leq \phi_i(V_i(t_i) \mid \bm{t}_{-i})
\end{align*}

Where by $\phi_i(V_i(\bm{t}) \mid \bm{t}_{-i})$ we denote the ironed virtual value of  $V_i(t_i)$,
when $t_i$ is sampled from $D_{i \mid \bm{t}_{-i}}$.
\end{lemma}

\begin{proof}
First we are going to describe how to set a flow that satisfies the partial specification requirements such that for $\vT \in R_i$ it holds that $\Psi_i^r(\vT) \leq \phi_i^N(V_i(t_i) \mid \vT_{-i})$,
where by $\phi_i^N(V_i(t_i) \mid \vT_{-i})$ we denote the non-ironed virtual value of $V_i(t_i)$,
when $t_i$ is sampled from $D_{i \mid \bm{t}_{-i}}$.
Then the way to set a flow that satisfies the partial specifications such that for $\vT \in R_i$ it holds that $\Psi_i^r(\vT) \leq \phi_i(V_i(t_i) \mid \vT_{-i})$ is similar to the ironing procedure of Section~4 by \cite{CaiDW16}.

For any two types $\vT,\vT'$,
$\lambda(\vT',\vT)> 0$ only if there exists $i \in[n]$ such that $\vT,\vT' \in R_i$,
$\vT_{-i}=\vT_{-i}$ and $t_i= \argmax \{\hat{t}_i \in T_i: V_i(t_i') > V_i(\hat{t}_i) \}$.
Let $\vT' \in R_i$,
and  $V = \max \{V_i(t_i): V_i(t_i')>V_i(t_i) \}$,
we define $D(\vT') =\{\vT\in T : V_i(t_i) = V \land t_{-i} = t_{-i}'\}$.
Note that $\lambda(\vT',\vT)>0$ only if $\vT \in D(\vT')$.
For any $\vT' \in R_i$ and $\vT \in D(\vT')\cap R_i$,
we set $\lambda(\vT',\vT)$ to be equal to $\frac{f(\bm{t})}{\Pr_{t_i' \sim D_i}[V_i(t_i')=V_i(t_i) \land \vT_{-i}'= \bm{t}_{-i}]}$ fraction of the total in flow at $\vT'$.
We note that for any type $\vT' \in T$, 
the sum of fractions of flows that it pushes to other types is at most one:
\begin{align*}
\sum_{\vT \in D(\vT')} \frac{f(\bm{t})}{\Pr_{t_i' \sim D_i}[V_i(t_i')=V_i(t_i) \land \vT_{-i}'= \bm{t}_{-i}]} = 1
\end{align*}

Therefore if $\sum_{\vT\in T}\lambda(\vT',\vT) < 1$,
we can dump any remaining flow in the sink.
It is clear that this flow satisfies the partial specifications.

Moreover for any $\vT\in R_i$,
the total in flow of $\vT$ is:

\begin{align*}
    \sum_{\vT' \in T} \lambda(\bm{t}',\bm{t})=\frac{f(\bm{t})}{\Pr_{t' \sim D}[V_i(t_i')=V_i(t_i) \land \vT_{-i}'= \bm{t}_{-i}]} \Pr_{t' \sim D}[V_i(t_i')>V_i(t_i) \land \bm{t}_{-i}' = \bm{t}_{-i}]
\end{align*}

Therefore we have that:

\begin{align*}
    \Psi^r_i(\bm{t}) = & V_i(\bm{t}) - \frac{1}{f(\bm{t})}\frac{f(\bm{t})}{\Pr_{t' \sim D}[\{V_i(t_i')=V_i(t_i) \land \bm{t}_{-i}' = \bm{t}_{-i}]} \Pr_{t' \sim D}[{V_i(t_i')>V_i(t_i)} \land \bm{t}_{-i}' = \bm{t}_{-i}] \left( {V_i(t_i')-V_i(t_i)} \right) \\
    = & V_i(\bm{t}) - \frac{\Pr_{t' \sim D}[{V_i(t_i')>V_i(t_i)} \land \bm{t}_{-i}' = \bm{t}_{-i}]}{\Pr_{t' \sim D}[{V_i(t_i')=V_i(t_i)} \land \bm{t}_{-i}' = \bm{t}_{-i}]}  \left( {V_i(t_i')-V_i(t_i)} \right) \\
    = & V_i(\bm{t}) - \frac{\Pr_{t' \sim D}[{V_i(t_i')>V_i(t_i)} \mid \bm{t}_{-i}' = \bm{t}_{-i}]}{\Pr_{t' \sim D}[{V_i(t_i')=V_i(t_i)} \mid \bm{t}_{-i}' = \bm{t}_{-i}]}  \left( {V_i(t_i')-V_i(t_i)} \right)
\end{align*}

Therefore, $\Psi_i^r(\bm{t})$ is equal to the non-ironed virtual valuation of {$V_i(t_i)$},
when $t_i$ is sampled from $D_{i\mid t_{-i}}$.

\end{proof}

Combining Lemma~\ref{lem:relaxed_v}, Lemma~\ref{lem:bound_vir_xos} and Lemma~\ref{lem:flow_xos} we get the following lemma.

\begin{lemma}[Adapted version of Theorem~2 from \cite{CaiZ17}]\label{lemma:benchmark_XOS_first_step}

	\begin{align*}
\rev(M,v, D) \leq &
2\sum_{\bm{t}\in T}f(\bm{t}) \sum_{i\in [n]} \pi_i(\bm{t})\phi(V_i(t_i) \mid \bm{t}_{-i})\mathds{1}[t \in R_i] \quad (\single)\\
&\qquad\qquad + 4\sum_{\bm{t}\in T}f(\bm{t}) \sum_{i\in [n]} v(\vT,[n] \backslash \{i\})\mathds{1}[\vT \in R_i] (\nf)
    \end{align*} \\

\end{lemma}

We further decompose the $(\nf)$ term the following way:


We note that in the case where the buyer is XOS,
we chose $2\bm{SRev}$ as the value that separates the $(\core)$ and the $(\tail)$ term.
We sum up in the following lemma.

\begin{prevproof}{Lemma}{lemma:benchmark_XOS}
\begin{align*}
(\nf) \leq &  \sum_{\bm{t}\in T}f(\bm{t}) \sum_{i\in [n]} v(t,[n] / i)\mathds{1}[t \in R_i] \\
\leq  & \sum_{\bm{t}\in T}f(\bm{t}) \left( v(t,C(\bm{t})) + \sum_{i \in [n]} V_i(t_i) \mathds{1}[V_i(t_i) \geq 2r \land t \not\in R_i] \right) \\
 \leq  & \sum_{\bm{t}\in T}f(\bm{t}) \cdot v(t,C(\bm{t}))
 \\
 & + \sum_{i \in [n]} \sum_{\substack{t_i \in T_i \\ {V_i(t_i)} \geq 2r}} \sum_{\bm{t_{-i}}\in T_{-i}} f((t_i,\bm{t}_{-i}))  \cdot  V_i(t_i) \mathds{1}[(t_i,\bm{t}_{-i}) \not\in R_i]  \\
 \leq  & \sum_{\bm{t}\in T}f(\bm{t}) \cdot v(t,C(\bm{t})) 
 \\
 & + \sum_{i \in [n]} \sum_{\substack{t_i \in T_i \\ {V_i(t_i)} \geq 2r}} \sum_{\bm{t_{-i}}\in T_{-i}} f(\bm{t}_{i}) f(\bm{t}_{-i} \mid t_i)  \cdot  V_i(t_i) \mathds{1}[(t_i,\bm{t}_{-i}) \not\in R_i]  \\
  \leq  & \sum_{\bm{t}\in T}f(\bm{t}) \cdot  v(t,C(\bm{t}))
 \\
 & + \sum_{i \in [n]} \sum_{\substack{t_i \in T_i \\ {V_i(t_i)} \geq 2r}} f(t_{i}) \cdot  V_i(t_i) \sum_{\bm{t_{-i}}\in T_{-i}}  f(\bm{t}_{-i} \mid t_i)  \mathds{1}[(t_i,\bm{t}_{-i}) \not\in R_i]  \\
  \leq  & \sum_{\bm{t}\in T}f(\bm{t}) \cdot v(t,C(\bm{t})) (\core)
 \\
 & + \sum_{i \in [n]} \sum_{\substack{t_i \in T_i \\ {V_i(t_i)} \geq 2r}} f(t_{i}) \cdot  V_i(t_i) \Pr_{\bm{t}'\sim D}\left[\bm{t}'\notin R_i \mid t_i'=t_i \right](\tail)
\end{align*}
The claim follows from the inequality above and Lemma~\ref{lemma:benchmark_XOS_first_step}.
\end{prevproof}
\section{Lower Bound: Polynomial Dependence on $\Delta$}\label{sec:LB_Delta_poly}

In this section,
we prove that for sufficiently large values $m\in \mathbb{N}$,
there exists an  type distribution represented by a MRF with maximum weighted degree $O(m)$,
such that the optimal revenue is at least $\Omega(m^{1/7})$ times the maximum revenue achieved by simple mechanisms.

To prove this statement,
we first modify the construction of Hart and Nisan \cite{HARTN_2019},
where they prove the following Theorem.
We present a high-level idea of the proofs in this section.

\begin{lemma}[Theorem~C from \cite{HARTN_2019}]
There exists a two item correlated distribution $D$ and a constant $c>0$,
such that for any $m \in \mathbb{N}$,
when a buyer with additive valuation is sampled from $D$,  $\frac{\rev(D)}{\brev(D)}\geq c \cdot m^{1/7}$.
\end{lemma}

Their construction,
relies on the following lemma.

\begin{lemma}[Proposition 7.5. from \cite{HARTN_2019}]\label{lem:hart_nisan_point}
Let $\{g_i\}_{ 0 \leq i \leq m} \in [0,1]^n$ and $\{y_i\}_{0 \leq i \leq m} \in \mathbb{R}_+^n$ be two sequences of $m+1$ points,
such that $g_0 = (0,\ldots,0)$.
For $i \geq 1$ we define:
\begin{align*}
    gap_i := \min_{0 \leq j < i} (g_i - g_j)\cdot y_i
\end{align*}


For any $m\in \mathbb{N}$,
there exists a sequence $\{ g_i \}_{0 \leq i\leq m}$ in $[0,1]^2$ such that $g_0=(0,0)$ and for each $1 \leq i \leq m$, $||g_i||_2 \leq 1$.
Moreover, if we set $y_i=g_i$ for all $ 0 \leq i \leq m$,
then $gap_i = \Omega\left(i^{-6/7}\right)$.
\end{lemma}

Their construction is developed inductively,
by placing points on ``shells'' of fixed radius.
More specifically, in the $N$-th shell, whose radius is $\frac{\sum_{i=1}^Ni^{-3/2}}{\sum_{i=1}^{\infty}i^{-3/2}}$, 
they place $N^{3/4}$ points so that the angle between any pair of points in the same shell is $\Omega(N^{-3/4})$.
They observed that for all points (except $g_0$),
$||g_i||_2 = \Theta(1)$,
since $\sum_{i=1}^{\infty}{i^{-3/2}}= \Theta(1)$ and $\min_{1 \leq i\leq m}||g_i||_2 = ||g_1||_2 = \Theta(1)$.
In their proof,
they only needed a lower bound for each $gap_i$,
but we also need an upper bound.
Lemma~\ref{lem:hart_nisan_mod} provides us with that bound.

\begin{lemma}\label{lem:hart_nisan_mod}
In the construction of the set of points $\{g_i\}_{0 \leq i \leq m}$ in Proposition~7.5 in \cite{HARTN_2019}, if:
\begin{itemize}
\item we place the first point of each shell in the same line that passes through the origin $(0,0)$
\item for the $N$-th shell, whose radius is $\frac{\sum_{i=1}^Ni^{-3/2}}{\sum_{i=1}^{\infty}i^{-3/2}}$, and any point $i$ in that shell (except the first point of that shell), there exists another point $j < i$ in that shell such that the angle between them is $\Theta(N^{-3/4})$,
\end{itemize}
then if we consider $\{y_i\}_{0\leq i \leq m}=\{g_i\}_{0\leq i \leq m}$,
for each point $i$ that is in the $N$-th shell,
we have that $gap_i = \Theta(N^{-3/2}) = \Theta(i^{-6/7})$.
\end{lemma}

\begin{proof}
First we note that there is no restriction that prevents us from placing the first point of each shell in the same line.
This is because different shells,
have different radius so there is no way two points coincide.
It is also trivial to ensure that the for point $i$ in the $N$-th shell,
which is not the first point in the shell,
there exists a point $j<i$ in the $N$-th shell such that the angle between them is $\Theta(N^{-3/4})$.
We note that the only assumption about the set of points that was made in the proof of Proposition~7.5 in \cite{HARTN_2019} (Lemma~\ref{lem:hart_nisan_point}),
was that for the $N$-th shell, the points are placed in a semicircle of the radius we described above and that the angle between any pair of points is $\Omega(N^{-3/4})$.
Therefore the result of Lemma~\ref{lem:hart_nisan_point} applies to this point set too.

{
Now we are going to prove that for the first point $i^*$ in the $N$-th shell,
it holds that $gap_{i^*} = \Theta(N^{-3/2})$.}
Let $j^*$ be the first point of the $N-1$ shell and $i^*$ be the first point of the $N$-th shell,
then $gap_{i^*} \leq (g_{i^*} -g_{j^*})\cdot g_{i^*}=(||g_{i^*}||_2 - ||g_{j^*}||_2)||g_{i^*}||_2 =  \frac{N^{-3/2}}{\sum_{i=1}^{\infty}i^{-3/2}}||g_{i^*}||_2 = O(N^{-3/2})$,
where the first equality holds because point $i^*$ and point $j^*$ lie in the same line that passes through the origin.
Since the results of Lemma~\ref{lem:hart_nisan_point} holds here,
we have that $gap_{i^*} = \Theta(N^{-3/2})$. 


Next,
we deal with the case where point $i$ is not the first point of the $N$-th shell. In the proof of Proposition~7.5 in \cite{HARTN_2019},
they noted that for two points that have angle $\theta$ between them, it holds that $\cos(\theta) = 1 - \Omega(\theta^2)$.
We note that for $\theta < \frac{\pi}{2}$, 
we can similarly prove that $\cos (\theta)= 1 - \Theta(\theta^2)$ using the Taylor expansion of $\cos (\theta)$.
For any point $j$ that is not the first point in the shell, there is another point $i<j$ such that the two point have an angle $\theta=\Theta(N^{-3/4})$.
Since $||g_j||_2 =\Theta(1)$,
we can conclude that $gap_i \leq (g_i-g_j)\cdot g_i = ||g_i||_2^2 - ||g_j||_2||g_i||_2\cos(\theta) = ||g_i||_2^2 \Theta(N^{-3/2}) = \Theta(N^{-3/2})$. Finally, since the $N$-th shell contains $N^{3/4}$ points, if point $i$ belongs to the $N$-th shell, then $i=\Theta(N^{7/4})$. Hence, $gap_i  = \Theta(i^{-6/7})$.
\end{proof}

\begin{lemma}[Modified version of Theorem~C from \cite{HARTN_2019}]\label{lem:hart_nisan_mod_result}
For any sufficiently large $m \in \mathbb{N}$,
there exists a two item correlated distribution $D$ with support $T= \supp(D)$ and an absolute constant $C >2$ such that $\inf_{\vT \in \supp(D)}f(\vT)\geq C^{-m}$.
Moreover, when a buyer with additive valuation is sampled from $D$,
then there exists another absolute constant $c>0$ such that $\frac{\rev(D)}{\brev(D)}\geq c\cdot m^{1/7}$.
\end{lemma}

\begin{proof}

Given any sequences $\{g_i\}_{0 \leq i\leq m}$ and $\{y_i\}_{0 \leq m}$ and a target value $\varepsilon$,
Proposition~7.1 in \cite{HARTN_2019} constructs the following distribution $D$: 
(i) Construct a sequence of positive numbers $\{t_i\}_{1\leq i\leq m}$ that increases fast enough, so that (a) $\xi_i:=||x_i||_1$ is increasing, where $x_i:=\frac{t_i y_i}{gap_i}$ and (b) $\frac{t_{i+1}}{t_i}\geq \frac{1}{\varepsilon}$; (ii) The buyer has type $x_i$ with probability $\frac{\xi_1}{\xi_i}- \frac{\xi_1}{\xi_{i+1}}$. Proposition~7.1 shows that for any choice of $\{t_i\}_{1\leq i\leq m}$ that satisfies property (a) and (b) in step (1) of the construction, the corresponding distribution $D$ has $\frac{\rev(D)}{\brev(D)} \geq (1-\varepsilon) \sum_{i=1}^m\frac{gap_i}{||y_i||_1}$.


If we choose $\{g_i\}_{1\leq i\leq m}$ and $y_i=g_i $ for all $i\in[m]$ as in Lemma~\ref{lem:hart_nisan_mod} and $\{t_i\}_{1\leq i\leq m}$ that satisfies property $(a)$ and $(b)$, it is not hard to verify that $\sum_{i=1}^m\frac{gap_i}{||y_i||_1}=\Omega(m^{1/7})$. Hence, the gap between $\rev(D)$ and $\brev(D)$ is as stated in the claim. The problem with this construction is that we cannot lower bound the probability that the rarest type shows up, as $\xi_i$ and $\xi_{i+1}$ can be very close to each other. To fix this issue, we modify the construction by replacing property (a) with a strengthened property (a*)~$\frac{\xi_i}{\xi_{i-1}}\in [2,C]$ for all $i>1$, where $C$ is an absolute constant that will be determined later.  

We first argue that if (a*) is satisfied, then the rarest type shows with sufficiently large probability. More specifically, type $x_i$ shows up with probability $$\frac{\xi_1}{\xi_i}- \frac{\xi_1}{\xi_{i+1}}\geq \frac{\xi_1}{\xi_{i+1}}\geq C^{-i}.$$ 

Next, we argue that for the sequences $\{g_i\}_{1\leq i\leq m}$ and $\{y_i\}_{1\leq i\leq m}$ as described in Lemma~\ref{lem:hart_nisan_mod}, there exists a sequence $\{t_i\}_{1\leq i\leq m}$  that satisfies (a*) and (b). Note that 

\begin{align*}
    \frac{\xi_i}{\xi_{i-1}} = \frac{t_i}{t_{i-1}} \frac{||g_i||_1}{||g_{i-1}||_1}\frac{gap_{i-1}}{gap_i}
\end{align*}

By the definition of $\{g_i\}_{1\leq i\leq m}$, each point $g_i$ is placed in a shell of radius at least $\frac{1}{\sum_{i=1}^\infty i^{-3/2}}$, so $||g_i||_2=\Theta(1)$ and $\frac{||g_i||_2}{||g_{i-1}||_2}=\Theta(1)$. 
Since $g_i\in [0,1]^2$, $||g_i||_1=\Theta(||g_i||_2)$, which implies that $\frac{||g_i||_1}{||g_{i-1}||_1}=\Theta(1)$. According to Lemma~\ref{lem:hart_nisan_mod}, $gap_i=\Theta(N^{-3/2})$ if $i$ belongs to the $N$-th shell, so $\frac{gap_{i-1}}{gap_i}=\Theta(1)$. Hence, there exists two positive absolute constants $C_1$ and $C_2$ such that $C_1\cdot \frac{t_i}{t_{i-1}}\leq \frac{\xi_i}{\xi_{i-1}} \leq C_2\cdot \frac{t_i}{t_{i-1}}$. For the rest of the proof, we take $\varepsilon$ to be $1/2$. If we choose $\{t_i\}_{1\leq i\leq m}$ such that $\frac{t_i}{t_{i-1}}$ to be $\max\{2/C_1, 2\}$ for all $i>1$, $\frac{\xi_i}{\xi_{i-1}}\in [2,C]$ for some absolute constant $C$. As the construction above satisfies both property (a*) and (b), we have $$\frac{\rev(D)}{\brev(D)} \geq \frac{1}{2} \sum_{i=1}^m\frac{gap_i}{||y_i||_1}=\Omega (m^{1/7})$$ for the induced distribution $D$.

\notshow{
By carefully going over the proof,
we can verify that also setting the values of $t_i$ to also satisfy that $\xi_i\geq 2 \xi_{i-1}$,
does not affect the proof.
From now on,
we also make the assumption that $\xi_i \geq 2\xi_{i-1}$ \argyrisnote{ and that we set $t_i$ to the minimum value such that the two constraints $\frac{t_i}{t_{i-1}}\geq \frac{1}{\epsilon}$ and that $\xi_i\geq 2 \xi_{i-1}$,
are both satisfied.}
Now we are going to prove that if we use the sequence $\{g_i\}_{0 \leq i \leq m}$ as described in Lemma~\ref{lem:hart_nisan_mod} and $\{y_i\}_{0 \leq i \leq m} = \{ g_i \}_{0 \leq i \leq m}$,
then there exists a constant $C>0$ such that $\xi_i \leq C \xi_{i-1}$.
Let point $g_i$ belong in  the $N$-th shell:

\begin{align*}
    \frac{\xi_i}{\xi_{i-1}} = \frac{t_i}{t_{i-1}} \frac{||g_i||_1}{||g_{i-1}||_1}\frac{gap_{i-1}}{gap_i}
\end{align*}

By Lemma~\ref{lem:hart_nisan_mod},
if point $j$ belongs in the $M$-th shell,
then $gap_j = \Theta(M^{-3/2})$.
Since $i$ belongs in the $N$-th shell,
then $i-1$ either belongs in the $N$-th shell, or in the $N-1$ shell.
In both cases $\frac{gap_{i-1}}{gap_i} = \Theta(1)$.
As we have noted $||g_i||_2 = \theta(1)$ and $||g_{i-1}||_2 = \theta(1)$,
therefore $\frac{||g_i||_2}{||g_{i-1}||_2} = \Theta(1)$.

Therefore there exists an absolute constant $c'$ such that for all $2 \leq j \leq m$, $\frac{\xi_j}{\xi_{j-1}} \geq c'\cdot\frac{t_j}{t_{j-1}}$.
\argyrisnote{
We remind that the parameters $t_i$ are set to their minimum value so that the two constraints described in the paragraph above are satisfied.
}
As long as $\frac{t_i}{t_{i-1}}\geq \frac{1}{c'}$,
then sequence $\xi_i\geq \xi_{i-1}$.
By setting $\epsilon=1/2$,
the other constraint becomes $\frac{t_i}{t_{i-1}}\geq \frac{1}{\epsilon}=2$.
Thus $\frac{t_i}{t_{i-1}}\leq \max\left(\frac{1}{c'},2\right) = \Theta(1)$.
Therefore we can conclude that there exists a constant $C>0$ such that $\xi_i \leq C \xi_{i-1}$.

By choosing the value of $t_1$ so that $\xi_1 = 1$,
the probability that agent has type $x_i$ becomes $\frac{1}{\xi_i} - \frac{1}{\xi_{i+1}}$.
Since we assumed that $\xi_{i+1} \geq 2\xi_i$,
then $\frac{1}{\xi_i} - \frac{1}{\xi_{i+1}}\geq \frac{1}{2\xi_i}$.
Since $\frac{\xi_i}{\xi_{i-1}} \leq C$,
then $\frac{1}{\xi_i} \geq \frac{1}{C^i}$,
which concludes the proof.
\argyrisnote{By combining the previous two inequalities,
we can infer that the probability that the agent has any specific type $x$ is at least $\frac{1}{2C^m}$.}}
\end{proof}

\notshow{
\argyrisnote{
Argyris: proof by bounding the highest value-------------
I have not finished it yet and it seems to me more complicated than the previous proof---------
\begin{proof}
By reading the proof of Theorem~C in \cite{HARTN_2019},
we can see that the constructed distribution is derived by Proposition~7.1 in \cite{HARTN_2019} with the sequence $\{g_i\}_{0 \leq i\leq m}$ like the one in Proposition~7.1 in \cite{HARTN_2019}.

Given the sequences $\{g_i\}_{0 \leq i\leq m}$ and $\{y_i\}_{0 \leq m}$ and a value $\epsilon$,
Proposition~7.1 constructs a distribution $D$,
such that $\frac{\rev(D)}{\brev(D)} \geq (1-\epsilon) \sum_{i=1}^m\frac{gap_i}{||y_i||_1}$.
Let $x_i = \frac{t_i y_i}{gap_i}$ and $\xi_i = ||x_i||_1$ and $\xi_{m+1}=+\infty$,
where $t_i$ are some parameters defined so that the sequence $\xi_i$ is increasing and $\frac{t_i}{t_{i-1}}\geq\frac{1}{\epsilon}$. 
In the constructed distribution,
the agent has type $x_i$ with probability $\frac{\xi_1}{\xi_i}- \frac{\xi_1}{\xi_{i+1}}$.

Now we are going to prove that if we use the sequence $\{g_i\}_{0 \leq i \leq m}$ as described in Lemma~\ref{lem:hart_nisan_mod} and $\{y_i\}_{0 \leq i \leq m} = \{ g_i \}_{0 \leq i \leq m}$,
then there exists a constant $C>0$ such that $\xi_i \leq C \xi_{i-1}$.
Let point $g_i$ belong in  the $N$-th shell:

\begin{align*}
    \frac{\xi_i}{\xi_{i-1}} = \frac{t_i}{t_{i-1}} \frac{||g_i||_1}{||g_{i-1}||_2}\frac{gap_{i-1}}{gap_i}
\end{align*}

By Lemma~\ref{lem:hart_nisan_mod},
if point $j$ belongs in the $M$-th shell,
then $gap_j = \Theta(M^{-3/2})$.
Since $i$ belongs in the $N$-th shell,
then $i-1$ either belongs in the $N$-th shell, or in the $N-1$ shell.
In both cases $\frac{gap_{i-1}}{gap_i} = \Theta(1)$.
As we have noted $||g_i||_2 = \theta(1)$ and $||g_{i-1}||_2 = \theta(1)$,
therefore $\frac{||g_i||_2}{||g_{i-1}||_2} = \Theta(1)$.

Therefore there exists an absolute constant $c'$ such that for all $2 \leq j \leq m$, $\frac{\xi_j}{\xi_{j-1}} \geq c'\cdot\frac{t_j}{t_{j-1}}$.
Therefore as long as $\frac{t_i}{t_{i-1}}\geq \frac{1}{c'}$,
then $\xi_i\geq \xi_{i-1}$.
By setting $\epsilon=1/2$,
in order for the sequence $t_i$ to also satisfy that $\frac{t_i}{t_{i-1}}\geq \frac{1}{\epsilon}=2$,
it has to be $\frac{t_i}{t_{i-1}}\geq \max\left(\frac{1}{c'},2\right) = \Theta(1)$.
Therefore we can conclude that there exists a constant $C>0$ such that $\xi_i \leq C \xi_{i-1}$.

By choosing the value of $t_1$ so that $\xi_1 = 1$,
then $\xi_i \leq C^i$.
Note that the highest value for the bundle that the agent has is $\max_{0 \leq i\leq m}||x_i||_1=\max_{0 \leq i \leq m} \xi_i = \xi_m$,
since the sequence $\{\xi_i\}_{1 \leq i \leq m}$ is increasing.
Now we consider the distribution $D'$ that has type $x_i$ with with probability $\Pr_{x \sim D}[x=x_i]$ only if $\Pr_{x \sim D}[x=x_i] \geq \frac{1}{2 \xi_m m}$.
Therefore a lower bound on $\inf_{\vT \in \supp(D')} \Pr_{\vT' \sim D'}[\vT' = \vT]$ is $\frac{1}{2 \xi_m m} \geq \frac{1}{2 C^m m}$.
Note that $\rev(D') \geq \rev(D) - \frac{1}{2}$.
This statement is true if we consider the mechanism that achieves the optimal revenue on $D$ and apply it on $D'$ and by noting that we lose at most $\frac{1}{2}$ revenue by the types that are included in $D$ but not $D'$.
Also $\brev(D') \leq \brev(D)$.
This implies the following:
\begin{align*}
    \frac{\rev(D')}{\brev(D')} \geq \frac{\rev(D) - 1/2}{\brev(D)} \geq ....
\end{align*}

\end{proof}

}
}

\begin{prevproof}{Theorem}{thm:poly_dependence_Delta}
By Lemma~\ref{lem:hart_nisan_mod_result},
there exists a type distribution $D$
and constants $c>0,c'>0$ such that $\inf_{\vT \in \supp(D)}\Pr_{\vT' \sim D}[\vT' = \vT] \geq c^m$ and when a buyer has an additive valuation sampled from $D$,
then $\frac{\rev(D)}{\brev(D)}\geq c'\cdot m^{1/7}$.

Using Lemma~\ref{lemma:mix_dobrushin} on $D$ with parameters $\alpha' = \frac{1}{2}$,
we get a MRF $D'$,
such that
its maximum weighted degree is bounded by $\Delta \leq \left|\log\left( \frac{c^{2m}}{2}\right) \right|= |2m \cdot \log(c) -\log(2)| = O (m)$
for sufficiently large $m>0$.
Moreover $\rev(D') \geq \frac{1}{2}\rev(D)$.

Since the marginal distributions of $D$ and $D'$ are the same,
we have that $\srev(D') = \srev(D)$.
We now prove that $\srev(D) \leq 2 \brev(D)$.
Let $rev_1$ be the optimal revenue when we only sell the first item,
and $rev_2$ to be the optimal revenue when we only sell the second item.
We can easily see that $\brev \geq \max(rev_1,rev_2)$.
Since $\srev = rev_1 + rev_2$, we can conclude that $\srev \leq 2 \brev$.

At this moment, we prove that $\brev(D) \leq 2 \srev(D)$.
Let $p^*$ be the price induced by the optimal grand bundle mechanism.
The revenue achieved by posting each item at price $p^*/2$ is a lower bound on $\srev(D)$.
Moreover, if we sell each item at price $p^*/2$,
then we are guaranteed to achieve at least half the revenue induced by $\brev(D)$ and our claim holds.

Thus $\frac{\rev(D')}{\max\{\brev(D'),\srev(D')\}} \geq \frac{1}{2}\frac{\rev(D')}{\srev(D')} \geq \frac{1}{4}\frac{\rev(D)}{\srev(D)} \geq \frac{1}{8}\frac{\rev(D)}{\brev(D)} \geq \frac{c\cdot m^{1/7}}{8}$,
which concludes the proof.
\end{prevproof}

\end{document}